\documentclass[runningheads]{llncs}
\usepackage{url,verbatim,multirow}
\usepackage{amsfonts,amsmath,amssymb,color,wasysym}
\usepackage{latexsym,verbatim,threeparttable,rotating}
\usepackage{hyperref}
\usepackage{array,hhline,color}
\usepackage[tableposition=below]{caption}
\usepackage{tabu}

\usepackage{caption}
\usepackage{adjustbox,multirow}

\usepackage[svgnames]{xcolor}
\usepackage{framed}

\usepackage{framed}

\definecolor{shadecolor}{named}{LightGray}

\usepackage{booktabs,graphicx}

\newtheorem{condition}[theorem]{Condition}

\usepackage[ruled]{algorithm2e}

\SetAlFnt{\small}
\SetAlCapFnt{\small}
\SetAlCapNameFnt{\small}
\SetAlCapHSkip{0pt}
\IncMargin{-\parindent}

\usepackage{enumitem}
\usepackage{latexsym}
\usepackage{mathtools}
\usepackage{multirow}
\usepackage{multicol}
\usepackage{mwe}
\usepackage{makecell}
\usepackage{pgfplots}
\usepackage[section]{placeins}
\usepackage{pifont}
\usepackage{ragged2e}
\usepackage{rotating}
\usepackage{subcaption}
\usepackage{tabu}
\usepackage{tabularx}
\usepackage{threeparttable}
\usepackage{tikz}
\usepackage[normalem]{ulem}
\usepackage{url}
\usepackage{verbatim}
\usepackage{wrapfig}
\usepackage{xcolor}
\usepackage{xspace}
\usepackage[framemethod=tikz]{mdframed}
\usepackage{algpseudocode}
\input{symbols_env}
\usepackage{verbatim,color}
\newcommand{\F}{\ensuremath{\mathbb{F}}}
\newcommand{\K}{\ensuremath{\mathbb{K}}}
\newcommand{\SharingSpec}{\ensuremath{\mathbb{S}}}

\newcommand{\Q}{\ensuremath{\mathcal{Q}}}
\newcommand{\R}{\ensuremath{\mathcal{R}}}
\newcommand{\Order}{\ensuremath{\mathcal{O}}}

\newcommand{\defined}{\ensuremath{\stackrel{def}{=}}}

\newcommand{\C}{\mathcal{C}}
\newcommand{\Z}{\mathcal{Z}}
\newcommand{\ckt}{\ensuremath{\mathsf{cir}}}

\newcommand{\pref}{\ensuremath{\mathsf{pref}}}

\newcommand{\Con}{\ensuremath{\mathsf{Con}}}

\newcommand{\PartySet}{\ensuremath{\mathcal{P}}}
\newcommand{\Partyset}{\ensuremath{\mathcal{P}}}
\newcommand{\AdvStruct}{\ensuremath{\mathcal{Z}}}
\newcommand{\AdvStructure}{\ensuremath{\mathcal{Z}}}
\newcommand{\Honest}{\ensuremath{\mathcal{H}}}

\newcommand{\E}{\ensuremath{\mathcal{E}}}

\newcommand{\CoreSet}{\ensuremath{\mathcal{CS}}}
\newcommand{\CSet}{\ensuremath{\mathcal{C}}}
\newcommand{\Kings}{\ensuremath{\mathcal{K}}}

\newcommand{\Adv}{\ensuremath{\mathcal{A}}}
\newcommand{\D}{\ensuremath{\mathsf{D}}}
\newcommand{\Sender}{\ensuremath{\mathsf{S}}}

\newcommand{\BA}{\ensuremath{\Pi_\mathsf{BA}}}
\newcommand{\VSS}{\ensuremath{\Pi_\mathsf{VSS}}}

\newcommand{\PiACast}{\ensuremath{\Pi_\mathsf{ACast}}}

\newcommand{\BCAST}{\ensuremath{\Pi_{\mathsf{BC}}}}

\newcommand{\ABA}{\ensuremath{\Pi_{\mathsf{ABA}}}}
\newcommand{\SBA}{\ensuremath{\Pi_{\mathsf{SBA}}}}

\newcommand{\Rec}{\ensuremath{\Pi_{\mathsf{Rec}}}}
\newcommand{\Mult}{\ensuremath{\Pi_{\mathsf{Mult}}}}

\newcommand{\ACS}{\ensuremath{\Pi_{\mathsf{ACS}}}}

\newcommand{\BatchBeaver}{\ensuremath{\Pi_{\mathsf{Beaver}}}}
\newcommand{\PiMPC}{\ensuremath{\Pi_{\mathsf{CirEval}}}}

\newcommand{\Offline}{\ensuremath{\Pi_{\mathsf{PreProcessing}}}}

\newcommand{\TimeBCAST}{\ensuremath{T_\mathsf{BC}}}
\newcommand{\TimeBA}{\ensuremath{T_\mathsf{BA}}}

\newcommand{\TimeABA}{\ensuremath{T_\mathsf{ABA}}}
\newcommand{\TimeSBA}{\ensuremath{T_\mathsf{SBA}}}

\newcommand{\TimeVSS}{\ensuremath{T_\mathsf{VSS}}}

\newcommand{\TimeACS}{\ensuremath{T_\mathsf{ACS}}}

\newcommand{\TimeMult}{\ensuremath{T_\mathsf{{Mult}}}}
\newcommand{\TimeOffline}{\ensuremath{T_\mathsf{{PreProcessing}}}}

\newcommand{\init}{\ensuremath{\texttt{init}}}
\newcommand{\echo}{\ensuremath{\texttt{echo}}}
\newcommand{\ready}{\ensuremath{\texttt{ready}}}
\newcommand{\Ready}{\ensuremath{\texttt{ready}}}
\newcommand{\OK}{\ensuremath{\texttt{OK}}}
\newcommand{\NOK}{\ensuremath{\texttt{NOK}}}
\newcommand{\resolve}{\ensuremath{\texttt{Resolve}}}

\newenvironment{myitemize}
{\begin{list}{$\bullet$}{ 
\itemindent=-0.1in
\itemsep=0.0in
\parsep=0.0in
\topsep=0.0in
\partopsep=0.0in}}{\end{list}}
\newcounter{itemcount}

\newenvironment{myenumerate}
{\setcounter{itemcount}{0}\begin{list}
{\arabic{itemcount}.}{\usecounter{itemcount} \itemindent=-0.2cm
\itemsep=0.0in
\parsep=0.0in
\topsep=5pt
\partopsep=0.0in}}{\end{list}}

\newenvironment{mydescription}
{\setcounter{itemcount}{0}\begin{list}
{\arabic{itemcount}.}{\usecounter{itemcount} \itemindent=-0.5cm
\itemsep=0.0in
\parsep=0.0in
\topsep=5pt
\partopsep=0.0in}}{\end{list}}

\begin{document}

\title{\bf Perfectly Secure Synchronous MPC with Asynchronous Fallback Guarantees Against General Adversaries}
\author{Ananya Appan\thanks{The work was done when the author was a student at the International Institute of Information Technology, Bangalore India.
    Email: {\tt{ananya.appan@iiitb.ac.in}}}
  \and Anirudh Chandramouli\thanks{The work was done when the author was a student at the International Institute of Information Technology, Bangalore India.
    Email: {\tt{anirudh.c@iiitb.ac.in}}}
  \and Ashish Choudhury\thanks{International Institute of Information Technology, Bangalore India.
 Email: {\tt{ashish.choudhury@iiitb.ac.in}}. }}
\institute{}
\date{}

\maketitle
\begin{abstract}
In this work, we study  {\it perfectly-secure} {\it multi-party computation} (MPC) against general ({\it non-threshold})
 adversaries. 
  Known protocols in a {\it synchronous} network are secure 
  against 
  $\Q^{(3)}$ adversary structures, while in an {\it asynchronous} network, known protocols are secure
  against  $\Q^{(4)}$ adversary structures.
   A natural question is whether there exists a {\it single} protocol which remains secure
   against $\Q^{(3)}$ and $\Q^{(4)}$ adversary structures in a {\it synchronous} and in an 
   {\it asynchronous} network respectively,
   where the parties are {\it not aware} of the network type.
    We design the {\it first} such {\it best-of-both-worlds}  protocol
    against general adversaries. 
    Our result generalizes the result of Appan, Chandramouli and Choudhury (PODC 2022),
   which presents a best-of-both-worlds perfectly-secure protocol against {\it threshold} adversaries. 
   
   To design our protocol, we present two important building blocks which are of independent interest.
   The first building block is a best-of-both-worlds perfectly-secure {\it Byzantine agreement} (BA) protocol
   for $\Q^{(3)}$ adversary structures, which remains secure {\it both} in a synchronous, as well as 
   an asynchronous network. The second building block is a {\it best-of-both-worlds} perfectly-secure 
    {\it verifiable secret-sharing} (VSS) protocol, which remains secure against 
    $\Q^{(3)}$ and $\Q^{(4)}$ adversary structures in 
    a {\it synchronous}
    network and an {\it asynchronous} network respectively. 
\end{abstract}


\section{Introduction}
\label{sec:intro}
Secure {\it multi-party computation} (MPC) \cite{Yao82,GMW87,BGW88} is one of the central pillars in modern cryptography.
 Informally, an MPC protocol 
 allows a set of mutually distrusting parties,  $\Partyset = \{P_1, \ldots, P_n \}$, to securely perform any computation over their private inputs
  without revealing any additional information about their inputs. 
  In any MPC protocol, the distrust among the parties is
   modeled by a centralized {\it adversary} $\Adv$, who can corrupt and control a subset of the parties during the protocol execution.
   We consider {\it computationally unbounded}, Byzantine (malicious) adversaries. This is the most powerful form of corruption where $\Adv$ can force the corrupt parties to behave {\it arbitrarily} during protocol
   execution. Security achieved against such an adversary is called {\it perfect security}. 
   
   Traditionally, the {\it corruption capacity} of $\Adv$ is modeled through a publicly-known {\it threshold} $t$, where it is assumed that
   $\Adv$ can corrupt {\it any} $t$ parties \cite{BGW88,CCD88,RB89}.  A more general form of corruption capacity
   is the {\it general-adversary} model (also known as the {\it non-threshold} setting) \cite{HM97}. 
    Here, $\Adv$ is characterized by a publicly-known {\it adversary structure} 
    $\AdvStructure \subset 2^{\PartySet}$, which enumerates {\it all possible}
     subsets of potentially corrupt parties, where $\Adv$ can select any subset from $\AdvStructure$ for corruption. 
     Notice that a {\it threshold} adversary is a {\it special} type of {\it non-threshold} adversary, where 
     $\Z$ consists of all subsets of $\PartySet$ of size up to $t$. 
     It is well-known that modelling $\Adv$ through $\AdvStructure$
      allows for more flexibility, especially when
    $\PartySet$ is small \cite{HM97,HM00}. 
  \paragraph{\bf Our Motivation and Results:}
  Traditionally, MPC protocols are designed assuming either a {\it synchronous} or {\it asynchronous} communication model. 
    In
    a {\it synchronous} MPC (SMPC) protocol, 
    the communication channels between the parties are assumed to be {\it synchronized},
    and every message is assumed to be delivered within some {\it known} time.
    Unfortunately, maintaining such time-outs in real-world networks like the Internet is extremely challenging. 
   Asynchronous MPC (AMPC) protocols operate assuming an {\it asynchronous} communication network, 
   where the channels are not synchronized, and messages can be arbitrarily (yet finitely) delayed. 
   Designing AMPC protocols is {\it more challenging} when compared to SMPC protocols.
   This is because, inherently, in {\it any} AMPC protocol, a receiver party {\it cannot distinguish}
   between a {\it slow sender} party (whose messages are arbitrarily delayed) and a {\it corrupt sender} party (who does not send
   any messages). Hence, to avoid an endless wait, no party can afford to receive messages from {\it all} the parties,
   as corrupt parties may never send their designated messages. 
   So, at every step in an AMPC protocol, 
   a receiver party can wait for messages from only a ``subset" of parties, ignoring messages from the remaining parties
   which may be {\it potentially honest}.
   In fact, in {\it any} AMPC protocol, it is {\it impossible}
    to ensure that the inputs of {\it all} honest parties are considered for computation, 
    since waiting for the inputs of 
    {\it all} the parties may turn out to be an endless wait. 
     
     Against {\it threshold} adversaries, perfectly-secure SMPC can tolerate 
     up to $t_s < n/3$ corrupt parties \cite{BGW88}. On the other hand, perfectly-secure AMPC
    can tolerate  up to $t_a < n/4$ corrupt parties \cite{BCG93}. 
     These impossibility results have been generalized to the following bounds against a {\it non-threshold} adversary:
     SMPC against a general-adversary is possible provided the underlying {\it synchronous adversary structure} 
     $\Z_s$ satisfies the $\Q^{(3)}$ condition \cite{HM00}.
     On the other hand, AMPC against a general-adversary is possible provided the underlying {\it asynchronous adversary structure} 
     $\Z_a$ satisfies the $\Q^{(4)}$ condition \cite{MSR02}.\footnote{An adversary structure $\Z$ satisfies
     the $\Q^{(k)}$ condition if the union of any $k$ subsets from $\Z$ {\it does not} cover $\PartySet$.}
     
     Typically, in any MPC protocol, it is {\it assumed} that the parties will be knowing whether 
     the underlying network is synchronous or asynchronous {\it beforehand}. 
      We envision a scenario where the parties are {\it not aware} of the network type, and 
      aim to design a {\it single} MPC protocol which offers the best possible security guarantees, both in the synchronous and the asynchronous communication model.
      We call such a protocol as a {\it best-of-both-worlds} MPC protocol.
               In a recent work, Appan et al.~\cite{ACC22a} presented a best-of-both-worlds {\it perfectly-secure}
         MPC protocol against {\it threshold} adversaries which could tolerate up to $t_s$ and $t_a$
         corruptions in a {\it synchronous} and {\it asynchronous} network respectively, for any $t_a < t_s$ where $t_a < n/4$
         and $t_s < n/3$,
         provided $3t_s + t_a < n$ holds. We aim to generalize this result against {\it general} adversaries, and ask the following question:         
        \begin{center}
        Let $\Adv$ be an adversary, characterized by adversary structures $\Z_s$ and $\Z_a$  
        in a 
        synchronous network and asynchronous network respectively, where $\Z_s \neq \Z_a$.      
        Then, is there a best-of-both-worlds 
        perfectly-secure MPC protocol which is secure against $\Adv$, irrespective of the 
        network type?        
        \end{center}
        No 
        prior work has addressed the above question. We present a 
        best-of-both-worlds perfectly-secure MPC protocol, provided $\Z_s, \Z_a$ satisfy the 
        $\Q^{(3, 1)}$ condition and if every subset in $\Z_a$ is a subset of some subset in 
        $\Z_s$.\footnote{$\Z_s, \Z_a$ satisfy the $\Q^{(k, k')}$ condition if the union of any $k$ and $k'$ subsets
        from $\Z_s$ and $\Z_a$ respectively {\it does not} cover $\PartySet$.} 
          Note that we focus on the case where $\Z_s \neq \Z_a$ as otherwise, the question is {\it trivial} to solve. 
          More specifically, if $\Z_s = \Z_a$, then the necessary condition of AMPC implies that 
          {\it even} $\Z_s$ satisfies the $\Q^{(4)}$ condition. Hence,           
           one can use {\it any} existing perfectly-secure AMPC protocol against general-adversaries
            (with appropriate time-outs) \cite{MSR02,CP20,ACC22b},
          which will
          work {\it even} in the synchronous network, with the guarantee that  the inputs of {\it all honest} parties are considered for the
          computation.
          Our goal, however, is to achieve security against $\Q^{(3)}$ adversary structures, if the underlying network is {\it synchronous}.
          For example, let $\PartySet = \{P_1, \ldots, P_8\}$. Consider the adversary structures
           $\Z_s = \{ \{P_1, P_2, P_3\}, \{P_2, P_3, P_4\}, \{P_3, P_4, P_5\}, \{P_4, P_5, P_6\}, \allowbreak \{P_7\}, \{P_8\}\}$
           and $\Z_a = \{ \{P_1, P_3\}, \{P_2, P_4\}, \{P_3, P_5\}, \{P_4, P_6\}\}$. Since $\Z_s$ and $\Z_a$ satisfy $\Q^{(3)}$
           and $\Q^{(4)}$ conditions respectively, it follows that {\it existing} SMPC protocols can tolerate $\Z_s$, while
           existing AMPC protocols can tolerate $\Z_a$. However, we show that {\it even} if the 
           parties are {\it not aware} of the exact network type, then using our protocol, one can {\it still} achieve security
           against $\Z_s$ if the network is {\it synchronous} or against $\Z_a$ if the network is {\it asynchronous}.
           The above example {\it also} demonstrates the flexibility offered by the non-threshold
           adversary model, in terms of tolerating {\it more} number of faults.
           More specifically, in the {\it threshold} model, using the protocol of  \cite{ACC22a}, one can tolerate up to
          $t_s = 2$ and $t_a = 1$ faults, in a {\it synchronous} and {\it asynchronous} network respectively.
          However, in the {\it non-threshold} model, our protocol can tolerate subsets of size larger than the maximum
          allowed
          $t_s$ and $t_a$ in synchronous and asynchronous network.       
          
            Even though our results generalize the results of \cite{ACC22a}, our protocols are relatively simpler compared to theirs. 
            For instance, one of the main ingredients used in their protocol is a best-of-both-worlds 
            {\it verifiable secret-sharing} (VSS) protocol. Their VSS is involved and 
            built upon another
            primitive called {\it weak polynomial-sharing} (WPS). 
            On the contrary, our best-of-both-worlds VSS protocol is relatively {\it simpler} and is {\it not}
            based on any WPS protocol.
\subsection{Technical Overview}
 Like in any generic MPC protocol, we assume that the underlying computation (which the parties want to perform securely)
  is modelled as some publicly-known function, abstracted 
   by some arithmetic circuit $\ckt$, over some algebraic structure $\K$, consisting of linear and non-linear (multiplication) gates.
   The problem of secure computation then reduces to secure {\it circuit-evaluation}, where
   the parties jointly and securely ``evaluate" $\ckt$ in a secret-shared fashion, such that all the values during the 
   circuit-evaluation remain {\it verifiably secret-shared} and where
   the shares of the corrupt parties {\it fail} to reveal the exact underlying value. 
   The secret-sharing used is typically {\it linear} \cite{CDM00}, thus allowing the parties to evaluate the linear gates
   {\it locally} (non-interactively). On the other hand, non-linear gates are evaluated by deploying
    the standard Beaver's method \cite{Bea91}
   using random, secret-shared {\it multiplication-triples} which are generated in a circuit-independent {\it preprocessing phase}. 
   Then, once all the gates are securely evaluated, the parties publicly reconstruct the secret-shared circuit-output.
   Apart from VSS \cite{CGMA85}, the parties also need to run instances of a
   {\it Byzantine agreement} (BA) protocol \cite{PSL80}
    to ensure that all the parties are
   on the ``same page" during the various stages of the circuit-evaluation. 
   The above framework  for shared circuit-evaluation is defacto 
   used in {\it all} generic perfectly-secure SMPC and AMPC protocols.
        Unfortunately, there are several challenges to adapt the framework in our  setting, where
      the parties will be
     {\it unaware} of the exact network type.
      \paragraph{\bf First Challenge --- A Best-of-Both-Worlds BA Protocol:}
     Informally, a BA protocol \cite{PSL80} allows parties with private inputs
      to reach agreement on a {\it common} output ({\it consistency}), such that the output is the 
       input of
     {\it honest} parties, if all honest parties have the same input ({\it validity}).
     {\it Perfectly-secure} BA protocols can be designed 
     against $\Q^{(3)}$ adversary structures {\it irrespective} of the network type \cite{FM98,Cho22}.
            However, the {\it termination} (also called {\it liveness}) guarantees are {\it different} for {\it synchronous} BA (SBA)
            and {\it asynchronous} BA (ABA) protocols. 
            The (deterministic) SBA protocols ensure that all honest parties obtain their output after some fixed time
             ({\it guaranteed liveness}) \cite{FM98}.
              On the other hand, to circumvent the FLP impossibility result \cite{FLP85},
             ABA protocols are {\it randomized} and provide  {\it almost-surely liveness} \cite{ADH08,BCP20,Cho22},
             where  the parties terminate the protocol with probability $1$, if they keep on running the protocol forever.           
             Known SBA protocols become insecure in an {\it asynchronous} network, while existing ABA protocols 
              can provide {\it only} almost-surely liveness in a {\it synchronous} network. 
             
             The {\it first} challenge to perform shared circuit-evaluation
             in our setting
              is to get a best-of-both-worlds BA protocol which provides 
              the security guarantees of SBA and ABA in a {\it synchronous} and an {\it asynchronous}
              network respectively. 
             We are {\it not} aware of any such BA protocol and hence, present a 
              BA protocol against $\Q^{(3)}$ adversary structures with the above properties.            
              As our BA protocol is technical, we defer its informal discussion to Section \ref{sec:BA}.
   \paragraph{\bf Second Challenge --- A Best-of-Both-Worlds VSS Protocol:}
   In a VSS protocol,
     there exists a {\it dealer} $\D$ with some private input $s$. 
     The protocol allows $\D$ to ``verifiably" distribute shares of $s$ to the parties, 
       such that 
       adversary's view remains independent of $s$, provided $\D$ is {\it honest} ({\it privacy}). 
        Moreover, 
        in a {\it synchronous} VSS (SVSS) protocol,
    every (honest) party obtains its shares after some {\it known} time-out ({\it correctness}).
    The {\it verifiability} here guarantees that even a {\it corrupt} $\D$ 
   shares some value ``consistently" within the known time-out ({\it commitment} property).
    Perfectly-secure SVSS against general adversaries is possible, provided the underlying
    adversary structure $\Z_s$ satisfies $\Q^{(3)}$ condition \cite{Mau02,HT13}.
  
    For an {\it asynchronous} VSS (AVSS) protocol, the {\it correctness} guarantees that for an {\it honest} 
    $\D$, the secret $s$ is eventually secret-shared.
    However, a {\it corrupt}
    $\D$ {\it may not} invoke the protocol in the first place, in which case the honest parties may not obtain any
    shares. 
     Hence, the {\it commitment} property of AVSS guarantees that if $\D$ is {\it corrupt} {\it and}
     if some honest party computes a share (implying that $\D$ has invoked the protocol), then all honest parties eventually 
     compute their shares.  Perfectly-secure AVSS against general adversaries is possible, 
    provided the underlying
    adversary structure $\Z_a$ satisfies the $\Q^{(4)}$ condition \cite{HT13,ACC22b}.
 
     Existing SVSS protocols become completely insecure in an asynchronous network, even if a
     single expected message from an {\it honest} party is {\it delayed}. 
     On the other hand, existing AVSS protocols become insecure against $\Q^{(3)}$ adversary structures
     (which will be the case, if the network is {\it synchronous}).
       Since ,in our setting, the parties will {\it not} be knowing the exact network type, 
    to maintain {\it privacy} during the shared circuit-evaluation, 
    we need to ensure that each value 
     remains secret-shared with respect to $\Z_s$ and {\it not} $\Z_a$, {\it even}
    if the network is {\it asynchronous}.\footnote{Since we are assuming that every subset in $\Z_a$ is a subset of some subset in
    $\Z_s$, the privacy will be maintained, both in the synchronous as well as
    asynchronous network, 
    if each value remains secret-shared with respect to $\Z_s$.} The {\it second} challenge to perform shared circuit-evaluation
    in our setting is to get a perfectly-secure VSS protocol which is secure with respect to $\Z_s$ and $\Z_a$
    in a {\it synchronous} and {\it asynchronous} network respectively, where {\it privacy}
    {\it always holds} with respect to $\Z_s$, {\it irrespective} of the network type. 
    We are not aware of any VSS protocol against general adversaries
     with these guarantees. Hence, we present a best-of-both-worlds  perfectly-secure VSS protocol satisfying the above properties.
     Since our VSS is slightly technical, we defer the informal discussion about the protocol to Section \ref{sec:VSS}.
    \subsection{Other Related Work}
    All existing works in the domain of the best-of-both-worlds protocols
    focus  only on {\it threshold} adversary model. The works of 
        \cite{BKL19} and \cite{BZL20,DHL21} show that the condition $2t_s + t_a < n$ is necessary and sufficient for 
         best-of-both-worlds {\it cryptographically-secure}
        BA and MPC respectively, tolerating {\it computationally bounded} adversaries. 
       Using the same condition, \cite{BKL21} presents a best-of-both-worlds {\it cryptographically-secure} atomic broadcast protocol. 
     The work of \cite{MO21} studies 
     Byzantine fault tolerance and state machine replication protocols for multiple thresholds, including $t_s$ and $t_a$. 
     The work of \cite{GLW22} presents best-of-both-worlds protocol for the task of approximate agreement using the condition
     $2t_s + t_a < n$.

\section{Preliminaries and Definitions}
\label{sec:prelims}
The parties in  $\PartySet$ are assumed to be connected 
 by pair-wise secure channels.
    The underlying communication network can be either synchronous or asynchronous, with parties
     being {\it unaware} about the exact type.
   In a {\it synchronous} network, every sent message  is delivered in the same order, within time
   $\Delta$. 
     In an {\it asynchronous} network, messages can be delayed  arbitrarily, but finitely,
      and {\it need not} be delivered in the same order. The only guarantee is that every sent message 
      is {\it eventually} delivered. The distrust is modeled by a centralized {\it malicious} (Byzantine)
      adversary $\Adv$, who can corrupt a subset of the parties in $\PartySet$ and force them to behave in any 
      arbitrary fashion during the execution of a protocol. 
      The adversary  is assumed to be {\it static}, and
      decides the set of corrupt parties at the beginning of the protocol execution. The
      adversary $\Adv$ is characterized by a {\it synchronous adversary structure}
      $\Z_s \subset 2^{\PartySet}$ and an {\it asynchronous adversary structure}
      $\Z_a \subset 2^{\PartySet}$. While in a {\it synchronous} network, $\Adv$ can corrupt
      any subset of parties from $\Z_s$, in an {\it asynchronous} network, 
      $\Adv$ can corrupt any subset from $\Z_a$.
      
      Given an arbitrary $\PartySet' \subseteq \PartySet$, and an arbitrary adversary structure $\Z \subset 2^{\PartySet}$, we
      say that $\Z$ satisfies the $\Q^{(k)}(\PartySet, \Z)$ condition \cite{HM97}, if the union of {\it any} $k$ subsets from $\Z$, 
      {\it does not} cover $\PartySet'$; i.e.~for every $Z_{i_1}, \ldots, Z_{i_k} \in \Z$,
      the condition $\PartySet' \not \subseteq Z_{i_1} \cup \ldots \cup Z_{i_k}$ holds. Given $\Z_s$ and $\Z_a$, we say that
      $\Z_s$ and $\Z_a$ satisfy the $\Q^{(k, k')}(\PartySet, \Z_s, \Z_a)$ condition, if the union of 
        any $k$ subsets from $\Z_s$ and any $k'$ subset from $\Z_a$, {\it does not} cover $\PartySet$. That is,
        for every $Z_{i_1}, \ldots,  Z_{i_k} \in \Z_s$ and every $Z_{j_1}, \ldots, Z_{j_{k'}} \in \Z_a$, 
        the condition $\PartySet \not \subseteq Z_{i_1} \cup \ldots \cup Z_{i_k} \cup Z_{j_1} \cup \ldots \cup Z_{j_{k'}}$ holds.
      
       We assume that
      $\Z_s$ and $\Z_a$ satisfy the following conditions, which we refer throughout the paper as {\it conditions $\Con$}.
      \begin{condition}[{\bf $\Con$}]
      \label{condition:Con}
      $\Z_s$ and $\Z_a$ satisfy the following conditions.
      \begin{myitemize}
        \item[--] $\Z_s \neq \Z_a$, and $\Z_s, \Z_a$ satisfy the $\Q^{(3, 1)}(\PartySet, \Z_s, \Z_a)$ condition.
        \item[--] For every subset $Z \in \Z_a$, there exists a subset $Z' \in \Z_s$, such that $Z \subseteq Z'$;
      \end{myitemize}
      \end{condition}
  Conditions $\Con$ imply that $\Z_s$ and $\Z_a$ satisfy the $\Q^{(3)}(\PartySet, \Z_s)$ 
   and  $\Q^{(4)}(\PartySet, \Z_a)$ conditions respectively.
   In our VSS and MPC protocols, 
   all computations are done over a finite algebraic structure $(\K, +, \cdot)$, which
   could be a ring or a field. We assume that each $P_i$ has an input $x_i \in \K$,
      and parties want to securely compute a function $f:\K^n \rightarrow \K$.
      Without loss of generality, $f$ is represented by an arithmetic circuit $\ckt$ over $\K$, 
      consisting of linear and non-linear (multiplication) gates, where
      $\ckt$ has $c_M$ multiplication gates and a multiplicative depth of $D_M$.
\paragraph{\bf  Termination Guarantees of Our Sub-Protocols:}
 As done in \cite{ACC22a},  for simplicity, we will {\it not} be specifying any {\it termination} criteria
 for our sub-protocols. The parties will keep on participating in these sub-protocol instances, {\it even} after receiving
 their outputs. The termination criteria of our MPC protocol will ensure the 
 termination of {\it all} underlying sub-protocol instances. We will be using an existing {\it randomized} ABA protocol \cite{Cho22}
  which ensures that the
 honest parties (eventually) obtain their respective output {\it almost-surely}. That is:
   \[ \underset{T \rightarrow \infty}{\mbox{lim}} \mbox{Pr}[\mbox{An honest } P_i \mbox{ obtains its output by local time } T] = 1,\]
 where the probability is over the random coins of the honest parties and adversary in the protocol. The property of
 almost-surely obtaining an output carries over to the “higher" level protocols, where ABA is used as a building block.
      
      We next discuss the 
   syntax and semantics of the secret-sharing, used in our VSS and MPC protocol.
    The secret-sharing is based on \cite{Mau02}, and is defined with respect to
   a given {\it sharing specification} $\SharingSpec$, which is a tuple of subsets of $\PartySet$. 
    \begin{definition}[\cite{Mau02}]
   Let $\SharingSpec = (S_1, \ldots, S_{|\SharingSpec|})$ be a sharing specification where,
   for $m = 1, \ldots, |\SharingSpec|$, each set $S_m \subseteq \PartySet$. Then 
   a value $s \in \K$ is said to be {\it secret-shared} with respect to $\SharingSpec$ if 
   there exist shares $s_1, \ldots, s_{|\SharingSpec|}$ such that $s = s_1 + \ldots + s_{|\SharingSpec|}$ and, 
   for $m = 1, \ldots, |\SharingSpec|$, the share $s_m$ is available with every (honest) party in $S_m$.  
  \end{definition}
  \noindent A secret-sharing of $s$ will be denoted by $[s]$, where $[s]_m$ denotes the $m^{th}$ share. Note that
  each $P_i$ holds multiple shares $\{[s]_m\}_{P_i \in S_m}$, corresponding to the sets from 
  $\SharingSpec$ to which it belongs. 
  The above secret-sharing is {\it linear} as 
  $[c_1s_1 + c_2 s_2] = c_1[s_1] + c_2[s_2]$ holds for any publicly-known $c_1, c_2 \in \K$.
  Hence, the parties can {\it non-interactively} compute any linear function over secret-shared inputs.
  
  For our protocols, we consider the specific sharing specification
    $\SharingSpec = (S_1, \ldots, \allowbreak S_q)$, where,
   for $m = 1, \ldots, q$, the set $S_m \defined \PartySet \setminus Z_m$, and where $\Z_s = \{Z_1, \ldots, Z_q \}$
    is the {\it synchronous} adversary structure. 
      
\subsection{Existing Primitives}
  \paragraph{\bf Asynchronous Reliable Broadcast (Acast):}
   An Acast protocol allows a {\it sender} 
  $\Sender \in \Partyset$ to send some message $m \in \{0, 1 \}^{\ell}$
   {\it identically} to all the parties.  The work of \cite{KF05} presents an Acast protocol against 
   $\Q^{(3)}$ adversary structures by generalizing the classic Bracha's Acast protocol against {\it threshold}
   adversaries. 
    While the protocol has been designed for an {\it asynchronous}
   network, it also provides certain guarantees in a {\it synchronous} network, as stated in Lemma \ref{lemma:Acast}.
     The Acast protocol $\PiACast$ and proof of Lemma \ref{lemma:Acast} are available in Appendix \ref{app:ExistingPrimitives}.
  \begin{lemma}
  \label{lemma:Acast}
Let $\Adv$ be an adversary characterized by an adversary structure $\AdvStructure$ satisfying
 the $\Q^{(3)}(\PartySet, \AdvStructure)$ condition. Then $\PiACast$ achieves the following.
 \begin{myitemize}
\item[--] {\it Asynchronous Network}: 
   {\bf (a) $\AdvStructure$-Liveness}: If $\Sender$ is {\it honest},  all honest parties eventually have an output.
    {\bf (b) $\AdvStructure$-Validity}: If $\Sender$ is {\it honest}, then each honest $P_i$ with an output, outputs $m$.
    {\bf (c) $\AdvStructure$-Consistency}: If $\Sender$ is {\it corrupt} and some honest $P_i$
     outputs $m^{\star}$, then all honest parties {\it eventually} output $m^{\star}$.
  \item[--] {\it Synchronous Network}:
           {\bf (a) $\AdvStructure$-Liveness}: If $\Sender$ is {\it honest}, then all honest parties obtain an output within time $3\Delta$.
    {\bf (b) $\AdvStructure$-Validity}: If $\Sender$ is {\it honest}, then every honest party with an output, outputs $m$.
    {\bf (c) $\AdvStructure$-Consistency}: If $\Sender$ is {\it corrupt} and some honest party outputs $m^{\star}$ at time $T$, then every honest $P_i$ outputs $m^{\star}$ by the
     end of time $T + 2 \Delta$.
  \item[--] {\it Communication Complexity}: $\Order(n^2 \ell)$ bits are communicated by the honest parties,
   where $\Sender$'s message is of size $\ell$ bits.
\end{myitemize}
\end{lemma}
\paragraph{\bf Terminologies for Using $\PiACast$:}
 We will say that  {\it $P_i$ Acasts $m$} to mean that $P_i$ acts as a sender $\Sender$ and invokes an instance of $\PiACast$
  with input $m$, and the parties participate in this instance. 
   Similarly, we say that {\it $P_j$ receives $m$ from the Acast of $P_i$} to mean that $P_j$ outputs $m$
  in the corresponding instance of $\PiACast$.
 \paragraph{\bf Public Reconstruction of a Secret-Shared Value:}
   Let  $s \in \K$ be a value, which is secret-shared with respect to the sharing specification 
  $\SharingSpec = \{S_m : S_m = \PartySet \setminus Z_m \; \mbox{ and } \; Z_m \in \Z_s \}$. 
   Let the goal is to publicly reconstruct $s$. 
   Since $\Z_s$ satisfies the $\Q^{(3)}(\PartySet, \Z_s)$ condition, we can use
   the reconstruction protocol $\Rec(s, \SharingSpec)$ of \cite{Mau02} which, {\it irrespective} of the network type, allows the parties to robustly reconstruct $s$. 
    In a {\it synchronous} network, the protocol will take
   $\Delta$ time, while in an {\it asynchronous} network, the parties eventually output $s$.
   The protocol incurs a communication of $\Order(|\Z_s| \cdot n^2 \log{|\K|})$ bits; see 
     Appendix \ref{app:ExistingPrimitives} for the details.
    \paragraph{\bf Beaver's Circuit-Randomization Method \cite{Bea91}:}
    Let $u$ and $v$ be secret-shared among the parties. The goal is to compute a secret-sharing of $w = u \cdot v$.
  Moreover, let $(a, b, c)$ be a shared {\it multiplication-triple} available with the parties such that $c = a \cdot b$.
  Then, Beaver's method allows the parties to compute a secret-sharing of $w$ such that,
   if $a$ and $b$ are random for the adversary, then 
   the view of the adversary remains independent of $u$ and $v$. 
   In a {\it synchronous} network, the parties compute $[w]$ within time $\Delta$, while
   in an {\it asynchronous} network, the parties eventually compute $[w]$. 
   Protocol $\BatchBeaver(([u], [v]), ([a], [b], [c]))$
    incurs a communication of $\Order(|\Z_s| \cdot n^2 \log{|\K|})$ bits, and is 
    presented in 
     Appendix \ref{app:ExistingPrimitives}.

\section{Best-of-Both-Worlds Byzantine Agreement (BA)}
\label{sec:BA}
 We begin with the definition of BA, which is 
 adapted from \cite{BZL20,ACC22a}.
\begin{definition}[{\bf BA}]
\label{def:BA}
Let $\Pi$ be a protocol for $\Partyset$, where every party $P_i$ has an input $b_i \in \{0, 1\}$ and a possible output from  $\{0, 1, \bot \}$.
 Moreover, let $\Adv$ be an adversary, characterized by adversary structure $\AdvStructure$, where $\Adv$ can corrupt
  any subset of parties from $\AdvStructure$ during the execution of $\Pi$.
     \begin{myitemize}
    \item[--] {\bf $\AdvStructure$-Guaranteed Liveness}: $\Pi$ has $\AdvStructure$-guaranteed liveness
    if all honest parties obtain an output.
       \item[--] {\bf $\AdvStructure$-Almost-Surely Liveness}: $\Pi$ has $\AdvStructure$-almost-surely liveness
       if, almost-surely, all honest parties obtain some output.
      \item[--] {\bf $\AdvStructure$-Validity}: $\Pi$ has $\AdvStructure$-validity if the following holds:
      If all honest parties have input $b$, then every honest party with an output, outputs $b$. 
      \item[--] {\bf $\AdvStructure$-Weak Validity}: $\Pi$ has $\AdvStructure$-weak validity if the following holds:
      If all honest parties have input $b$, then every honest party with an output, outputs $b$ or $\bot$.      
      \item[--] {\bf $\AdvStructure$-Consistency}: $\Pi$ has $\AdvStructure$-consistency
      if all honest parties with an output, output the same value (which can be $\bot$).
      \item[--] {\bf $\AdvStructure$-Weak Consistency}: $\Pi$ has $\AdvStructure$-weak consistency
      if all honest parties with an output, output either a common $v \in \{0, 1 \}$ or $\bot$.
      \end{myitemize}
  $\Pi$ is called a {\it $\AdvStructure$-perfectly-secure synchronous BA} (SBA) protocol if, in a {\it synchronous} network, it has
    {\it $\AdvStructure$-guaranteed liveness}, {\it $\AdvStructure$-validity}, and {\it $\AdvStructure$-consistency}.
   $\Pi$ is called a {\it $\AdvStructure$-perfectly-secure 
  asynchronous BA} (ABA) protocol if, in an {\it asynchronous network} it has 
  {\it $\AdvStructure$-almost-surely liveness},  {\it $\AdvStructure$-validity} and {\it $\AdvStructure$-consistency}.
 \end{definition}
 To design our BoBW BA protocol, we will be using an
  {\it existing} perfectly-secure SBA and a perfectly-secure ABA protocol, whose properties
  we review next.
  \paragraph{\bf Existing SBA and ABA Protocols:}
  We assume the existence of a $\AdvStructure$-perfectly-secure SBA protocol $\SBA$ with 
   $\Q^{(3)}(\PartySet, \AdvStructure)$ condition, which {\it also} provides 
   $\AdvStructure$-guaranteed liveness in an {\it asynchronous} network.\footnote{We {\it do not} require any other
   property from $\SBA$ in an {\it asynchronous} network.}. For the sake of efficiency, we design a candidate for
   $\SBA$ by generalizing the simple SBA protocol of \cite{BGP89}, which was designed to tolerate
   $t < n/3$ corruptions. The protocol requires at most $3n$ rounds in a {\it synchronous} network and hence, within
   time $\TimeSBA \defined 3n \cdot \Delta$, all honest parties will get an output in a {\it synchronous}
   network. The protocol incurs a communication of
    $\Order(n^3 \ell)$ bits if the inputs of the parties are of size $\ell$ bits.
   To achieve $\AdvStructure$-guaranteed liveness in an {\it asynchronous} network,  
    the parties can run $\SBA$ till time $\TimeSBA$, and 
    then output $\bot$ if no ``valid" output is computed as per the protocol at time 
    $\TimeSBA$. 
    This guarantees that even in an {\it asynchronous} network, all {\it honest}
    parties obtain {\it some} output at local time $\TimeSBA$. 
    Our $\SBA$ protocol and the proof of its properties are available in Appendix \ref{app:SBA}.
       
   For the {\it asynchronous} network, \cite{Cho22} presents a $\AdvStructure$-perfectly-secure ABA 
   protocol $\ABA$, provided
   $\AdvStructure$ satisfies the $\Q^{(3)}(\PartySet, \AdvStructure)$ condition. Protocol $\ABA$ has the following properties
   in the synchronous and asynchronous network.
    \begin{lemma}[\cite{Cho22}]
 \label{lemma:ABAGuarantees}
 Let $\Adv$ be an adversary characterized by an adversary structure 
   $\AdvStructure$, satisfying the $\Q^{(3)}(\PartySet, \AdvStructure)$ condition. Then, there exists 
  a BA protocol $\ABA$ tolerating $\Adv$ such that:
 \begin{myitemize}
 \item[--] {\bf Asynchronous Network}: The protocol is a {\it $\AdvStructure$-perfectly-secure} ABA protocol with the following liveness
  guarantees.
     \begin{myitemize}
      \item[-- ] If the inputs of all {\it honest} parties are the same, then $\ABA$ achieves $\AdvStructure$-guaranteed liveness.
      Else, $\ABA$ achieves  $\AdvStructure$-almost-surely liveness.
      \end{myitemize}
 \item[--] {\bf Synchronous Network}: The protocol achieves $\AdvStructure$-validity, $\AdvStructure$-consistency, and 
    the following liveness guarantees.
    \begin{myitemize}
    \item[--] If the inputs of all {\it honest} parties are the same, then $\ABA$ achieves $\AdvStructure$-guaranteed liveness, and 
    all honest parties obtain their output within time $\TimeABA = k \cdot \Delta$, for some constant $k$. 
     \item[--] Else, $\ABA$ achieves $\AdvStructure$-almost-surely liveness
    and requires $\Order(\mbox{poly}(n) \cdot \Delta)$ expected time to generate the output.
    \end{myitemize}
 \item[--] {\bf Communication Complexity}: 
  $\Order(|\AdvStructure| \cdot n^5 \log|\F| + n^6 \log n)$ bits are communicated by the honest parties, if their inputs
  are the same. Else, $\ABA$ incurs an expected communication of 
   $\Order(|\AdvStructure| \cdot n^7 \log|\F| + n^8 \log n)$ bits. Here $\F$ is a finite field such that
   $|\F| > n$ holds.
 \end{myitemize}
 \end{lemma}
 We note that $\ABA$ {\it cannot} be considered as a best-of-both-worlds BA protocol.
  This is because
  it achieves {\it $\AdvStructure$-guaranteed liveness} in a {\it synchronous} network {\it only} when {\it all} honest parties have the {\it same} input. 
  If the honest parties start  $\ABA$ with {\it different} inputs, then {\it instead} of guaranteed liveness,
  the parties may keep on running the protocol forever, {\it without} obtaining any output (though the probability
  of this happening is asymptotically $0$). 
   We next design a BA protocol which gets rid of this problem,
   and which is secure in {\it any} network. 
   To design the protocol, we need a special type of {\it broadcast} protocol, which we design first.
 \subsection{Synchronous Broadcast with Asynchronous Guarantees}
  We begin with the definition of broadcast, 
   adapted from \cite{BZL20,ACC22a}.
\begin{definition}[{\bf Broadcast}]
\label{def:BCAST}
Let $\Pi$ be a protocol, where a sender $\Sender \in \Partyset$ has input $m \in \{0, 1\}^{\ell}$, and parties obtain a possible output,
  including $\bot$.
 Moreover, let $\Adv$ be an adversary, characterized by adversary structure $\AdvStructure$, where $\Adv$ can corrupt
  any subset of parties from $\AdvStructure$ during the execution of $\Pi$.
   \begin{myitemize}
   \item[--] {\bf $\AdvStructure$-Liveness}: $\Pi$ has $\AdvStructure$-liveness if 
    all honest parties obtain some output.
    \item[--] {\bf $\AdvStructure$-Validity}: $\Pi$ has $\AdvStructure$-validity if the following holds:
    if $\Sender$ is {\it honest}, then every honest party with an output, outputs $m$.
     \item[--] {\bf $\AdvStructure$-Weak Validity}: $\Pi$ has $\AdvStructure$-weak validity if the following holds: 
      if $\Sender$ is {\it honest}, then every honest party with an output, outputs either $m$ or $\bot$.
     \item[--] {\bf $\AdvStructure$-Consistency}: $\Pi$ has $\AdvStructure$-consistency if the following holds:
     if $\Sender$ is {\it corrupt}, then every honest party with an output, outputs a common value.
     \item[--] {\bf $\AdvStructure$-Weak Consistency}: $\Pi$ has $\AdvStructure$-weak consistency if the
     following holds: if $\Sender$ is {\it corrupt}, then every honest party with an output, outputs a common
      $m^{\star} \in \{0, 1 \}^{\ell}$ or $\bot$.
  \end{myitemize}
$\Pi$ is called a {\it $\AdvStructure$-perfectly-secure broadcast} protocol if it has {\it $\AdvStructure$-Liveness}, 
{\it $\AdvStructure$-Validity}, and {\it $\AdvStructure$-Consistency}.
 \end{definition}
 We next design a special broadcast protocol $\BCAST$, which is a $\AdvStructure$-perfectly-secure broadcast protocol
  in a {\it synchronous} network. Additionally, in an {\it asynchronous} network, the protocol has
  $\AdvStructure$-Liveness, $\AdvStructure$-Weak Validity and $\AdvStructure$-Weak Consistency. 
  Looking ahead, we will combine the protocols $\BCAST$ and $\ABA$ to get our best-of-both-worlds
  BA protocol.
    We note that
  the existing Acast protocol $\PiACast$ {\it does not} guarantee the same properties as $\BCAST$ since,
  for a {\it corrupt} $\Sender$, there is {\it no} liveness guarantee (irrespective of the network type).
  Moreover, in a {\it synchronous} network, 
     {\it all} honest parties {\it may not} obtain an output within the {\it same} time, if $\Sender$ is {\it corrupt}
    (see Lemma \ref{lemma:Acast}).\footnote{Looking ahead, this property from
    $\BCAST$ will be crucial when we design our best-of-both-worlds
    BA protocol.}  
    
    To design $\BCAST$ (Fig \ref{fig:BCAST}), we generalize an idea used in \cite{ACC22a} against {\it threshold} adversaries.
     The idea is to carefully ``stitch" together protocol $\PiACast$ with
    the  protocol $\SBA$. In the protocol, $\Sender$ first Acasts its message. 
    If the network is {\it synchronous}, then at time $3\Delta$, all honest parties should have an output.
    To confirm this, the parties start participating in an instance of $\SBA$, where the input of each party is the output that party has
    obtained from $\Sender$'s Acast instance 
     at time $3\Delta$. It is possible that a party has no output at time $3\Delta$ (implying that either the network is {\it asynchronous}
     or $\Sender$ is {\it corrupt}), in which case the input of the party for $\SBA$ will be $\bot$.
     Finally,  at time $3\Delta + \TimeSBA$, the parties output $m^{\star}$,
     if it has been received from the
     Acast of $\Sender$ {\it and} is the output of $\SBA$ as well, else, they output $\bot$.
      It is easy to see that the protocol has {\it guarantees liveness} in {\it any} network
          since all parties will have some output at (local) time $3\Delta + \TimeSBA$. 
          Moreover, in a {\it synchronous} network, if some {\it honest} party has an output $m^{\star} \neq \bot$ at time
          $3\Delta + \TimeSBA$, then {\it all} the honest parties will also have the output $m^{\star}$ at time 
        $3\Delta + \TimeSBA$. This is
         because if any {\it honest} party obtains an output $m^{\star}$,
    then {\it at least} one {\it honest} party must have received $m^{\star}$ from $\Sender$'s Acast by time $3\Delta$
    and so by time $3\Delta + \TimeSBA$, {\it all} honest parties will receive $m^{\star}$ from $\Sender$'s Acast. 
    \paragraph {\bf Eventual Consistency and Validity for $\BCAST$ in Asynchronous Network:}
 In $\BCAST$, the parties set a ``time-out" of $3\Delta + \TimeSBA$ to guarantee liveness. However, in this process,
  the protocol {\it only} guarantees {\it weak validity} and {\it weak consistency}  in an {\it asynchronous} network. This is because
   some {\it honest} parties may receive $\Sender$'s message from the Acast of $\Sender$ within time $3\Delta + \TimeSBA$, 
   while others may fail to do so. 
  The time-out is essential, as we need {\it liveness} from
   $\BCAST$ (irrespective of the network type) when 
   used later in our best-of-both-worlds BA protocol. 
   
    Looking ahead, we will use $\BCAST$ in our VSS protocol for broadcasting values.
     However, the  {\it weak validity} and {\it weak consistency} properties may lead to a situation
       where, in an {\it asynchronous} network,
              one subset of {\it honest} parties may output a value different from $\bot$ at the end of the time-out, while 
      others may output $\bot$. For the security of the VSS protocol, we would require the latter category of parties 
      to {\it eventually} output the common non-$\bot$ value if the parties {\it continue} participating in $\BCAST$.
      To achieve this goal, we make a provision in $\BCAST$. Namely,  
      each
    $P_i$ who outputs $\bot$ at time $3\Delta + \TimeSBA$ ``switches" its output to $m^{\star}$, if
     $P_i$ {\it eventually} receives $m^{\star}$ from $\Sender$'s Acast. 
    We stress that this switching is {\it only} for the parties who obtained $\bot$
    at time $3\Delta + \TimeSBA$. 
    To differentiate between the two ways of obtaining output, we use the terms {\it regular-mode}
    and {\it fallback-mode}. Regular-mode consists of the process of deciding the output at time $3\Delta + \TimeSBA$, while
    fallback-mode is the process of deciding the output beyond time $3\Delta + \TimeSBA$.

       \begin{protocolsplitbox}{$\BCAST$}{Synchronous broadcast with asynchronous guarantees.}{fig:BCAST}
\centerline{\underline{(Regular Mode)}}
   \begin{myitemize}
   \item[--] On having the input $m \in \{0, 1 \}^{\ell}$, sender $\Sender$ Acasts $m$.
   \item[--] At time $3\Delta$, each $P_i \in \Partyset$ participates in an instance of $\SBA$, where the input of $P_i$ is $m^{\star}$ if
    $m^{\star} \in \{0, 1\}^{\ell}$ is received from the Acast of 
    $\Sender$, else the input is $\bot$.  
    \item[--] {\bf (Local Computation)}: At time $3\Delta + \TimeSBA$, each $P_i \in \Partyset$ does the following.
       \begin{myitemize}
        \item[--] If some $m^{\star} \in \{0, 1 \}^{\ell}$ is received from the Acast of $\Sender$ {\it and} $m^{\star}$ is 
        computed as the output during the instance of $\SBA$, then output $m^{\star}$. 
         Else output $\bot$.
      \end{myitemize}
\end{myitemize}
\centerline{\underline{{\color{blue}(Fallback Mode)}}}
 \begin{myitemize}
 \item[--] {\color{blue} Every $P_i \in \Partyset$ who has computed the output $\bot$ at time $3\Delta + \TimeSBA$,
 changes it to $m^{\star}$, if $m^{\star}$ is received by $P_i$ from the Acast of $\Sender$.
 }
\end{myitemize}
\end{protocolsplitbox}

The properties of $\BCAST$, stated in Theorem \ref{thm:BCAST}, are proved in 
  Appendix \ref{app:HBA}. 
\begin{theorem}
\label{thm:BCAST}
Let $\Adv$ be an adversary, characterized by $\AdvStructure$, satisfying $\Q^{(3)}(\PartySet, \AdvStructure)$
 condition. Let $\Sender$ has input
   $m \in \{0, 1 \}^{\ell}$ for $\BCAST$. Then
  $\BCAST$ achieves the following,
   with a communication complexity of $\Order(n^3 \ell)$ bits,
    where $\TimeBCAST =  3\Delta + \TimeSBA$.   
 \begin{myitemize}
   \item[--] {\it Synchronous} network: 
      {\bf (a) $\AdvStructure$-Liveness}: At time $\TimeBCAST$, each honest party has an output. 
    {\bf (b) $\AdvStructure$-Validity}: If $\Sender$ is {\it honest}, then at time $\TimeBCAST$, each honest party
    outputs $m$.
      {\bf (c) $\AdvStructure$-Consistency}: If $\Sender$ is {\it corrupt}, then the output of every honest party is the same at 
      time $\TimeBCAST$.     
     {\bf (d) $\AdvStructure$-Fallback Consistency}: If $\Sender$ is {\it corrupt}, and some honest 
     party outputs $m^{\star} \neq \bot$ at time
    $T$ through fallback-mode, then every honest party outputs $m^{\star}$ by 
    time $T + 2\Delta$.    
\item[--] {\it Asynchronous Network}:
   {\bf (a) $\AdvStructure$-Liveness}: At time $\TimeBCAST$, each honest party has an output.
    {\bf (b) $\AdvStructure$-Weak Validity}: If $\Sender$ is {\it honest}, then at  time $\TimeBCAST$, 
    each honest party
    outputs $m$ or $\bot$.
    {\bf (c) $\AdvStructure$-Fallback Validity}: If $\Sender$ is {\it honest}, then each honest party
     with output $\bot$ at time $\TimeBCAST$, eventually outputs 
    $m$ through fallback-mode.
    {\bf (d) $\AdvStructure$-Weak Consistency}: If $\Sender$ is {\it corrupt}, then at time $\TimeBCAST$, 
    each honest party
    outputs a common
     $m^{\star} \neq \bot$ or $\bot$.
    {\bf (e) $\AdvStructure$-Fallback Consistency}: If $\Sender$ is {\it corrupt}, and some honest party 
    outputs $m^{\star} \neq \bot$ at  time
    $T$ where $T \geq \TimeBCAST$, then 
    each honest party eventually outputs $m^{\star}$.
   \end{myitemize}
\end{theorem}
In the rest of the paper, we use the following terminologies while using $\BCAST$.
\paragraph {\bf Terminologies for $\BCAST$:}  
 We say that  {\it $P_i$ broadcasts $m$} to mean that $P_i$ invokes an instance of 
 $\BCAST$ as  $\Sender$
  with input $m$, and the parties participate in this instance. Similarly, we say that {\it $P_j$ receives $m$ from the broadcast
   of $P_i$ through regular-mode (resp.~fallback-mode)},
   to mean that $P_j$ has the output $m$
  at time $\TimeBCAST$ (resp.~after time $\TimeBCAST$)
  during the instance of $\BCAST$. 
 \subsection{Protocols $\BCAST + \ABA \Rightarrow$ Best-of-Both-Worlds BA} 
 We now combine protocols $\BCAST$ and $\ABA$, by generalizing the
  idea used in \cite{ACC22a} against {\it threshold} adversaries.
  In the protocol, every party first
  broadcasts its input bit (for the BA protocol) through an instance of $\BCAST$. If the network is {\it synchronous}, then 
  all honest parties should have received the inputs of all the (honest) sender parties from their broadcasts through
  regular-mode by time $\TimeBCAST$.
  Consequently, at time $\TimeBCAST$, the parties decide an output for {\it all} the $n$ instances of $\BCAST$.
  Based on these outputs, the parties decide their respective inputs for the $\ABA$ protocol.
  Specifically, if ``sufficiently many" outputs from the $\BCAST$ instances are found to be {\it same}, then the 
  parties consider this output value as their input for the $\ABA$ instance. Else, they stick to their original inputs. The overall output for
  $\BA$ is then set to be the output from $\ABA$.
\begin{protocolsplitbox}{$\BA$}{The best-of-both-worlds BA.
  The above code is executed by every $P_i \in \PartySet$.}{fig:BA}
\justify
\begin{myitemize}
\item[--] 
 On having input $b_i \in \{0, 1 \}$, broadcast $b_i$.
\item[--] For $j = 1, \ldots, n$, let $b_i^{(j)} \in \{0, 1, \bot \}$ be received from the broadcast of $P_j$ through {\bf regular-mode}. 
  Include $P_j$ to a set $\R$ if $b_i^{(j)} \neq \bot$. 
  Compute the input $v_i^{\star}$ for an instance of $\ABA$ as follows.
    \begin{myitemize}
    \item[--] If 
    $\PartySet \setminus \R \in \AdvStructure$, then compute $v_i^{\star}$ as follows.
       \begin{myitemize}
       \item[--] If there exists a subset of parties $\R_i \subseteq \R$, such that $\R \setminus \R_i  \in \AdvStructure$
       and $b_i^{(j)} = b$ for all the parties $P_j \in \R_i$, then set
        $v_i^{\star} = b$.\footnote{If there are multiple such $\R_i$, then break the tie using some 
        pre-determined rule.}       
       \item[--] Else set $v_i^{\star} = 1$.
       \end{myitemize}
    \item[--] Else set $v_i^{\star} = b_i$.
    \end{myitemize}
\item[--] {\color{red} At time $\TimeBCAST$},
 participate in an instance of $\ABA$ with input $v_i^{\star}$. Output the result of $\ABA$.
\end{myitemize}
\end{protocolsplitbox}

The properties of $\BA$, stated in Theorem \ref{thm:BA}, are proved in  Appendix \ref{app:HBA}. 
\begin{theorem}
\label{thm:BA}
Let $\Adv$ be an adversary characterized by an adversary structure $\AdvStructure$ satisfying the $\Q^{(3)}(\PartySet, \AdvStructure)$
 condition. Moreover, let $\ABA$ be an ABA protocol, 
   satisfying the conditions as stated in Lemma \ref{lemma:ABAGuarantees}.
   Then, $\BA$ achieves the following.
\begin{myitemize}
\item[--] {\bf Synchronous Network}: The protocol is a {\it $\AdvStructure$-perfectly-secure} SBA 
 protocol, where all honest parties obtain an output within time $\TimeBA = \TimeBCAST + \TimeABA$.
  The protocol incurs a communication of $\Order(|\AdvStructure| \cdot n^5 \log|\F| + n^6 \log n)$ bits.
\item[--] {\bf Asynchronous Network}: The protocol is a {\it $\AdvStructure$-perfectly-secure} ABA protocol, 
with an expected communication of  $\Order(|\AdvStructure| \cdot n^7 \log|\F| + n^8 \log n)$ bits.
\end{myitemize}
\end{theorem}

\section{Best-of-Both-Worlds VSS Protocol}
\label{sec:VSS}
 We present our best-of-both-worlds VSS protocol 
   $\VSS$ (Fig \ref{fig:VSS}), assuming that the conditions $\Con$
    (see Condition \ref{condition:Con} in Section \ref{sec:prelims}) hold.
     In the protocol, there exists a 
    {\it dealer} $\D \in \PartySet$ with a private input $s \in \K$. The goal is to ``verifiably" generate
   a secret-sharing of $s$ with respect to the sharing specification
   $\SharingSpec = \{S_m : S_m = \PartySet \setminus Z_m \; \mbox{ and } \; Z_m \in \Z_s \}$, {\it irrespective}
   of the network type.
  If $\D$ is {\it honest}, then in an {\it asynchronous} network, $s$ is {\it eventually} secret-shared. In a {\it synchronous} network, $s$ is secret-shared after a {\it fixed} time such that the view of the adversary
   remains independent of $s$, {\it irrespective} of the network type.
   Note that $s$ is {\it always} secret-shared with respect to $\SharingSpec$, which  
   is defined with respect to the {\it synchronous} adversary structure $\Z_s$, {\it even} if the network is {\it asynchronous}.
   
       The {\it verifiability}  of $\VSS$ guarantees that 
   if $\D$ is {\it corrupt}, then either no honest party obtains
   any output (this happens if $\D$ {\it does not} invoke the protocol at the first place), or
   there exists some value $s^{\star} \in \K$ (which may be different from $s$) to which
   $\D$ is ``committed" and which is secret-shared with respect to $\SharingSpec$.
    Note that in the latter case, we cannot bound the time within which $s^{\star}$ will be secret-shared, 
   even if the network is {\it synchronous}. This is because 
   a {\it corrupt} $\D$ may delay its messages arbitrarily,
   and the parties will {\it not know} the network type. 
   
   Protocol $\VSS$ is obtained by carefully ``stitching" together the {\it synchronous} VSS (SVSS)
    and {\it asynchronous} VSS (AVSS) protocols of
    \cite{Mau02} and \cite{CP20} respectively. We first explain
    the idea behind these protocols individually, and then proceed to explain how we stitch them together.
    \paragraph{\bf SVSS  Against $\Q^{(3)}$ Adversary Structures:}
     The SVSS protocol of \cite{Mau02} is executed in a sequence of {\it synchronized phases}, and can tolerate $\Q^{(3)}$ adversary 
     structures. Consider an arbitrary adversary structure  $\Z$ satisfying the $\Q^{(3)}$ condition, and let
      $\SharingSpec_{\Z} = (S_1, \ldots, S_{|\Z|})$ be the sharing specification where 
      $S_m = \PartySet \setminus Z_m$, for $m = 1, \ldots, |\Z|$. 
      To share $s$, during the {\it first} phase, $\D$ picks a random vector of shares $(s_1, \ldots, s_{|\Z|})$, which sum up to
      $s$. Then all the parties in $S_m$ are given the share $s_m$. To verify if all the (honest) parties 
      in $S_m$ have received the {\it same} share from $\D$, the parties in $S_m$ perform a {\it pairwise consistency} check of their
      supposedly common share during the {\it second} phase,
       and publicly broadcast the results during the {\it third} phase.
      If any party in $S_m$ publicly complaints for an inconsistency, then during the {\it fourth} phase, $\D$
      makes the share $s_m$ corresponding to $S_m$ {\it public} by broadcasting it. Note that this {\it does not} violate the privacy for an 
      {\it honest} $\D$, since a complaint for inconsistency from $S_m$ implies that $S_m$ has at least one {\it corrupt}
      party and so, the adversary will already know the value of $s_m$.
      If $\D$ {\it does not} ``resolve" any complaint during the fourth phase (implying $\D$ is {\it corrupt}),
      then it is {\it publicly discarded}, and everyone takes a default sharing of some publicly-known value on the behalf of
      $\D$.  
       The protocol ensures that by the end of the {\it fourth} phase,
       {\it all honest} parties in $S_m$ have the {\it same} share, and that the sum of these shares across all the $S_m$ sets is the value shared by $\D$.      
   \paragraph{\bf AVSS Against $\Q^{(4)}$ Adversary Structures:}       
    The AVSS protocol of \cite{CP20} closely follows the SVSS protocol of \cite{Mau02}. However, the phases are {\it no
    longer} synchronized. Moreover, during the pairwise consistency phase, the parties {\it cannot} afford to wait to know the status of the
    consistency checks between all pairs of parties, since the potentially {\it corrupt} parties may {\it never} respond.
    Instead, corresponding to every $S_m$, the parties check for the existence of a set of ``core" parties
    $\C_m \subseteq S_m$, with $S_m \setminus \C_m \in \Z$,
     who publicly confirmed the receipt of the same share from $\D$. To ensure that all the parties agree
     on the core sets, $\D$ is assigned the task of identifying the core sets and broadcasting them.
    The protocol proceeds {\it only} upon the receipt of core sets from $\D$ and their verification.
    While an {\it honest} $\D$ will eventually find and broadcast core sets, a {\it corrupt}
    $\D$ may {\it not} do so, in which case the parties obtain no shares.
    Once the core sets are identified and verified, it guarantees that all the (honest)
    parties in every core set $\C_m$ have received the same share from $\D$.
    The goal will then be to ensure that even the (honest) parties ``outside" $\C_m$
    (namely, the parties in $S_m \setminus \C_m$) get this common share.
    Since $\Z$ now satisfies the $\Q^{(4)}$ condition, the ``majority" of the parties in $\C_m$ are guaranteed to be
    {\it honest}. Hence, the parties in $S_m \setminus \C_m$ can ``filter" out the common share held by the parties in $\C_m$,
    by applying the ``majority rule" on the shares received from the parties in $\C_m$ during pairwise consistency tests.
    \paragraph{\bf Best-of-Both-Worlds VSS Protocol with Conditions $\Con$:}
    In protocol $\VSS$, the parties first start executing the steps of the above SVSS protocol, {\it assuming} a {\it synchronous} network,
    where all the instances of broadcast happen by executing an instance of $\BCAST$ with respect to the adversary structure
    $\Z_s$. If indeed the network in {\it synchronous}, then within time $2\Delta + \TimeBCAST$, the results of
    pairwise consistency tests will be publicly available. Moreover, if any inconsistency is reported, then within time
    $2\Delta + 2\TimeBCAST$, the dealer $\D$ should have resolved all those inconsistencies by making public the ``disputed" shares. 
    However, unlike the SVSS protocol, the parties {\it cannot} afford to discard $\D$ if it fails to resolve any inconsistency
    within time  $2\Delta + 2\TimeBCAST$,  as the network could be {\it asynchronous},
    and $\D$'s responses may be arbitrarily {\it delayed}, even if $\D$ is {\it honest}.
    Moreover, in an {\it asynchronous} network, some honest parties may be seeing the inconsistencies
    being reported within time $2\Delta + \TimeBCAST$ as well as $\D$'s responses within time  
    $2\Delta + 2\TimeBCAST$, while other honest parties {\it may not} be seeing these inconsistencies and $\D$'s
    responses within these timeouts. This may result in the {\it former} set of honest parties considering
    the shares made public by $\D$, while the latter set of honest parties, thinking that the network is {\it asynchronous}, wait for core sets of parties to be made public
    by $\D$ (as done in the AVSS). However, this may lead to the 
    {\it violation} of the {\it commitment} property in case
    $\D$ is {\it corrupt}, and network is {\it asynchronous}.
     In more detail, consider a set $S_m$ for which pairwise {\it inconsistency} is reported, and for which $\D$ also finds a set of
    core parties $\C_m$. Then, it might be possible that the parties in $\C_m$ have received the common share $s_m$ from $\D$,
    but in response to
    the inconsistencies reported for $S_m$, dealer $\D$ responds with $s'_m$, where $s'_m \neq s_m$.
    This will lead to a situation where one set of honest parties (who see inconsistencies and $s'_m$ within the timeout of
    $2\Delta + \TimeBCAST$ and $2\Delta + 2\TimeBCAST$ respectively) consider $s'_m$ as the share for $S_m$, while another set
    of honest parties, who do not see the inconsistencies and $s'_m$ within the timeout, eventually see $\C_m$ and filter out the share $s_m$.     
    
    To deal with the above problem, apart from resolving the inconsistencies reported for {\it any} set $S_m$, 
    the dealer $\D$ {\it also} finds and broadcasts a core set of parties $\C_m$, who have confirmed receiving the same share
    from $\D$ corresponding to {\it all} the sets $S_m$, such that $S_m \setminus \C_m \in \Z_s$. 
    {\it Additionally}, if there is any inconsistency reported for $S_m$, then {\it apart} from $\D$, {\it every}
    party in $S_m$ {\it also} makes public its version of the share corresponding to $S_m$, which it has received from $\D$.
    Now, at time  $2\Delta + 2\TimeBCAST$, the parties check if $\D$ has broadcasted a core set $\C_m$
     for each $S_m$. Moreover, if
    any inconsistency has been reported corresponding to $S_m$, the parties check if ``sufficiently many" parties from
    $\C_m$ have made public the same share, as made public by $\D$. This {\it prevents} a {\it corrupt}
    $\D$ from making public a share, which is {\it different} from the share which it distributed to the parties in $\C_m$.
    
   If the network is {\it asynchronous}, then different parties may have {\it different} ``opinion" regarding whether 
   $\D$ has broadcasted ``valid" core sets $\C_m$. Hence, at time $2\Delta + 2\TimeBCAST$, the parties run an instance of
   our $\BA$ protocol to decide what the case is. If the parties find that $\D$ has broadcasted valid core sets
   $\C_m$ corresponding to each $S_m$, then the parties in $S_m$ proceed to compute their share as follows:
   if $\D$ has made public the share for $S_m$ in response to any inconsistency, then it is taken as the share for $S_m$.
   If no share has been made public for $S_m$, then the parties check if ``sufficiently many" parties have reported the same share
   during the pairwise consistency test within time $2\Delta$, which we show should have happened if the network is
   {\it synchronous}, and if the parties maintain sufficient timeouts. If none of these conditions hold, then the parties
   proceed to filter out the common share, held by the parties in $\C_m$, through ``majority rule".
   
   If $\BA$ indicates that $\D$ has {\it not} made public core sets within time $2\Delta + 2\TimeBCAST$, then 
   either the network is {\it asynchronous} or $\D$ is {\it corrupt}. So the parties resort to the steps used in
   AVSS. Namely,
   $\D$ finds and broadcasts a set of core parties $\E_m$ corresponding to each $S_m$, where
    $S_m \setminus \E_m \in \Z_a$.
   Then, the parties filter out the common share, held by the parties in $\E_m$, through majority rule.

\begin{protocolsplitbox}{$\VSS(\D, s, \SharingSpec)$}{Best-of-both-worlds VSS protocol for the sharing 
 specification $\SharingSpec$.}{fig:VSS}
Let $\SharingSpec = (S_1, \ldots, S_m, \ldots, S_q)$, where for $m = 1, \ldots, q$,
 the set $S_m \defined \PartySet \setminus Z_m$, and where $\Z_s = \{Z_1, \ldots, Z_q \}$
    is the {\it synchronous} adversary structure. 
\begin{myitemize}
    \item[--] {\bf Phase I --- Share Distribution}: $\D$, on having the input $s$, does the following.
        \begin{myitemize}
            \item[--] Randomly select shares $s^{(1)}, \ldots, s^{(q)} \in \K$ such that $s = s^{(1)} + \ldots + s^{(q)}$.
             For $m = 1, \ldots, q$, send the share $s^{(m)}$ to every party in the set $S_m$.
        \end{myitemize}
    \item[--] {\bf Phase II --- Pairwise Consistency Checks}: For $m = 1, \ldots, q$, each party $P_i \in S_m$ does the following. 
        \begin{myitemize}
            \item[--] Upon receiving $s_i^{(m)}$ from $\D$, {\color{red} wait till the local time becomes a multiple of $\Delta$}. Then,
            send $s_i^{(m)}$ to every party $P_j \in S_m$. 
        \end{myitemize}
    \item[--] {\bf Phase III  --- Broadcasting Results of Pairwise Consistency Checks}: 
    For $m = 1, \ldots, q$, each party $P_i$ in $S_m$ does the following. 
        \begin{myitemize}
            \item[--] On receiving $s_j^{(m)}$ from any $P_j \in S_m$, {\color{red} wait till the local time becomes a multiple of $\Delta$}. Then,  
            do the following.
              \begin{myitemize}
              \item[--] If a share $s_i^{(m)}$ corresponding to $S_m$ has been received from $\D$, then, 
              broadcast $\OK(m, i, j)$ if $s_i^{(m)} = s_j^{(m)}$ holds. Else, broadcast $\NOK(m, i)$.
              \item[--] If $s_j^{(m)}$ and $s_k^{(m)}$ have been received from any $P_j$ and $P_k$ respectively, belonging to
              $S_m$ such that $s_j^{(m)} \neq s_k^{(m)}$, then broadcast $\NOK(m, i)$.
              \end{myitemize}
           \end{myitemize}
    \item[--] {\bf Local Computation --- Constructing Consistency Graphs}: Each party $P_i \in \PartySet$ does the following.
        \begin{myitemize}
            \item[--] For $m = 1, \ldots, q$, construct an undirected consistency graph $G_i^{(m)}$ over the parties in $S_m$, 
            where the edge $(P_j,P_k)$ is included in $G_i^{(m)}$ if $P_i$ has received $\OK(m, j, k)$ and $\OK(m, k, j)$
             from the broadcast of $P_j$ and $P_k$ respectively, either through regular or fallback mode. 
        \end{myitemize}
    \item[--] {\bf Phase IV --- Resolving Complaints and Broadcasting Core Sets Based On $\Z_s$}: 
    Each $P_i \in \PartySet$ (including $\D$) does the following at time $2\Delta + \TimeBCAST$.
        \begin{myitemize}
            \item[--] If $\NOK(m, j)$ is received from the broadcast of any $P_j \in S_m$ through regular-mode corresponding to any
            $m \in \{1, \ldots, q \}$, 
            then do the following:
              \begin{myitemize}
               \item[--] {\bf If $P_i = \D$}: Broadcast $\resolve(m, s^{(m)})$.
               \item[--] {\bf If $P_i \neq \D$}: Broadcast $\resolve(m, s_i^{(m)})$, provided $P_i \in S_m$ and $P_i$ has received
               $s_i^{(m)}$ from $\D$.
             \end{myitemize}
             \item[--] {\bf (If $P_i = \D$)}: For $m = 1, \ldots, q$, check if there exists some
             $\C_m \subseteq S_m$ which constitutes
             a clique in the graph $G_\D^{(m)}$, such that $S_m \setminus \C_m \in \Z_s$. 
             If $\C_1, \ldots, \C_q$ are found, then broadcast
             $\C_1, \ldots, \C_q$.
           \end{myitemize}    
    \item[--] {\bf Local Computation --- Verifying and Accepting Core sets}: Each party $P_i \in \PartySet$ (including $\D$) does the following
    at time $2\Delta + 2\TimeBCAST$.
        \begin{myitemize}
            \item[--] If $\C_1, \ldots, \C_q$ is received from the broadcast of $\D$ through the regular mode, 
            then {\it accept} these sets, if all the following hold.
            \begin{myitemize}
                \item[--] For $m = 1, \ldots, q$, the set $\C_m$ constitutes a clique in the consistency graph $G_i^{(m)}$ at time
                $2\Delta + \TimeBCAST$. In addition, $S_m \setminus \C_m \in \Z_s$.
                \item[--] For $m = 1, \ldots, q$, if $\NOK(m, j)$ was received from the broadcast of any
                 $P_j \in S_m$ through regular mode at time $2\Delta + \TimeBCAST$, then the following 
                 must hold true at time $2\Delta + 2\TimeBCAST$.
                \begin{myitemize}
                    \item[--] $\resolve(m, s^{(m)})$ is received from the broadcast of $\D$ through regular-mode.
                    \item[--] $\resolve(m, s^{(m)})$ is received from the broadcast of a set of parties $\C'_m$ through regular-mode,
                    where $\C'_m \subseteq \C_m$, and $\C_m \setminus \C'_m \in \Z_s$.
                \end{myitemize}
            \end{myitemize}
        \end{myitemize}
    \item[--] {\bf Phase V --- Deciding Whether Core Sets Based on $\Z_s$ have Been Accepted by Any Honest Party}:
     At time $2\Delta + 2\TimeBCAST$, each $P_i \in \PartySet$ participates in an instance of $\BA$ with input $b_i = 1$ if 
     it has accepted sets $\C_1, \ldots, \C_q$, else, with input $b_i = 0$, and {\color{red} waits for time $\TimeBA$}.  
    \item[--] {\bf Local Computation --- Computing Shares Through Core Sets Based on $\Z_s$}: 
    If the output of $\BA$ is 1, then each party $P_i \in \PartySet$ does the following.
      \begin{myitemize}
       \item[--] If $\C_1, \ldots, \C_q$ are not received yet, then {\color{red} wait to receive them}
        from the broadcast of $\D$ through
       fallback-mode. 
       \item[--] For $m = 1, \ldots, q$, compute the share $s^{(m)}_i$ corresponding to $S_m$ as follows, provided
       $P_i \in S_m$.
        \begin{myitemize}
            \item[--] If, at time $2\Delta + 2\TimeBCAST$,
             $\resolve(m, s^{(m)})$ was received from the broadcast of $\D$ and from a subset of parties $\C'_m \subseteq \C_m$ 
             {\bf through
             regular-mode},
            where $\C_m \setminus \C'_m \in \Z_s$, then output $s^{(m)}_i = s^{(m)}$.
            \item[--] Else, if a common value, say $s^{(m)}$, was received from a subset of parties $\C''_m \subseteq \C_m$ 
            {\bf at time 
            $2\Delta$} where $\C_m \setminus \C''_m \in \Z_s$, then output $s^{(m)}_i = s^{(m)}$.
            \item[--] Else wait till there exists a subset of parties $\C'''_m \subseteq \C_m$ where
            $\C_m \setminus \C'''_m \in \Z_a$, such that a common value, say
            $s^{(m)}$, is received from all the parties in $\C'''_m$. Upon finding such a $\C'''_m$, 
            output $s^{(m)}_i = s^{(m)}$.
       \end{myitemize}

      \end{myitemize}
    \item[--] {\bf Phase VI --- Broadcasting Core Sets Based on $\Z_a$}: If the output of $\BA$ is $0$, 
    then for $m = 1, \ldots, q$, dealer $\D$ does the following in its graph $G_\D^{(m)}$.
        \begin{myitemize}
            \item[--] Check if there exists a subset of parties
            $\E_m \subseteq S_m$, which constitutes a clique in the graph $G_\D^{(m)}$, such that
            $S_m \setminus \E_m \in \Z_a$. 
            \item[--] Upon finding $\E_1, \ldots, \E_q$, broadcast $\E_1, \ldots, \E_q$.
        \end{myitemize}
    \item[--] {\bf Local Computation --- Computing Shares Through Core Sets Based on $\Z_a$}:
     If the output of $\BA$ is 0, then each party $P_i \in \PartySet$ does the following.
          \begin{myitemize}
          \item[--] Participate in any instance of $\BCAST$ invoked by $\D$ for broadcasting sets $\E_1, \ldots, \E_q$, {\color{red} only after time}
          $2\Delta + 2\TimeBCAST + \TimeBA$. 
          \item[--] Wait till sets $\E_1, \ldots, \E_q$ are received from the broadcast of $\D$, and then {\it accept}
          these sets if they satisfy the following conditions.
             \begin{myitemize}
              \item[--] For $m = 1, \ldots, q$, the set $\E_m$ constitutes a clique in the graph $G_i^{(m)}$.
              \item[--] For $m = 1, \ldots, q$, the condition $S_m \setminus \E_m \in \Z_a$ holds.
             \end{myitemize}
           \item[--] If $\E_1, \ldots, \E_q$ are accepted, then
           compute the share $s_i^{(m)}$ corresponding to every $S_m$ where $P_i \in S_m$ as follows.
             \begin{myitemize}
             \item[--] If $P_i \in \E_m$, then output $s_i^{(m)}$ received from $\D$.
             \item[--] Else, wait till there exists a subset  $\E'_m \subseteq \E_m$ where
             $\E_m \setminus \E'_m \in \Z_s$, such that there exists a common values, say $s^{(m)}$, received
             from all the parties in $\E'_m$. Upon finding such an $\E'_m$, output $s^{(m)} _i= s^{(m)}$.
             \end{myitemize}
            \end{myitemize}
\end{myitemize}
\end{protocolsplitbox}

The properties of $\VSS$, stated in Theorem \ref{thm:VSS}, are proved in  Appendix \ref{app:VSS}. 
\begin{theorem}
\label{thm:VSS}
Let $\Z_s$ and $\Z_a$ be adversary structures, 
  satisfying the conditions $\Con$ (see Section \ref{sec:prelims}).
   Moreover, let $\SharingSpec = \{S_m : S_m = \PartySet \setminus Z_m \; \mbox{ and } \; Z_m \in \Z_s \}$.
    Then, $\VSS$ achieves the following, where $\D$ has input $s \in \K$ for $\VSS$.
   \begin{myitemize}
   \item[--] If $\D$ is honest, then the following hold.
         \begin{myitemize}
         \item[--] {\bf $\Z_s$-correctness}: In a synchronous network, $s$ is secret-shared with respect to $\SharingSpec$ at 
            time $\TimeVSS = 2\Delta + 2\TimeBCAST + \TimeBA$. 
         \item[--] {\bf $\Z_a$-correctness}: In an asynchronous network, almost-surely, $s$ is eventually secret-shared
          with respect to $\SharingSpec$.
         \item[--] {\bf Privacy}: Adversary's view remains independent of $s$ in any network.  
         \end{myitemize}   
   \item[--] If $\D$ is corrupt, then either no honest party obtains any output or there exists some
   $s^{\star} \in \K$, such that the following hold.\footnote{In the {\it best-of-both-worlds}
   setting, it is {\it not} necessary that the honest parties obtain an output within a known time-out
   in a {\it synchronous} network for a {\it corrupt} $\D$ (unlike the {\it commitment}
   property of traditional SVSS). This is because a {\it corrupt} $\D$ may not invoke the
   protocol and the parties will {\it not} be knowing the network type.} 
          \begin{myitemize}
          \item[--]  {\bf $\Z_a$-commitment}: In an asynchronous network, almost-surely, $s^{\star}$ is eventually secret-shared
          with respect to $\SharingSpec$.          
           \item[--]  {\bf $\Z_s$-commitment}: In a synchronous network, $s^{\star}$ is secret-shared with respect to $\SharingSpec$,
            such that the following hold.
              \begin{myitemize}
		       \item[--] If any honest party outputs its shares at time $\TimeVSS$, then all honest parties
		       output their shares at time $\TimeVSS$.
		       \item[--] If any honest party outputs its shares at time $T > \TimeVSS$,
		    then every honest party outputs its shares by time $T + 2\Delta$.  
	       \end{myitemize}   
          \end{myitemize}
   \item[--] {\bf Communication Complexity}: The protocol incurs a communication of
    $\Order(|\Z_s| \cdot n^4 (\log{|\K|} + \log{|\Z_s|} + \log n) + n^5 \log n)$ bits, and invokes one instance of $\BA$.
   \end{myitemize}
 \end{theorem}
\paragraph{\bf $\VSS$ for $L$ Secrets:} If $\D$ has $L$ secrets to share, then it can invoke
 $L$ independent instances of $\VSS$. However, instead of computing and broadcasting $L$ number of 
  $\C_1, \ldots, \C_q$ sets,
  $\D$ can compute and broadcast sets $\C_1, \ldots, \C_q$ once, for {\it all} the $L$ instances of 
  $\VSS$.\footnote{Such common $\C_1, \ldots, \C_q$ sets for all the $L$ instances are guaranteed for an {\it honest} $\D$.}
   The parties
   will need to execute a {\it single} instance of $\BA$ to decide whether $\D$ has broadcasted valid
   $\C_1, \ldots, \C_q$ sets, corresponding to {\it all} $L$ instances of $\VSS$.
   The resultant protocol will incur a 
   communication of
    $\Order(L \cdot |\Z_s| \cdot n^4 (\log{|\K|} + \log{|\Z_s|} + \log n) + n^5 \log n)$ bits and invokes one instance of $\BA$.
    To avoid repetition, we do not provide the formal details.

\section{The Preprocessing Phase Protocol}
\label{sec:preprocessing}
 Our protocol for the preprocessing phase allows the parties to generate secret-sharing
 of $c_M$ number of multiplication-triples, random for the adversary, with respect to the sharing specification
  $\SharingSpec = \{S_m : S_m = \PartySet \setminus Z_m \; \mbox{ and } \; Z_m \in \Z_s \}$.
    Before discussing the protocol, we discuss two sub-protocols used.
  \subsection{Agreement on a Common Subset (ACS)}
\label{sec:ACS}
 In the protocol $\ACS$, there exists a set
 $\Q \subseteq \PartySet$, such that it will be {\it guaranteed} that $\Z_s$ and $\Z_a$
   {\it either} satisfy the $\Q^{(1, 1)}(\Q, \Z_s, \Z_a)$ condition 
   {\it or} $\Q^{(3, 1)}(\Q, \Z_s, \Z_a)$ condition.  
  Moreover, each party in $\Q$ will have
  $L$ values, which it would like to secret-share using 
   $\VSS$.\footnote{Looking ahead, 
   in our {\it preprocessing phase} protocol, $\Q = S_l \cap S_m$ corresponding to some $S_l, S_m \in \SharingSpec$ and hence,
   the $\Q^{(1, 1)}(\Q, \Z_s, \Z_a)$ condition will be satisfied.
   In our MPC protocol, 
   $\Q = \PartySet$ and hence 
   the $\Q^{(3, 1)}(\Q, \Z_s, \Z_a)$ condition will be satisfied.}
   As {\it corrupt} dealers may {\it not} invoke their instances of $\VSS$, 
   the parties can compute outputs from {\it only} a subset of $\VSS$ instances
   corresponding to a subset of parties $\Q \setminus Z$, for some $Z \in \Z_s$
    ({\it even} in a {\it synchronous} network).
   However, in an {\it asynchronous} network, {\it different} parties may
   compute outputs from $\VSS$ instances of {\it different} subsets of  $\Q \setminus Z$ parties, corresponding to 
   {\it different}
   $Z \in \Z_s$.
   Protocol $\ACS$ allows the parties to agree on a {\it common} subset $\CoreSet$ of parties, where
   $\Q \setminus \CoreSet \in \Z_s$, such that {\it all} honest parties will be able to compute
   their outputs corresponding to the $\VSS$ instances of the parties in $\CoreSet$. 
   Moreover, in a {\it synchronous} network,  {\it all honest}
   parties from $\Q$ 
   are guaranteed to be present in $\CoreSet$.\footnote{This property will be very crucial in a {\it synchronous} network.}  
   Protocol $\ACS$ is obtained by generalizing the ACS protocol of \cite{ACC22a}, which was designed for {\it threshold}
   adversaries. The formal description of the protocol $\ACS$ and its properties 
      are available in Appendix \ref{app:Preprocessing}.  
  \subsection{The Multiplication Protocol}
  \label{sec:mult}
  Protocol $\Mult$ takes as input the secret-shared pairs of values $\{([a^{(\ell)}], [b^{(\ell)}]) \}_{\ell = 1, \ldots L}$, and securely generates
   $\{[c^{(\ell)}] \}_{\ell = 1, \ldots, L}$, where $c^{(\ell)} = a^{(\ell)} \cdot b^{(\ell)}$, without revealing any additional information
    to the adversary. 
    The protocol is obtained by ``combining" the {\it synchronous} multiplication protocol of \cite{Mau02}, with the 
   {\it asynchronous} multiplication protocol of \cite{CP20}, and adapting them to the best-of-both-worlds setting. 
     For simplicity, we discuss the idea of $\Mult$ for the case when $L = 1$. The modifications for a general 
     $L$ are straight forward.
   
    Let $[a]$ and $[b]$ be the inputs to the protocol. The goal is to securely compute $[c]$, where
   $c = a \cdot b$. For this, the parties securely compute a secret-sharing of each summand $[a]_l \cdot [b]_m$.
   A secret-sharing of $c$ can then be obtained by summing the secret-sharing of each summand $[a]_l \cdot [b]_m$.
   To generate a secret-sharing of the summand $[a]_l \cdot [b]_m$, the parties do the following:
   let $\Q_{l, m}$ be the set of parties who are guaranteed to have both the shares $[a]_l$, as well as $[b]_m$.
   Notice that $\Q_{l, m}$ is {\it not} empty and $\Z_s$ and $\Z_a$ satisfy the $\Q^{(1, 1)}(\Q_{l, m}, \Z_s, \Z_a)$ condition,
   since $\Z_s$ and $\Z_a$ satisfy the $\Q^{(3, 1)}(\PartySet, \Z_s, \Z_a)$ condition. Hence, {\it irrespective} of the network type,
   the set $\Q_{l, m}$ is {\it bound} to have {\it at least} one honest party. 
   Each party in the set $\Q_{l, m}$ is asked to independently secret-share the summand $[a]_l \cdot [b]_m$, 
   and the parties then agree on a common subset of parties $\R_{l, m}$ from $\Q_{l, m}$, where
   $\Q_{l, m} \setminus \R_{l, m} \in \Z_s$, which have shared some summand.
   For this, the parties execute an instance of the $\ACS$ protocol. The properties of $\ACS$ guarantees that
   in a {\it synchronous} network, {\it all honest} parties from $\Q_{l, m}$ are present in $\R_{l, m}$. On the other hand,
   even if the network is {\it asynchronous}, the set $\R_{l, m}$ is {\it bound} to have {\it at least} one honest party from 
   $\Q_{l, m}$. Hence, {\it irrespective} of the network type, it will be {\it guaranteed} that {\it at least} one
   party in $\R_{l, m}$ has secret-shared the summand $[a]_l \cdot [b]_m$. 
    However, since the exact identity of the honest parties in $\R_{l, m}$ is {\it not} known, the parties 
    check if {\it all} the parties in
    $\R_{l, m}$ have shared the same summand. 
    The idea here is that if {\it all} the parties in $\R_{l, m}$ have shared the same summand,
     then any of these secret-sharings can be taken as a secret-sharing of $[a]_l \cdot [b]_m$. Else,
     the parties publicly reconstruct the shares $[a]_l$ and $[b]_m$ and 
      compute a default secret-sharing of $[a]_l \cdot [b]_m$.
     Notice that in the latter case, the privacy of $a$ and $b$ is still preserved, as in this case,
     the set $\R_{l, m}$ consists of {\it corrupt} parties, who already know the values of both $[a]_l$ as well
     as $[b]_m$.
   
   Protocol $\Mult$ and its properties are available in Appendix \ref{app:Preprocessing}.
\subsection{The Preprocessing Phase Protocol}
Given protocols $\ACS$ and $\Mult$, the preprocessing phase protocol $\Offline$ is standard and straight forward.
The protocol has two stages. During the {\it first} stage, the parties  securely generate
  secret-sharing of $c_M$ pairs of random values. For this, 
   the parties run an instance of $\ACS$, where the input for each party
   will be $c_M$ pairs of random values.
    During the {\it second} stage, a secret-sharing of the product of each pair is computed securely
  by executing an instance of $\Mult$.
  Protocol $\Offline$ and its properties are available  in Appendix \ref{app:Preprocessing}.

\section{Best-of-both-Worlds Circuit-Evaluation Protocol}
\label{sec:MPC}
Given protocols $\Offline$ and $\ACS$, the circuit-evaluation protocol $\PiMPC$
  for evaluating $\ckt$ is standard and straight forward. Here, we outline the protocol steps and state its properties. We defer the full details to Appendix \ref{app:MPC}. 
    Protocol $\PiMPC$
    consists of four phases. 
  In the {\it first} phase, the parties generate secret-sharing of $c_M$ random multiplication-triples through $\Offline$. 
   Additionally, they invoke $\ACS$ to generate secret-sharing of their respective inputs for $f$, and 
     agree on a {\it common} subset of parties $\CoreSet$, where $\PartySet \setminus \CoreSet \in \Z_s$, such that the inputs
     of the parties in $\CoreSet$ are secret-shared. The inputs of the remaining parties are set to
     $0$. Note that in a {\it synchronous} network,
    {\it all honest} parties will be in $\CoreSet$. 
     In the {\it second} phase, the parties securely evaluate each gate in the circuit in a secret-shared fashion,
     after which the parties {\it publicly} reconstruct 
  the secret-shared output in the {\it third} phase. The {\it last} phase is the  {\it termination phase},
   where the parties check
    whether ``sufficiently many" parties have obtained the same output, in which case the parties
    ``safely" take that output
   and terminate the protocol (and all the underlying sub-protocols). 
    \begin{theorem}
\label{thm:MPC}
Let $\Adv$ be an adversary, characterized by adversary structures $\Z_s$ and $\Z_a$
 in a synchronous and asynchronous network respectively,
  satisfying the conditions $\Con$ (see Condition \ref{condition:Con} in Section \ref{sec:prelims}). 
 Moreover, let $f: \K^n \rightarrow \K$ be a function represented by an arithmetic circuit $\ckt$ over $\K$, consisting of
 $c_M$ number of multiplication gates, with a multiplicative depth of $D_M$, with each party having an input $x_i \in \K$.
 Then, $\PiMPC$ incurs a communication of 
 $\Order(c_M \cdot |\Z_s|^3 \cdot n^5 (\log{|\K|} + \log{|\Z_s|} + \log n) + |\Z_s|^2 \cdot n^6 \log n)$ bits, invokes $\Order(|\Z_s|^2 \cdot n)$ instances of $\BA$,
   and achieves the following. 
      \begin{myitemize}
      \item[--] In a synchronous network, all honest parties output $y = f(x_1, \ldots, x_n)$ at time $(30n + D_M + 6k + 38) \cdot \Delta$,
       where $x_j = 0$ for every $P_j \not \in \CoreSet$, such that
       $\PartySet \setminus \CoreSet \in \Z_s$,
       and every honest party is
        present in $\CoreSet$; here $k$ is the constant from Lemma \ref{lemma:ABAGuarantees}, as determined by the 
        protocol $\ABA$.
       \item[--] In an asynchronous network, almost-surely, the honest parties eventually output $y = f(x_1, \ldots, x_n)$ 
       where $x_j = 0$ for every $P_j \not \in \CoreSet$
       and where $\PartySet \setminus \CoreSet \in \Z_s$.
       \item[--] The view of $\Adv$ remains independent of the inputs of the honest parties in $\CoreSet$.                
   \end{myitemize}
\end{theorem}

\bibliographystyle{plain}

\bibliography{main}

\newpage

\begin{center}
{\Large \bf Supplementary Material}
\end{center}

\appendix

\section{Existing Primitives}
\label{app:ExistingPrimitives}
\subsection{Asynchronous Reliable Broadcast with Weak Synchronous Guarantees}
Let $\AdvStructure$ be an adversary structure such that $\AdvStructure$ satisfies the $\Q^{(3)}(\PartySet, \AdvStructure)$
 condition. The Acast protocol $\PiACast$ with respect to $\AdvStructure$ is presented in Fig \ref{fig:Acast}. The current 
  description of the protocol is taken from \cite{ACC22b}.
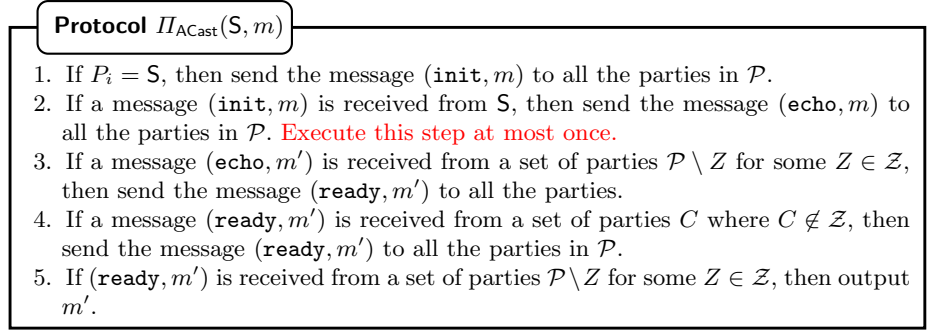
\begin{protocolsplitbox}{$\PiACast(\Sender, m)$}{The perfectly-secure Acast protocol with respect to an adversary structure 
 $\AdvStructure$ satisfying the $\Q^{(3)}(\PartySet, \AdvStructure)$ condition. The above code is executed by every
  $P_i \in \PartySet$}{fig:Acast}
	\justify
\begin{myenumerate}
    \item If $P_i = \Sender$, then send the message ($\init, m$) to all the parties in $\PartySet$.
      \item If a message ($\init, m$) is received from $\Sender$, then send the message ($\echo, m$)
        to all the parties in $\PartySet$. {\color{red} Execute this step at most once.}
    \item If a message ($\echo, m'$) is received from a set of parties 
        $\PartySet \setminus Z$ for some $Z \in \AdvStructure$, then send the message ($\Ready, m'$) to all the parties.
    \item If a message ($\Ready, m'$) is received from a set of parties $C$
     where $C \not \in \AdvStructure$, then send the message $(\Ready, m')$ to all
        the parties in $\PartySet$.
    \item If ($\Ready, m'$) is received from a set of parties 
         $\PartySet \setminus Z$ for some $Z \in \AdvStructure$, then output $m'$.
    \end{myenumerate} 
\end{protocolsplitbox}

We next prove the properties of the protocol $\PiACast$. The proofs for the case of asynchronous network are borrowed from \cite{ACC22b}.  \\~\\
{\bf Lemma \ref{lemma:Acast}.}
{Let $\Adv$ be an adversary characterized by an adversary structure $\AdvStructure$,  satisfying
 the $\Q^{(3)}(\PartySet, \AdvStructure)$ condition, such that $\Adv$ can corrupt any subset of parties from $\AdvStructure$
 during the execution of $\PiACast$.
 Then $\PiACast$ achieves the following.
 \begin{myitemize}
\item[--]{\it Asynchronous Network}: 
    \begin{myitemize}
        \item[--]{\bf $\AdvStructure$-Liveness}: If $\Sender$ is {\it honest}, then all honest parties eventually obtain an output.
         \item[--]{\bf $\AdvStructure$-Validity}: If $\Sender$ is {\it honest}, then every honest party with an output, outputs $m$.
         \item[--]{\bf $\AdvStructure$-Consistency}: If $\Sender$ is {\it corrupt} and some honest party outputs $m^{\star}$, 
         then every honest party {\it eventually} 	outputs $m^{\star}$.
       \end{myitemize}	
  \item[--] {\it Synchronous Network}:
        \begin{myitemize}
           \item[--]{\bf $\AdvStructure$-Liveness}: If $\Sender$ is {\it honest}, then all honest parties obtain an output within time $3\Delta$.
   	  \item[--]{\bf $\AdvStructure$-Validity}: If $\Sender$ is {\it honest}, then every honest party with an output, outputs $m$.
	  \item[--]{\bf $\AdvStructure$-Consistency}: If $\Sender$ is {\it corrupt} and some honest party outputs $m^{\star}$ at time $T$, then every honest $P_i$ outputs $m^{\star}$ by the
     end of time $T + 2 \Delta$.
       \end{myitemize}   
  \item[--] {\it Communication Complexity}: $\Order(n^2 \ell)$ bits are communicated, where $\Sender$'s message is of size $\ell$ bits.
\end{myitemize}
}
\begin{proof}
We first prove the properties assuming an {\it  asynchronous} network.
 Let $Z_c \in \AdvStructure$ be the set of parties {\it corrupted} by the adversary during the protocol execution, and
  let $\Honest \defined \PartySet \setminus Z_c$ be the set of {\it honest} parties.
  We start with the {\it validity} property, for which we consider an {\it honest} $\Sender$. We show that 
  all honest parties eventually output $m$. 
  This is because all honest parties complete steps $2-5$ in the protocol even if the corrupt parties do not send their messages.
   This follows from the fact that the messages of the parties in $\Honest$ are eventually delivered to all the honest parties,
   and $\PartySet \setminus \Honest = Z_c \in \AdvStructure$. 
    The parties in $Z_c$ may send 
    $\echo$ messages for $m'$, where $m' \neq m$.
   Similarly, the parties in $Z_c$ may send 
   $\ready$ messages for $m'$, where $m' \neq m$.
   However, since $Z_c \in \AdvStructure$, and since $\PartySet \setminus Z_c = \Honest \not \in \AdvStructure$, 
   no honest party ever generates a $\ready$ message for $m'$,  neither in step $3$, nor in step
    $4$.\footnote{If $\Honest \in \AdvStructure$, then $\AdvStructure$
   {\it does not} satisfy the $\Q^{(2)}(\PartySet, \AdvStructure)$ condition, which is a contradiction.} 
   This also proves the {\it liveness} property. 
   
   We next prove the {\it consistency} property for which
   we consider a {\it corrupt} $\Sender$. Let $P_h$ be an {\it honest} party who outputs $m^{\star}$.
    We have to show that all honest parties
   eventually obtain the output $m^{\star}$.
     Since $P_h$ obtained the output $m^{\star}$, it received $\Ready$ messages for $m^{\star}$ during step $5$ of the protocol
    from a set of parties $\PartySet \setminus Z$, for some $Z \in \AdvStructure$. 
       Let $\Honest^{(m^{\star})}$
        be the set of {\it honest} parties whose $\Ready$ messages are received by $P_h$ during step $5$. It is easy to see that
   $\Honest^{(m^{\star})} \not \in \AdvStructure$, as otherwise, $\AdvStructure$ {\it does not} satisfy
   the  $\Q^{(3)}(\PartySet, \AdvStructure)$ condition, which is a contradiction.
    The $\Ready$ messages of the parties in $\Honest^{(m^{\star})}$ are eventually delivered to every honest party
   and hence, {\it each} honest party (including $P_h$) eventually executes step $4$ and sends a
   $\Ready$ message for $m^{\star}$.
   It follows that the $\Ready$ messages of {\it all} the parties in $\PartySet \setminus Z_c$
    are eventually delivered
   to every honest party (irrespective of whether the parties in $Z_c$ send all the required
   messages) guaranteeing that all honest parties eventually obtain {\it some} output. 
   We wish to show that this output is $m^{\star}$.
   
    On the contrary, let
   $P_{h'} \neq P_h$ be an {\it honest} party who outputs
   $m^{\star \star} \neq m^{\star}$. This implies that
   $P_{h'}$ received $\Ready$ message for $m^{\star \star}$ from at least
    one {\it honest} party.
   From the protocol steps, it follows that an {\it honest} party generates a $\Ready$
   message for some potential $m$, only if it either receives $\echo$ messages for $m$ during step 3 
   from a set of parties $\PartySet \setminus Z$ for some $Z \in \AdvStructure$,
   or $\Ready$ messages for $m$ from a set of parties $C \not \in \AdvStructure$ during step 4.
   So, in order that a subset of parties $\PartySet \setminus Z$, for some $Z \in \AdvStructure$, eventually generates
   $\Ready$ messages for some potential $m$ during step 5, it must be
   the case that some {\it honest} party has received $\echo$ messages for $m$ during step 1 from a set of parties $\PartySet \setminus Z'$ for some
   $Z' \in \AdvStructure$, and has generated a
   $\Ready$ message for $m$.
   
    Since $P_h$ received the $\Ready$ message for $m^{\star}$ from at least one honest party, it must be the case that 
   some honest party has received $\echo$ messages for $m^{\star}$ from a set of parties $\PartySet \setminus Z_{\alpha}$,
    for some $Z_{\alpha} \in \AdvStructure$.
   Similarly, since $P_{h'}$ received a $\Ready$ message for $m^{\star \star}$ from at least one honest party, it must be the case that 
   some honest party has received $\echo$ messages for $m^{\star \star}$ from a set of parties 
    $\PartySet \setminus Z_{\beta}$, for some $Z_{\beta} \in \AdvStructure$.
   Let ${\cal T} = (\PartySet \setminus Z_{\alpha}) \cap (\PartySet \setminus Z_{\beta})$. Since $\AdvStructure$
   satisfies the $\Q^{(3)}(\PartySet, \AdvStructure)$ condition, it follows
   that $\AdvStructure$
   satisfies the $\Q^{(1)}({\cal T}, \AdvStructure)$ condition, and hence, ${\cal T}$ is guaranteed to have at least one {\it honest} party. 
   This further implies that there exists some honest party who generated an $\echo$ message for $m^{\star}$ as well
   as $m^{\star \star}$ during step 1, which is impossible. This is because an honest party executes step 1 at most once,
   and hence, generates an $\echo$ message at most once.
      
   The proofs of the properties in a {\it synchronous} network closely follow the proofs of the properties in the {\it asynchronous} network.
   Abusing the notation, let $Z_c \in \AdvStructure$ be the set of {\it corrupt} parties,
   and let $\Honest \defined \PartySet \setminus Z_c$ be the set of {\it honest} parties.
       If $\Sender$ is {\it honest}, then its $\init$ message for $m$ is delivered within time
   $\Delta$. As a result, every
   {\it honest} party sends an $\echo$ message for $m$ to all the parties, which is delivered within time $2\Delta$. 
   Consequently, every
   {\it honest} party sends a $\ready$ message for $m$ to all the parties, which is delivered within time
   $3 \Delta$. Since $\PartySet \setminus \Honest = Z_c \in \AdvStructure$, every honest
   party will receive the $\ready$ messages for $m$ from all the parties in $\Honest$ within time $3 \Delta$ and output $m$.
   This proves the {\it liveness} and {\it validity} in the synchronous network.
   
   If $\Sender$ is {\it corrupt}, and some honest party $P_h$ outputs $m^{\star}$ at time $T$, then it implies that
   $P_h$ has received $\ready$ messages for $m^{\star}$ during step $5$ of the protocol at time $T$ from a set of {\it honest}
   parties $\Honest^{(m^{\star})}$, such that $\Honest^{(m^{\star})} \not \in \AdvStructure$. 
   These ready messages are guaranteed to be received by every other honest party within time $T + \Delta$. Consequently, 
   every {\it honest} party who has not yet executed step $4$ will do so, and will send a $\ready$ message for $m^{\star}$ at time $T + \Delta$.
   Hence, 
   by the end of time $T + \Delta$, every honest party would have sent a $\ready$ message for $m^{\star}$ to every other honest party, which will be
   delivered within time $T + 2\Delta$.
   As a result, every honest party will output $m^{\star}$ latest at time $T + 2 \Delta$. This proves the {\it consistency}
   in the synchronous network.
   
   The communication complexity (both in a synchronous as well as asynchronous network) 
   simply follows from the fact that every honest party may need to send 
   an $\echo$ and  $\ready$ message for $m$ to every other party.
\end{proof}
\subsection{Reconstruction Protocol}
Let  $s \in \K$ be a value, which is secret-shared with respect to the sharing specification 
  $\SharingSpec = \{S_m : S_m = \PartySet \setminus Z_m \; \mbox{ and } \; Z_m \in \Z_s \}$. 
  Then protocol $\Rec$ (Fig \ref{fig:Rec}) allows the parties to publicly reconstruct $s$.
  
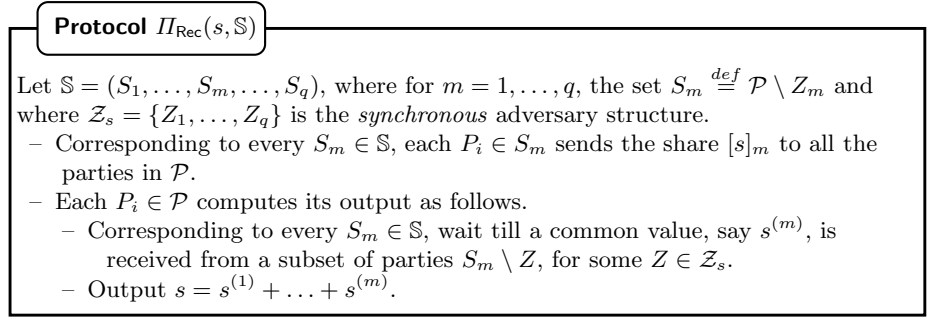
  \begin{protocolsplitbox}{$\Rec(s, \SharingSpec)$}{Protocol for publicly reconstructing a secret-shared value}{fig:Rec}
Let $\SharingSpec = (S_1, \ldots, S_m, \ldots, S_q)$, where for $m = 1, \ldots, q$,
 the set $S_m \defined \PartySet \setminus Z_m$ and where $\Z_s = \{Z_1, \ldots, Z_q \}$
    is the {\it synchronous} adversary structure. 
   \begin{myitemize}  
     \item[--] Corresponding to every $S_m \in \SharingSpec$, each $P_i \in S_m$ sends the share $[s]_m$ to all the parties 
     in $\PartySet$.
     \item[--] Each $P_i \in \PartySet$ computes its output as follows.
       \begin{myitemize}
        \item[--] Corresponding to every $S_m \in \SharingSpec$, wait till a common value, say $s^{(m)}$, is received from a subset
        of parties $S_m \setminus Z$, for some $Z \in \Z_s$. 
        \item[--] Output $s = s^{(1)} + \ldots + s^{(m)}$.
       \end{myitemize}
    \end{myitemize}
\end{protocolsplitbox}

We next prove the properties of the protocol $\Rec$.
\begin{lemma}
\label{lemma:Rec}
Let  $s \in \K$ be a value, which is secret-shared with respect to the sharing specification 
  $\SharingSpec = \{S_m : S_m = \PartySet \setminus Z_m \; \mbox{ and } \; Z_m \in \Z_s \}$. 
  Then protocol $\Rec$ achieves the following with a communication complexity of
  $\Order(|\Z_s| \cdot n^2 \log{|\K|})$ bits.
   \begin{myitemize}
   \item[--] {\bf Synchronous network}: All honest parties output $s$ within time $\Delta$.   
   \item[--] {\bf Asynchronous network}: All honest parties eventually output $s$.
   \end{myitemize}
\end{lemma}
\begin{proof}
The communication complexity follows easily from the protocol steps, since corresponding to each $S_m \in \SharingSpec$,
  {\it every} party in $S_m$ needs to send its version of the share $[s]_m$ to {\it all} the parties in $\PartySet$.
 
 Next, consider an {\it asynchronous} network and let $Z^{\star} \in \Z_s$ be the set of {\it corrupt} parties. Consider an {\it arbitrary} 
  $S_m \in \SharingSpec$ and an {\it arbitrary honest} party $P_i$. We show that $P_i$ eventually sets $s^{(m)}$ to $[s]_m$.
  This will imply that $P_i$ eventually outputs $s$. Let $\Honest_m \defined S_m \setminus Z^{\star}$ be the set of
  {\it honest} parties in $S_m$. All the parties in $\Honest_m$ send $[s]_m$ to $P_i$, which are eventually delivered to $P_i$.
  Moreover, $S_m \setminus \Honest_m \in \Z_s$. Hence, it is confirmed that $P_i$ eventually sets
  $s^{(m)}$ to some value. To complete the proof, we need to show that this values will be the same as $[s]_m$.
  So, let there exist some $Z_{\alpha} \in \Z_s$, such that $P_i$ received a common value from all the parties in $S_m \setminus Z_{\alpha}$,
  and sets $s^{(m)}$ to this common value. Since $\Z_s$ satisfies the $\Q^{(3)}(\PartySet, \Z_s)$ condition, it follows that
  $\Honest_m \cap (S_m \setminus Z_{\alpha}) \neq \emptyset$. Thus there exists at least one {\it honest} party in the set
  $S_m \setminus Z_{\alpha}$, who would have sent $[s]_m$ to $P_i$, implying that $P_i$ sets $s^{(m)}$ to $[s]_m$.
  
  The proof of the lemma statement in a {\it synchronous} network is exactly the same as above. Moreover, all honest parties will output
  $s$ within time $\Delta$, since the shares of all honest parties will be delivered within time $\Delta$.
\end{proof}
\subsection{Beaver's Protocol}
  Let $u$ and $v$ be two values from $\K$, which are secret-shared with respect to the sharing specification
  $\SharingSpec = \{S_m : S_m = \PartySet \setminus Z_m \; \mbox{ and } \; Z_m \in \Z_s \}$.
  Moreover, let $(a, b, c) \in \K^3$ be a multiplication-triple, where $a, b$ and $c$ are secret-shared with respect to
  $\SharingSpec$ and where $c = a \cdot b$.  Then, Beaver's protocol for securely computing $[u \cdot v]$, is given in Fig \ref{fig:Beaver}.
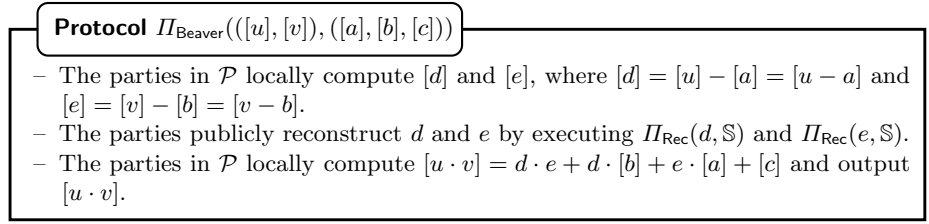
\begin{protocolsplitbox}{$\BatchBeaver(([u],[v]), ([a],[b],[c]))$}{Beaver's protocol for securely computing a secret-sharing of the product of
  two secret-shared values}{fig:Beaver}
	\justify
    \begin{myenumerate}
    \item[--] The parties in $\Partyset$ locally compute $[d]$ and $[e]$, where $[d] =
        [u] - [a] = [u - a]$ and $[e] = [v] - [b] = [v - b]$.
    \item[--] The parties publicly reconstruct $d$ and $e$ by executing $\Rec(d, \SharingSpec)$ and
        $\Rec(e, \SharingSpec)$.
    \item[--] The parties in $\Partyset$ locally compute $[u\cdot v] = d\cdot e +
        d\cdot[b] + e\cdot[a] + [c]$ and output $[u \cdot v]$.
    \end{myenumerate}
\end{protocolsplitbox}
The properties of the protocol $\BatchBeaver$ are stated in Lemma \ref{lemma:BatchBeaver}.
\begin{lemma} \label{lemma:BatchBeaver}
  Let $\Adv$ be an adversary, characterized by an adversary structure $\Z_s$ in a synchronous network and adversary structure $\Z_a$ in an asynchronous network, satisfying the conditions $\Con$ (see Condition \ref{condition:Con} in Section \ref{sec:prelims}).
  Moreover, let $u$ and $v$ be two values from $\K$, which are secret-shared with respect to the sharing specification
  $\SharingSpec = \{S_m : S_m = \PartySet \setminus Z_m \; \mbox{ and } \; Z_m \in \Z_s \}$.
  Furthermore, let $(a, b, c) \in \K^3$ be a multiplication-triple, where $a, b$ and $c$ are secret-shared with respect to
  $\SharingSpec$ and where $c = a \cdot b$. 
   Then protocol $\BatchBeaver$
    achieves the following.
    \begin{myenumerate}
        \item {\bf $\Z_s$-correctness}: In a synchronous network, all honest parties output $[w]$ within time $\Delta$, where
        $w = u \cdot v$.
        \item {\bf $\Z_a$-Correctness}:  In an asynchronous network, 
        all honest parties eventually output $[w]$, where
        $w = u \cdot v$.
        \item {\bf Privacy}: If $(a, b, c)$ is random from the point of view of the adversary, then 
        the view of $\Adv$ is distributed independent
            of $u$ and $v$, irrespective of the network type.
        \item {\bf Communication Complexity:} The protocol incurs a
            communication of $\Order(|\Z_s| \cdot n^2 \log{|\K|})$ bits.
                \end{myenumerate}
\end{lemma}
\begin{proof}
Note that $u \cdot v = ((u - a) + a) \cdot ((v - b) + b) = de + db + ea + c$ holds, 
 where $d = u - a$ and $e = v - b$. From the {\it linearity} property of the secret-sharing,
  the parties will be able to locally compute $[d]$ and $[e]$. Moreover, from the properties of $\Rec$,
   all honest parties will be able to compute $d$ and $e$ in a {\it synchronous} network, within time
   $\Delta$. On the other hand, in an {\it asynchronous} network, the honest parties will eventually compute
   $d$ and $e$. The $\Z_s$-correctness and $\Z_a$-correctness now follow easily. 
   
   The privacy
is argued as follows: the only step where the parties communicate is during the
reconstruction of $d$ and $e$. As $d = u - a$, if $a$ is random for 
$\Adv$, then even after learning $d$, in the view of $\Adv$, the value of $u$ remains as secure as it was before. Similarly, if $b$ is random, then even after learning $e$, in the view of $\Adv$, the
value $v$ remains as secure as before. 

The communication complexity follows from the communication
complexity of $\Rec$. \end{proof}

\section{Synchronous BA with Asynchronous Guarantees}
\label{app:SBA}
Let $\AdvStructure$ be an adversary structure
 satisfying the $\Q^{(3)}(\PartySet, \AdvStructure)$
 condition. We present a protocol $\SBA$ (Fig \ref{fig:SBA}) which is $\AdvStructure$-perfectly-secure
  in a {\it synchronous} network, and which has $\AdvStructure$-guaranteed liveness in an {\it asynchronous} network. 
   The protocol is obtained by generalizing the SBA protocol of \cite{BGP89} which was designed against {\it threshold} adversaries
   to tolerate $t < n/3$ corruptions.
   
   Let $\Kings \subseteq \Partyset$ denote a predetermined set of publicly-known parties, such that
   $\Kings \not\in \AdvStructure$. Such a set $\Kings$ always exists, as the set $\PartySet$ trivially constitutes a
   candidate $\Kings$ set. Protocol $\SBA$ proceeds in phases,
   where in each phase, a {\it unique} party from $\Kings$ is publicly designated as the {\it king}.
   Hence, the number of phases is $|\Kings|$. The general idea behind the protocol is very simple.
   In each phase, 
   the parties attempt to check if {\it all honest} parties have the {\it same} bit.
   Then, depending upon the answer, the parties do one of the following.   
   \begin{myitemize}
    \item[{\bf (1)}:] If all honest parties are found to hold a common bit, then the parties ``stick" to that bit for all subsequent phases,
    {\it irrespective} of the king parties.
   \item[{\bf (2)}:] Else, the parties take the ``help" of the king-party of this phase
    so that if the king-party is {\it honest}, then, at the end of this phase, all honest parties have the same bit
   \end{myitemize}
   If all the honest parties start the protocol with a 
  {\it common} input bit, say $b$, then throughout the protocol, the parties retain the bit $b$, and finally output $b$, thus ensuring
   {\it validity}.
  On the other hand, {\it consistency} is guaranteed because there will be {\it at least} one phase where the corresponding 
  king-party is guaranteed to be {\it honest}, and hence, at the end of that phase, all honest parties will have a common bit, which
  they retain till the end of the protocol.    
   \begin{protocolsplitbox}{$\SBA$}{A perfectly-secure SBA protocol with asynchronous guaranteed liveness.
    The above code is executed by every $P_i \in \PartySet$}{fig:SBA}
  \justify
  Let $\Kings \subseteq \Partyset$ be a publicly-known predetermined set of king-parties, such that
 $\Kings \not\in \AdvStructure$. 
  For simplicity and without loss of generality, let $\Kings = \{P_1, \ldots, P_{|\Kings|} \}$.
  \begin{myitemize}
  \item On having the input $b_i$, initialize the {\it preference bit} $\pref_i = b_i$.
  \item For $k = 1, \ldots, |\Kings|$, do the following in phase $k$, where $P_{k} \in \Kings$ is the predetermined
   designated king-party for phase $k$:
    \begin{myenumerate}
    \item Send the current preference bit $\pref_i$ to every party in $\Partyset$.
    \item If some preference bit $b$ is received from a set of parties $\Partyset \setminus Z$ for some $Z \in \AdvStructure$, then send $(\mathsf{propose}, b)$ to every party in $\Partyset$.
      \begin{myitemize}
      \item[--] If $(\mathsf{propose}, b)$ is received from some set of parties $C \not\in \AdvStructure$, then update the preference bit $\pref_i = b$
      \end{myitemize}
    \item If $P_i = P_{k}$, then send $(\mathsf{king}, k, \pref_k)$ to every party in
      $\Partyset$.
    \item Setting the preference bit for the next phase:
      \begin{myitemize}
      \item[--] If, during step 2, 
       $(\mathsf{propose}, b)$ was received from a set of parties $\Partyset \setminus Z$ for some $Z \in \AdvStructure$, then set $\pref_i = b$.
      \item[--] Else, set $\pref_i = \pref_k$, where $(\mathsf{king}, k, \pref_k)$
       was received from the king-party during step 3.
      \end{myitemize}
    \end{myenumerate}
    \item Output $\pref_i$.
  \end{myitemize}
\end{protocolsplitbox}
\justify
Before proving the properties of $\SBA$, we prove some helping lemmas which will be useful later.
\begin{lemma}
  \label{lem:sba-same-bit}
  In a synchronous network,
  if all honest parties hold the same preference bit $b$ at the beginning of a phase $k$ in $\SBA$, then
  they retain the same bit $b$ at the end of phase $k$.
\end{lemma}
\begin{proof}
  Let $Z_c \in \AdvStructure$ be the set of corrupt parties and $\Honest \defined \PartySet \setminus Z_c$ be
   the set of {\it honest} parties. Since every party in $\Honest$ holds
  the same preference bit $b$ at the beginning of phase $k$, they will all send $b$ to each
  party in step 1 of phase $k$.

  Thus, each party in $\Honest$ receives the bit $b$ from the set of parties $\Honest$, where
  $\Honest = \PartySet \setminus Z_c$.  
  The parties in $Z_c$ may send the bit $\bar{b}$ as their preference bit.
  However, since $\Z$ satisfies the $\Q^{(3)}(\PartySet, \AdvStructure)$ condition,
   $\Honest \not \in \AdvStructure$. Consequently, each party in $\Honest$ 
    sends $(\mathsf{propose}, b)$ to each party in
   step 2 of phase $k$.
   
  Thus, each party in $\Honest$ receives $(\mathsf{propose}, b)$ from the set of parties
  $\Honest$ during step 2 of phase $k$. Now due to the same reasons as above, 
  it follows that every party in $\Honest$ sets $b$ as their preference bit at the end of phase $k$.
\end{proof}
\begin{lemma}
  \label{lem:sba-honest-propose}
  In a synchronous network, if any honest party $P_i \in \Partyset$ proposes a bit $b$
   in step 2 of phase $k$ in $\SBA$, then no other honest party
  $P_j \in \Partyset$ proposes a bit $\bar{b}$.
\end{lemma}
\begin{proof}
  We prove the lemma through a contradiction. So, let there exist two honest parties $P_i, P_j \in \Partyset$, such that 
  $P_i$ sends $(\mathsf{propose}, b)$, and $P_j$ sends $(\mathsf{propose}, \bar{b})$ during 
  step $2$ of phase $k$. This implies that
   $P_i$ must have received the preference bit $b$ during step $1$ of
  phase $k$ from a set of parties $\Partyset \setminus Z_{\alpha}$, for some $Z_{\alpha} \in \AdvStructure$.
  Similarly, $P_j$ must have received the preference bit $\bar{b}$ during step 1 of phase $k$
  from a set of parties $\Partyset \setminus Z_{\beta}$, for some $Z_{\beta} \in \AdvStructure$. Let
  $\mathcal{T} = (\Partyset \setminus Z_{\alpha}) \cap (\Partyset \setminus Z_{\beta})$. Since
  $\AdvStructure$ satisfies the $Q^{(3)}(\Partyset, \AdvStructure)$ condition, it follows that
  $\AdvStructure$ satisfies the $Q^{(1)}(\mathcal{T}, \AdvStructure)$ condition, and hence, there
  exists at least one {\it honest} party in $\mathcal{T}$, say $P_k$. 
  Hence, $P_k$ must
  have sent $b$ as its 
  preference bit to $P_i$ and $\bar{b}$ as its preference bit to $P_j$ during step $1$ of  
  phase $k$, which is impossible.
\end{proof}
\begin{lemma}
  \label{lem:sba-honest-king}
  In a synchronous network, for any phase $k$, if the designated king-party $P_{k}$ is honest, 
   then all honest parties have the
  same preference bit at the end of phase $k$.
\end{lemma}
\begin{proof}
  Since $P_{k}$ is {\it honest}, it sends the {\it same} $(\mathsf{king}, k, \pref_k)$ message to {\it all} the 
  parties during the step $3$ of phase $k$. 
    To prove the lemma, we consider two cases that are possible for the honest parties
   during step 4 of phase $k$:
  \begin{myitemize}
  \item {\bf Case I --- All honest parties set $\pref_k$ as their preference bit}:
   In this case, the lemma holds trivially.
  \item {\bf Case II --- Some honest $P_i$ does not set $\pref_k$ as its preference bit}:
    In this case, during step $4$, party $P_i$ must have set its preference bit to $b$, where
    the bit $b$ is proposed to $P_i$ during step $2$ by a set of
    parties $\Partyset \setminus Z$, for some $Z \in \AdvStructure$.
    Then, consider the set $(\Partyset \setminus Z) \setminus Z_c = \PartySet \setminus (Z \cup Z_c)$,
    where $Z_c \in \AdvStructure$ is the set of {\it corrupt}
    parties. The set $\Partyset \setminus (Z \cup Z_c)$ is {\it non-empty}, consisting of {\it only
    honest} parties and 
    $\Partyset \setminus (Z \cup Z_c) \not \in \AdvStructure$. This is because
    the set $\AdvStructure$ satisfies the $\Q^{(3)}(\PartySet, \AdvStructure)$ condition.
    
    Now, during step $2$ of phase $k$,
    party $P_{k}$ will receive the $(\mathsf{propose}, b)$ message from 
    all the parties in  $\Partyset \setminus (Z \cup Z_c)$. Moreover,  
    from Lemma \ref{lem:sba-honest-propose}, it follows that any honest party
    {\it outside} $\Partyset \setminus (Z \cup Z_c)$ will {\it never} send
    a $(\mathsf{propose}, \bar{b})$ message to $P_{k}$ during step $2$ of phase $k$.
    Furthermore, even though the parties in $Z_c$ may send a $(\mathsf{propose}, \bar{b})$ message to $P_{k}$ during 
    step $2$ of phase $k$, these messages {\it will not} be considered by $P_{k}$ to determine $\pref_k$, since
    $Z_c \in \AdvStructure$.
    Based on the above arguments, it follows that $P_{k}$ would set $\pref_k = b$ at the end of step $2$ of phase $k$.
    Since $P_{k}$ is {\it honest}, it sends $b$ as $\pref_k$ during the step $3$ of phase $k$. 
    Thus, even in this case, every honest
    party sets the same preference at the end of phase $k$.
  \end{myitemize}
\end{proof}
\justify
We now prove the properties of the protocol $\SBA$ in a {\it synchronous} network.
\begin{lemma}
  \label{lem:sba-proof}
  Let $\Adv$ be an adversary characterised by an adversary structure $\AdvStructure$,
  satisfying the $Q^{(3)}(\Partyset, \AdvStructure)$ condition, such that $\Adv$ can corrupt
  any subset of parties from $\AdvStructure$ during the execution of $\SBA$. Then, $\SBA$
  achieves the following in a synchronous network.
  \begin{myitemize}
  \item[--] {\bf $\AdvStructure$-Liveness}: All honest parties obtain an output within time $3n \cdot \Delta$. 
  \item[--] {\bf $\AdvStructure$-Validity}: If all honest parties have input $b$, then they output $b$.
  \item[--] {\bf $\AdvStructure$-Consistency}: All honest parties output the same value.
  \item[--] {\bf Communication Complexity}: The protocol incurs a communication of
  $\mathcal{O}(n^3)$ bits.\footnote{If the inputs of the parties are of size $\ell$ bits, then
   the parties can run $\ell$ independent instances of the protocol, one corresponding to each
   bit of their input. This will incur a communication complexity of $\Order(n^3 \ell)$ bits.}
  \end{myitemize}
\end{lemma}
\begin{proof}
  Each phase involves $3$ rounds of communication and consequently, each phase requires
   $3\Delta$ time. The 
   $\AdvStructure$-liveness then follows from the fact that all the parties
   will output something after $|\Kings|$ phases, where $|\Kings| \leq n$.

  If all honest parties hold the same input $b$, then at the start of phase $1$, they will
  all set their preference bit as $b$. Then from repeated application of 
   Lemma \ref{lem:sba-same-bit},
    it follows that the honest parties retain the bit $b$ as their
  preference bit throughout at the end of each phase.
  Thus, $\SBA$ achieves $\AdvStructure$-validity.

For $\AdvStructure$-consistency, let there exist a phase $k \in \{1, \ldots, |\Kings| \}$, such that the designated
 king-party $P_{k}$ of this phase is {\it honest}. From Lemma
  \ref{lem:sba-honest-king},  it follows that at the end of phase $k$, every honest party sets
  the same preference bit. Then, from repeated application of Lemma \ref{lem:sba-same-bit},
   it follows that the honest parties retain the same preference bit till the end of the protocol
   and output that bit. To complete the proof for consistency, we need to show
     that there exists at least one {\it honest} $P_{k} \in \Kings$. 
     However, this follows from the fact that  $\Kings \not\in \AdvStructure$. 
     
     In each phase, $\Order(n^2)$ bits are exchanged and so overall $\Order(|\Kings| \cdot n^2)$ bits are exchanged. 
     The communication complexity then follows, since $|\Kings| \leq n$.
    \end{proof}
\justify
\begin{remark}[{\bf $\AdvStructure$-guaranteed liveness for $\SBA$ in an asynchronous network}]
For any given adversary structure $\AdvStructure$, the largest possible set $\Kings$ is the
set $\Partyset$ of all parties, as $\Partyset \not\in \AdvStructure$. So the maximum time taken by
 $\SBA$ to generate an output (in a {\it synchronous} network) 
  is $3n \cdot \Delta$. Thus, to achieve  $\AdvStructure$-guaranteed liveness in an {\it asynchronous} network, 
   it suffices to execute $\SBA$
   till time $3n \times \Delta$ and 
  then output $\bot$, if no output can be deduced from $\SBA$.
    Therefore, $\SBA$ achieves $\AdvStructure$-guaranteed liveness even in an
  asynchronous network.
\end{remark}

\section{Properties of the Best-of-Both-Worlds BA Protocol}
\label{app:HBA}
We first start with the properties of the protocol $\BCAST$ (see Fig \ref{fig:BCAST} for the formal steps of the protocol).\\~\\
\noindent {\bf Theorem \ref{thm:BCAST}.}
{\it Let $\Adv$ be an adversary characterized by an adversary structure $\AdvStructure$ satisfying the $\Q^{(3)}(\PartySet, \AdvStructure)$
 condition. Moreover, let $\Sender$ have input
   $m \in \{0, 1 \}^{\ell}$ for $\BCAST$. Then,
  $\BCAST$ achieves the following
   with a communication complexity of $\Order(n^3 \ell)$ bits,
    where $\TimeBCAST =  3\Delta + \TimeSBA$ and $\TimeSBA \leq 3n \cdot \Delta$.
 \begin{myitemize}
   \item[--] {\it Synchronous} network: 
        \begin{myitemize}
           \item[--]{\bf (a) $\AdvStructure$-Liveness}: At time $\TimeBCAST$, each honest party has an output. 
	    \item[--] {\bf (b) $\AdvStructure$-Validity}: If $\Sender$ is {\it honest}, then at time $\TimeBCAST$, each honest party
    outputs $m$.
           \item[--] {\bf (c) $\AdvStructure$-Consistency}: If $\Sender$ is {\it corrupt}, then the output of every honest party is the same at 
      time $\TimeBCAST$.     
            \item[--]   {\bf (d) $\AdvStructure$-Fallback Consistency}: If $\Sender$ is {\it corrupt} and some honest 
     party outputs $m^{\star} \neq \bot$ at time
    $T$ through fallback-mode, then every honest party outputs $m^{\star}$ by 
    time $T + 2\Delta$.    
      \end{myitemize}
\item[--] {\it Asynchronous Network}:
   \begin{myitemize}
     \item[--] {\bf (a) $\AdvStructure$-Liveness}: At time $\TimeBCAST$, each honest party has an output.
    \item[--] {\bf (b) $\AdvStructure$-Weak Validity}: If $\Sender$ is {\it honest}, then at  time $\TimeBCAST$, 
    each honest party
    outputs $m$ or $\bot$.
    \item[--] {\bf (c) $\AdvStructure$-Fallback Validity}: If $\Sender$ is {\it honest}, then each honest party
     with output $\bot$ at time $\TimeBCAST$, eventually outputs 
    $m$ through fallback-mode.
    \item[--] {\bf (d) $\AdvStructure$-Weak Consistency}: If $\Sender$ is {\it corrupt}, then at time $\TimeBCAST$, 
    each honest party
    outputs a common
     $m^{\star} \neq \bot$ or $\bot$.
    \item[--] {\bf (e) $\AdvStructure$-Fallback Consistency}: If $\Sender$ is {\it corrupt} and some honest party 
    outputs $m^{\star} \neq \bot$ at  time
    $T$ where $T \geq \TimeBCAST$, then 
    each honest party eventually outputs $m^{\star}$.
   \end{myitemize}
   \end{myitemize}
}
\begin{proof}
The {\it $\AdvStructure$-liveness} property follows from the fact that 
 every honest party outputs something (including $\bot$) at (local)
  time $\TimeBCAST$, irrespective of the type of the network.
   We next prove the  rest of the properties of the protocol in the {\it synchronous} network.
 
 If $\Sender$ is {\it honest}, then due to the {\it $\AdvStructure$-liveness} and
  {\it $\AdvStructure$-validity} properties of $\PiACast$ in the {\it synchronous} network, all honest
  parties receive $m$ from the Acast of $\Sender$
   at time $3\Delta$. Consequently, all honest parties participate with input $m$ in the instance of $\SBA$.
  The {\it $\AdvStructure$-guaranteed liveness} and {\it $\AdvStructure$-validity} properties
   of $\SBA$ in the {\it synchronous} network guarantees that at time $3\Delta + \TimeSBA$, all
    honest parties will have
  $m$ as the output from the instance
  of $\SBA$. As a result, all honest parties output $m$ at time $\TimeBCAST$, thus proving the {\it $\AdvStructure$-validity} property.
  
  To prove the {\it $\AdvStructure$-consistency} property, we consider a {\it corrupt} $\Sender$.
  From the {\it $\AdvStructure$-consistency} property of $\SBA$ in the {\it synchronous} network, 
  all honest parties will have the {\it same} output from the instance of $\SBA$ at time $\TimeBCAST$.
    If all honest parties have the output $\bot$ for $\BCAST$ at time $\TimeBCAST$, then {\it $\AdvStructure$-consistency}
     holds trivially. So, consider the case when
    some {\it honest} party, say $P_i$, has the output $m^{\star} \neq \bot$ for $\BCAST$ at time $\TimeBCAST$. 
    This implies that all honest parties have the output $m^{\star}$ from the instance of $\SBA$. 
    Moreover, at time $3\Delta$, at least one {\it honest} party, say $P_h$, has received $m^{\star}$ from the Acast of $\Sender$. 
    If the latter does not hold, then all honest parties would have participated
     with input $\bot$ in the instance of $\SBA$, and from the {\it $\AdvStructure$-validity} of $\SBA$
      in the {\it synchronous}
    network, all honest  parties would compute $\bot$ as the output during the instance of $\SBA$, which is a contradiction.
    Since $P_h$ has received
     $m^{\star}$ from $\Sender$'s Acast at time $3\Delta$, it follows from the {\it $\AdvStructure$-consistency} property of
      $\PiACast$ in the {\it synchronous}
    network that {\it all} honest parties will receive
     $m^{\star}$ from $\Sender$'s Acast by time $5\Delta$. Moreover, $5\Delta < 3\Delta + \TimeSBA$ holds.
    Consequently, at time $3 \Delta + \TimeSBA$,
     {\it all} honest parties will have $m^{\star}$ from $\Sender$'s Acast {\it and} as the output of
    $\SBA$, implying that all honest parties output $m^{\star}$ for $\BCAST$.   
    
    We next prove the  {\it $\AdvStructure$-fallback consistency} property for which
    we again consider a {\it corrupt} $\Sender$. 
    Let $P_h$ be an {\it honest} party who outputs $m^{\star} \neq \bot$ at time $T$ through fallback-mode.
    Note that $T > \TimeBCAST$, as the output during the fallback-mode is computed only after time $\TimeBCAST$.  
     We also note that {\it each} honest party has output $\bot$ at 
    time $\TimeBCAST$. This is because, from the proof of the {\it $\AdvStructure$-consistency} property of 
    $\BCAST$ (see above), if any {\it honest} party has an output 
    $m' \neq \bot$ at time $\TimeBCAST$, then {\it all} honest parties (including $P_h$) must have computed the output 
    $m'$ at time $\TimeBCAST$. Hence, $P_h$ will never change its output to
    $m^{\star}$.\footnote{Recall that in the protocol $\BCAST$, the parties who obtain
    an output different from $\bot$ at time $\TimeBCAST$, never change their output.}  
    Now since $P_h$ has obtained the output $m^{\star}$, it implies that at time $T$, it has received $m^{\star}$ from the Acast of
    $\Sender$. It then follows from the {\it $\AdvStructure$-consistency} of 
    $\PiACast$ in the {\it synchronous} network that every honest party will also receive 
    $m^{\star}$ from the Acast of $\Sender$, latest by time $T + 2 \Delta$ and output $m^{\star}$.
    This completes the proof of all the properties in the {\it synchronous} network.
   
  We next prove the properties of the protocol $\BCAST$ in an {\it asynchronous} network.   
    The {\it $\AdvStructure$-weak validity} property follows from the {\it $\AdvStructure$-validity} property of 
    $\PiACast$ in the {\it asynchronous} network, which ensures that
    no honest party ever receives an $m'$ from the Acast of $\Sender$, where $m' \neq m$.
    So, if at all any honest party outputs a value different from $\bot$ at time $\TimeBCAST$, it has to be $m$.
       The {\it $\AdvStructure$-weak consistency} property follows using similar arguments
   as used to prove {\it $\AdvStructure$-consistency} in the {\it synchronous} network;
   however we now rely on the {\it $\AdvStructure$-validity} and 
     {\it $\AdvStructure$-consistency}
     properties of $\PiACast$ in the asynchronous network. 
     The latter property ensures  that
    for a {\it corrupt} $\Sender$, two different honest parties never end up receiving $m_1$ and $m_2$ from the Acast of $\Sender$, where $m_1\neq m_2$. 
    
    For the {\it $\AdvStructure$-fallback validity} property, consider an {\it honest} $\Sender$, and let 
     $P_i$ be an arbitrary {\it honest} party
      who outputs $\bot$ at (local) time $\TimeBCAST$. Since the parties keep on participating in the protocol beyond time $\TimeBCAST$,
       it follows from the {\it $\AdvStructure$-liveness} and
        {\it $\AdvStructure$-validity} properties of $\PiACast$ in the {\it asynchronous} network that party $P_i$ will {\it eventually} 
     receive $m$ from the Acast of $\Sender$, by executing the steps of the fallback-mode of $\BCAST$. 
     Consequently, party $P_i$ eventually changes its output from
     $\bot$ to $m$.
     
     For the {\it $\AdvStructure$-fallback consistency} property, we consider a {\it corrupt} $\Sender$. Let 
     $P_j$ be an {\it honest} party who outputs some $m^{\star}$ different from $\bot$ at time $T$, where $T \geq \TimeBCAST$.
      This implies that $P_j$ has
      obtained $m^{\star}$ from the Acast of $\Sender$.
           Now, consider an arbitrary {\it honest} $P_i$. From the
            {\it $\AdvStructure$-liveness} and {\it $\AdvStructure$-weak consistency} properties of 
            $\BCAST$ in {\it asynchronous} network proved above,
     it follows that $P_i$ outputs either $m^{\star}$ or $\bot$ at local time $\TimeBCAST$. 
     If $P_i$ has output $\bot$, then from the {\it $\AdvStructure$-consistency}
     property of $\PiACast$ in the {\it asynchronous} network, it follows that 
     $P_i$ will also eventually obtain $m^{\star}$ from the Acast of $\Sender$, by executing the steps of the fallback-mode of $\BCAST$. 
     Consequently, party $P_i$ eventually changes its output from
     $\bot$ to $m^{\star}$.
     
     The {\it communication complexity} (both in the synchronous as well as asynchronous network)
      follows from the communication complexity of $\SBA$ and $\PiACast$.
\end{proof}
\subsection{Properties of the Protocol $\BA$}
In this section, we prove the properties of our best-of-both-worlds
  BA protocol $\BA$ (see Fig \ref{fig:BA} for the formal description of the protocol). \\~\\
\noindent {\bf Theorem \ref{thm:BA}.}
{\it Let $\Adv$ be an adversary, characterized by an adversary structure $\AdvStructure$, satisfying the $\Q^{(3)}(\PartySet, \AdvStructure)$
 condition. Moreover, let $\ABA$ be an ABA protocol 
   satisfying the conditions as stated in Lemma \ref{lemma:ABAGuarantees}.
   Then $\BA$ achieves the following.
\begin{myitemize}
\item[--] {\bf Synchronous Network}: The protocol is a {\it $\AdvStructure$-perfectly-secure} SBA 
 protocol, where all honest parties obtain an output within time $\TimeBA = \TimeBCAST + \TimeABA$.
  The protocol incurs a communication of $\Order(|\AdvStructure| \cdot n^5 \log|\F| + n^6 \log n)$ bits.
\item[--] {\bf Asynchronous Network}: The protocol is a {\it $\AdvStructure$-perfectly-secure} ABA protocol 
with an expected communication of  $\Order(|\AdvStructure| \cdot n^7 \log|\F| + n^8 \log n)$ bits.
\end{myitemize}
}
\begin{proof}
We start with the properties in a {\it synchronous} network. The $\AdvStructure$-liveness property of $\BCAST$ in the {\it synchronous}
 network guarantees
 that all honest parties will have some output, from each instance of $\BCAST$, at time $\TimeBCAST$. Moreover, 
  the {\it $\AdvStructure$-validity} and {\it $\AdvStructure$-consistency} properties of $\BCAST$ in the {\it synchronous} network
   guarantee that irrespective of the sender parties,
  {\it all} honest parties will have a common output from each individual instance of $\BCAST$, at time $\TimeBCAST$.
  Now since the parties decide their respective inputs for the instance of $\ABA$ {\it deterministically} based on the individual outputs
  from the $n$ instances of $\BCAST$ at time $\TimeBCAST$, it follows that 
  all honest parties participate with a {\it common} input in the protocol $\ABA$. Hence, all honest parties obtain an output by the end of
  time $\TimeBCAST + \TimeABA$, thus ensuring {\it $\AdvStructure$-guaranteed liveness} of $\BA$. 
  Moreover, the {\it $\AdvStructure$-consistency} property of $\ABA$ in the {\it synchronous} network
   guarantees that all honest parties have a {\it common} output from the instance of $\ABA$, which is taken as the output of
   $\BA$,  thus
  proving the {\it $\AdvStructure$-consistency} of $\BA$.

    For proving the validity property in a synchronous network, let all {\it honest} parties have the same input bit $b$. 
      Let $Z_c \in \AdvStructure$ be the set of {\it corrupt} parties and $\Honest \defined \PartySet \setminus Z_c$
    be the set of {\it honest} parties. 
    From the 
    {\it $\AdvStructure$-validity} of $\BCAST$ in the {\it synchronous} network,
    all honest parties will receive $b$ as the output at time $\TimeBCAST$ in all the 
    $\BCAST$ instances corresponding to the sender parties in $\Honest$. 
  Since $\AdvStructure$ satisfies the $\Q^{(3)}(\PartySet, \AdvStructure)$
    condition, it follows that $\Honest \not \in \AdvStructure$.
    Consequently, all honest parties will find a {\it common} subset $\R$ in the protocol,
    as the set $\Honest$ constitutes a candidate $\R$. Moreover, $\Honest \subseteq \R$.
    Furthermore, $\R \setminus \Honest \subseteq Z_c \in \AdvStructure$ and
    $\R \setminus Z_c \subseteq \Honest \not \in \AdvStructure$.
    Hence, all honest parties $P_i$ will find a {\it common} subset $R_i \subseteq \R$, as per the protocol, such that
    $b$ is computed as the output during the $\BCAST$ instances of all the parties in $\R_i$.
    As a result, all honest parties will participate with input $b$ in the
    instance of $\ABA$ and hence, output $b$ at the end of $\ABA$, which follows from the
   {\it $\AdvStructure$-validity} of $\ABA$ in the {\it synchronous} network. This proves the {\it $\AdvStructure$-validity} of $\BA$.
   
   We next prove the properties of $\BA$ in an {\it asynchronous} network. 
    The {\it $\AdvStructure$-consistency} of the protocol $\BA$ follows from the {\it $\AdvStructure$-consistency} of the protocol
    $\ABA$ in the {\it asynchronous} network, since the overall output of the protocol $\BA$ is same as the output
    of the protocol $\ABA$. The {\it $\AdvStructure$-liveness} of the protocol $\BCAST$ in the {\it asynchronous}
    network guarantees that all honest parties will have some output from all the $n$ instances of $\BCAST$
    at local time $\TimeBCAST$. Consequently, all honest parties will participate with some input 
    in the instance of $\ABA$. The {\it $\AdvStructure$-almost-surely liveness} of
    $\ABA$ in the {\it asynchronous} network then implies the
    {\it $\AdvStructure$-almost-surely liveness} of
    $\BA$.
    
    For proving the validity in an {\it asynchronous} network, let all {\it honest} parties have the same input bit $b$. 
      Let $Z_c \in \AdvStructure$ be the set of {\it corrupt} parties and $\Honest \defined \PartySet \setminus Z_c$
    be the set of {\it honest} parties. We claim that all {\it honest} parties participate with the input
    $b$ during the instance of $\ABA$. The {\it $\AdvStructure$-validity} of $\ABA$ in the {\it asynchronous} network then
    automatically implies the  {\it $\AdvStructure$-validity} of $\BA$. 
    
    To prove the above claim, consider an arbitrary
    {\it honest} party $P_h$. There are two possible cases. If $P_h$ fails to find a subset $\R$ satisfying the
    protocol conditions, then the claim holds trivially, as $P_h$ participates in the instance of $\ABA$ with its input for
    $\BA$, which is the bit $b$. So, consider the case when $P_h$ finds a subset $\R$ such that
    $\PartySet \setminus \R \in \AdvStructure$, and where, corresponding to each $P_j \in \R$, 
    party $P_h$ has computed an output $b_h^{(j)} \in \{0, 1 \}$ at local time $\TimeBCAST$ during the instance $\BCAST^{(j)}$.
    Now, consider the subset of {\it honest} parties $\Honest \cap \R$ in the set $\R$.
     It follows that $\R \setminus (\Honest \cap \R) \subseteq Z_c \in \AdvStructure$.
     Also, since $\AdvStructure$ satisfies the $\Q^{(3)}(\PartySet, \AdvStructure)$ condition, it follows that
    $(\Honest \cap \R) \not \in \AdvStructure$.
    Moreover, $P_h$ will compute the output $b$ at local time $\TimeBCAST$ in the instance of $\BCAST$ corresponding to
    {\it every} $P_j \in (\Honest \cap \R)$, which follows from the {\it $\AdvStructure$-weak validity} of $\BCAST$ in the {\it asynchronous} network.
    From these arguments, it follow that $P_h$ will find a 
    candidate subset $\R_h$, where $\R \setminus \R_h \in \AdvStructure$ and where
    $b$ is computed as the output  at local time $\TimeBCAST$ in the instance of $\BCAST$, corresponding to
    {\it every} $P_j \in \R_h$. This is because the subset of parties $\Honest \cap \R$ constitutes a candidate
    $\R_h$. Consequently, $P_h$ will set $b$ as its input for the instance of $\ABA$, thus proving the claim.
    
   The communication complexity, both in a synchronous as well as in an asynchronous network, follows easily from the 
   protocol steps and from the communication complexity of $\SBA$ and $\ABA$.
\end{proof}

\section{Properties of the Best-of-Both-Worlds VSS Protocol}
\label{app:VSS}
In this section we prove the properties of the protocol $\VSS$ (see Fig \ref{fig:VSS} for the formal description of the protocol).
 Recall that we want to prove the properties of $\VSS$ assuming an adversary $\Adv$ characterized by an adversary structure $\Z_s$ in a   
  {\it synchronous} network, and an adversary structure $\Z_a$ in an {\it asynchronous} network, satisfying the following conditions.
  \begin{myitemize}
  \item[--] $\Z_s \neq \Z_a$;
  \item[--] For every subset $Z \in \Z_a$, there exists a subset $Z' \in \Z_s$, such that $Z \subseteq Z'$;
  \item[--] $\Z_s$ and $\Z_a$ satisfy the $\Q^{(3, 1)}(\PartySet, \Z_s, \Z_a)$ condition.
  \end{myitemize}
 Moreover, we are considering the sharing specification
  $\SharingSpec = \{S_m : S_m = \PartySet \setminus Z_m \; \mbox{ and } \; Z_m \in \Z_s \}$.
  Note that the above conditions
   automatically imply that $\Z_s$ satisfies the $\Q^{(3)}(\PartySet, \Z_s)$ condition
  and $\Z_a$ satisfies the $\Q^{(4)}(\PartySet, \Z_a)$ condition.
  Before proving the properties, we prove a related property, which will later be useful while proving the properties of $\VSS$.
\begin{lemma}
\label{lemma:GenVSSSync}
In protocol $\VSS$ if the network is synchronous and if the output of $\BA$ is $1$, then the following hold.
\begin{myitemize}
    \item[--] All honest parties participated in the instance of $\BA$ with input $1$.
    \item[--] Corresponding to every $S_m \in \SharingSpec$, all honest parties in $\C_m$ received a common share
     from $\D$ by time $\Delta$.
\end{myitemize}
\end{lemma}
\begin{proof}
Let the network be {\it synchronous} and let all the honest parties output $1$ during the instance of $\BA$. 
 This implies at least one {\it honest} party, say $P_h$, has participated in $\BA$ with input $1$. If this is not the case and if {\it all}
 honest parties participated with input $0$ in $\BA$, then from the $\Z_s$-validity of $\BA$ in the {\it synchronous} network,
 all honest parties would compute the output $0$ during the instance of $\BA$, which is a contradiction.
 
 Since $P_h$ has participated with input $1$ in the instance of $\BA$, it implies that $P_h$ has received the sets
  $\C_1, \ldots, \C_q$ from the broadcast of $\D$ at time $2\Delta + 2\TimeBCAST$ through regular-mode and accepted these sets.
  From the $\Z_s$-validity and $\Z_s$-consistency of $\BCAST$ in the {\it synchronous} network, every other honest party would
  also receive these sets from the broadcast of $\D$ at time $2\Delta + 2\TimeBCAST$. 
  Since $P_h$ has accepted $\C_1, \ldots, \C_q$, it 
   implies that $P_h$ has checked that all the following hold.
       \begin{myitemize}
                \item[--] For $m = 1, \ldots, q$, the set $\C_m$ constitutes a clique in the consistency graph $G_h^{(m)}$ of $P_h$ at time
                $2\Delta + \TimeBCAST$. That is, for every $P_i, P_j \in \C_m$, the messages $\OK(m, i, j)$ and
                $\OK(m, j, i)$ have been received from the broadcast of $P_i$ and $P_j$ respectively by time 
                $2\Delta + \TimeBCAST$. By the $\Z_s$-validity and $\Z_s$-consistency of $\BCAST$ in the {\it synchronous} network,
                these messages are also received by every other honest party by time $2\Delta + \TimeBCAST$. 
                Consequently, $\C_1, \ldots, \C_q$ will constitute a clique 
                in the consistency graphs $G_k^{(1)}, \ldots, G_k^{(m)}$ respectively of {\it every} honest
                party at time $2\Delta + \TimeBCAST$.                 
                \item[--] For $m = 1, \ldots, q$, the condition
                $S_m \setminus \C_m \in \Z_s$ holds. It is easy to see that {\it every} honest party will find that this condition
                holds.
                \item[--] For $m = 1, \ldots, q$, if $\NOK(m, j)$ was received from the broadcast of any
                 $P_j \in S_m$ through regular mode at time $2\Delta + \TimeBCAST$, then the following 
                 hold at time $2\Delta + 2\TimeBCAST$.
                \begin{myitemize}
                    \item[--] $\resolve(m, s^{(m)})$ is received from the broadcast of $\D$ through regular-mode.
                    \item[--] $\resolve(m, s^{(m)})$ is received from the broadcast of a subset of parties $\C'_m$ through regular-mode,
                    where $\C'_m \subseteq \C_m$ and $\C_m \setminus \C'_m \in \Z_s$.
                \end{myitemize}
                From the $\Z_s$-validity and $\Z_s$-consistency of $\BCAST$ in the {\it synchronous} network, the above conditions
                will be also satisfied for {\it every} honest party.
                In more detail, if any $\NOK(m, j)$ was received by $P_h$ from the broadcast of any
                 $P_j \in S_m$ through regular-mode at time $2\Delta + \TimeBCAST$, then the same 
                 $\NOK(m, j)$ message would be received through regular-mode at time 
                $2\Delta + \TimeBCAST$ by {\it every} honest party. 
                Due to a similar reason, any corresponding $\resolve(m, s^{(m)})$ message which is received by 
                $P_h$ through regular-mode at time $2\Delta + 2\TimeBCAST$, either from the broadcast of $\D$ or from the broadcast of
                any party in $\C_m$, will also be received by {\it every} honest party.
            \end{myitemize}
It thus follows that the conditions
 for accepting the $\C_1, \ldots, \C_q$ will hold for {\it every} honest party at time $2\Delta + 2\TimeBCAST$
  and so, {\it all} honest parties participate with input $1$ during the instance of $\BA$. This proves the first part of the lemma.
  
  We now prove the second part of the lemma. The statement is obviously true, if $\D$ is {\it honest}. 
  So, we consider a {\it corrupt} $\D$. Let $S_m$ be an arbitrary set in $\SharingSpec$.
    We first note that all the {\it honest} parties in $\C_m$ received a common share, say $s^{(m)}$, from $\D$.
    This is because, from the proof of the first part of the lemma, there exists some {\it honest} party $P_h$, who
    has received $\C_m$ from the broadcast of $\D$ through regular-mode at time $2\Delta + 2\TimeBCAST$
    and has accepted $\C_m$. And while accepting $\C_m$, party $P_h$ has verified that
    $\C_m$ constitutes a clique in its consistency graph $G_h^{(m)}$ at time $2\Delta + \TimeBCAST$.
    Hence, the messages $\OK(m, i, j)$ and $\OK(m, j, i)$ have been received by $P_h$ from the broadcast of {\it every honest}
    $P_i, P_j \in \C_m$ by time $2\Delta + \TimeBCAST$. This automatically implies that $s_i^{(m)} = s_j^{(m)}$ holds, where
    $s_i^{(m)}$ and $s_j^{(m)}$ denotes the shares received from $\D$ by $P_i$ and $P_j$ respectively, corresponding to
    $S_m$. We wish to show that both $P_i$ and $P_j$ would have received their respective shares within time $\Delta$.
    
    On the contrary, let $P_i$ receive $s^{(m)}_i$ from $\D$ at time $\Delta + \delta$, where $\delta > 0$. From the protocol steps, $P_i$ 
    starts performing pairwise consistency checks only when its local time becomes a multiple of $\Delta$. Hence, 
    $P_i$ must have started sending its share $s^{(m)}_i$ to the other parties at time $c \cdot \Delta$, where $c > 1$. 
    Similarly, from the protocol steps, $P_j$ will start broadcasting
    the  $\OK(m, j, i)$ message, {\it only} at time $(c+1) \cdot \Delta$, since it waits till its local time becomes a multiple of $\Delta$,
    before broadcasting any $\OK$ or $\NOK$ messages.
    Now, by the $\Z_s$-validity of $\BCAST$ in the {\it synchronous} network, 
    it takes $\TimeBCAST$ time for the $\OK(m, j, i)$ messages
     to be received by {\it any} honest party. 
     Hence, the edge $(P_i, P_j)$ gets added in the consistency graph
     $G_k^{(m)}$ of {\it every} honest party 
      only at time $(c + 1) \cdot \Delta + \TimeBCAST$. However, this is a contradiction, since 
      $(c + 1) \cdot \Delta + \TimeBCAST > 2\Delta + \TimeBCAST$.
\end{proof}

We now proceed to prove the properties of the protocol $\VSS$. We start with the {\it correctness} property in the
 {\it synchronous} network.
\begin{lemma}
\label{lemma:VSSSynchronousCorrectness}
In protocol $\VSS$, if $\D$ is honest and participates with input $s$, then in a synchronous network,
   $s$ is secret-shared, with respect to the sharing specification $\SharingSpec$,
    at 
   time $\TimeVSS = 2\Delta + 2\TimeBCAST + \TimeBA$. 
\end{lemma}
\begin{proof}
Let $Z_s^{\star} \in \Z_s$ be the set of {\it corrupt} parties,
 and let $\Honest_s \defined \PartySet \setminus Z_s^{\star}$ be the set of {\it honest} parties.
  We show that corresponding to {\it every} $S_m \in \SharingSpec$, all {\it honest} parties in $S_m$
  output the share $s^{(m)}$ at time $\TimeVSS$, where $s^{(m)}$ is the share picked by $\D$, corresponding to 
  $S_m$. The lemma then follows from the fact that since $\D$ is {\it honest}, it selects the shares
  $s^{(1)}, \ldots, s^{(q)}$, satisfying the condition $s^{(1)} + \ldots + s^{(q)} = s$. So, consider an {\it arbitrary}
  $S_m \in \SharingSpec$, and let $\Honest_m = S_m \setminus Z_s^{\star}$ be the set of {\it honest}
  parties in $S_m$.
  
During Phase I, {\it every} $P_i \in \Honest_m$ receives the share $s^{(m)}_i$ from $\D$ within time $\Delta$, where
 $s^{(m)}_i = s^{(m)}$. During Phase II, {\it every} $P_i \in \Honest_m$
 sends $s_i^{(m)}$ to {\it every} $P_j \in S_m$,
  which takes at most $\Delta$ time to be delivered.  Hence, by time $2\Delta$, {\it every} $P_i \in \Honest_m$
  receives $s^{(m)}_j$ from {\it every} $P_j \in \Honest_m$, such that $s^{(m)}_i = s^{(m)}_j$ holds. 
 Consequently, during Phase III, {\it every} $P_i \in \Honest_m$ 
 broadcasts $\OK(m, i, j)$ corresponding to {\it every} $P_j \in \Honest_m$, and vice versa. From the $\Z_s$-validity
 of $\BCAST$ in the {\it synchronous} network,
  it follows that {\it all} the parties in $\Honest_m$ receive the 
   $\OK(m, i, j)$ and $\OK(m, j, i)$ messages through regular-mode at time $2\Delta + \TimeBCAST$, 
   from the broadcast of {\it every}  $P_i \in \Honest_m$ and {\it every} 
   $P_j \in \Honest_m$.
    Hence, corresponding to {\it every} $P_i, P_j \in \Honest_m$, 
    the edge $(P_i, P_j)$ will be added to the graph $G_k^{(m)}$ of {\it every}
    $P_k \in \Honest_m$. Furthermore, from the $\Z_s$-consistency
     property of $\BCAST$ in the {\it synchronous} network, 
     the graph $G_k^{(m)}$ will be the {\it same} for {\it every} $P_k \in \Honest_m$ (including $\D$) at time $2\Delta + \TimeBCAST$. 
     
     From the above arguments, the set of parties in $\Honest_m$ will constitute a clique in the consistency graph of {\it all} the parties
     in $\Honest_m$. Moreover, $S_m \setminus \Honest_m \subseteq Z_s^{\star} \in \Z_s$.
     Hence, during 
     Phase IV, $\D$ will be able to find a candidate $\C_m$ set and broadcast it, which will be received by {\it all} the
      parties in $\Honest_s$ through 
     regular-mode within time $2\Delta + 2\TimeBCAST$ (follows from the $\Z_s$-validity of $\BCAST$ in the {\it synchronous}
     network).
     Now, consider an {\it arbitrary honest} party $P_k$, who receives the message 
     $\NOK(m, j)$ from a party $P_j \in S_m$ through regular-mode at time $2\Delta + \TimeBCAST$.
    From the $\Z_s$-validity and $\Z_s$-consistency properties of $\BCAST$ in the {\it synchronous}
     network, this  $\NOK(m, j)$ message will be received by {\it all} the parties in $\Honest_m$ through regular-mode
     at time $2\Delta + \TimeBCAST$. Consequently, {\it all} the parties in $\Honest_m$ (including $\D$) will respond
     by broadcasting the $\resolve(m, s^{(m)})$ message, which will be received by {\it all} the {\it honest} parties 
     through regular mode at time $2\Delta + 2\TimeBCAST$. 
     Since $\C_m \setminus \Honest_m \subseteq Z_s^{\star} \in \Z_s$, it follows that the conditions for accepting 
     $\C_m$ will hold for {\it all} the honest parties at time $2\Delta + 2\TimeBCAST$. Consequently, {\it all honest}
     parties will participate with input $1$ in the instance of $\BA$ and from the
     $\Z_s$-validity and $\Z_s$-guaranteed liveness properties 
      of $\BA$ in the {\it synchronous} network, all honest parties will compute the output
     $1$ in the instance of $\BA$ at time $2\Delta + 2\TimeBCAST + \TimeBA$. 
     
     Finally, we show that {\it all} the parties in $\Honest_m$ output $s^{(m)}$
     at time $2\Delta + 2\TimeBCAST + \TimeBA$. For this, we 
      consider the following possible cases with respect to $\C_m$.       
    \begin{enumerate}
        \item {\bf An $\NOK(m, \star)$ message was received at time $2\Delta + \TimeBCAST$ through regular-mode from
        the broadcast of some party in $S_m$}: Since {\it all honest} parties accept $\C_m$, it implies that 
        corresponding to this $\NOK$ message, all honest parties have received a  $(\resolve, s^{(m)})$ message
        from the broadcast of $\D$ through regular-mode, as well as from the broadcast of a subset of parties $\C'_m \subseteq \C_m$
        through regular-mode, where
        $\C_m \setminus \C'_m \in \Z_s$, at time $2\Delta + 2\TimeBCAST$. Hence, according to the protocol steps, 
        {\it every honest} party in $S_m$ outputs $s^{(m)}$ as the share, corresponding to
        $S_m$. Also, it is easy to see that the honest parties output $s^{(m)}$ as the share at time 
        $2\Delta + 2\TimeBCAST + \TimeBA$. 
        \item {\bf No $\NOK(m, \star)$ message was received through regular-mode
        from the broadcast of any party in $S_m$ within time $2\Delta + \TimeBCAST$}:
        From 
         Lemma \ref{lemma:GenVSSSync}, {\it all} parties in $\Honest_m \cap \C_m$ would have received the common share
         $s^{(m)}$ from 
         $\D$, corresponding to the set $S_m$, within time $\Delta$. 
         As part of the pairwise consistency test, the share $s^{(m)}$ from all the parties in $\Honest_m \cap \C_m$
          would have been delivered 
         to all the {\it honest} parties in $S_m$ within time $2\Delta$. 
         This implies that within time $2\Delta$, {\it all} honest parties would have received $s^{(m)}$ 
         from a subset of parties $\C''_m \subseteq \C_m$, where $\C_m \setminus \C''_m \in \Z_s$.
         This is because the set $(\C_m \cap \Honest_m)$ definitely constitutes a candidate $\C''_m$.
         Moreover, {\it no} honest party would have ever received a value {\it different} from 
         $s^{(m)}$ within time $2\Delta$, from any party in $S_m$. 
         On the contrary,  if any {\it honest} $P_i$ receives  $s_j^{(m)}$ from $P_j$ and 
          $s_k^{(m)}$ from $P_k$ within time $2\Delta$, where $P_j, P_k \in S_m$ and where
         $s_j^{(m)} \neq s_k^{(m)}$, then $P_i$ would have broadcasted an $\NOK(m, i)$ message at time
         $2\Delta$, which would have been received by all honest parties through regular-mode at time $2\Delta + \TimeBCAST$,
         which is a contradiction.
         
         Hence, in this case also, 
        {\it every honest} party in $S_m$ outputs $s^{(m)}$ as the share, corresponding to
        $S_m$. Also, 
         it is easy to see that the honest parties output $s^{(m)}$ as the share at time 
        $2\Delta + 2\TimeBCAST + \TimeBA$.        
    \end{enumerate}
\end{proof}   
We next prove the correctness property in an {\it asynchronous} network.
\begin{lemma}
\label{lemma:VSSAsynchronousCorrectness}
In protocol $\VSS$, if $\D$ is honest and participates with input $s$, then in an asynchronous network, almost-surely 
 $s$ is eventually secret-shared with respect to the sharing specification $\SharingSpec$.
\end{lemma}
\begin{proof}
Let $Z_a^{\star} \in \Z_s$ be the set of {\it corrupt} parties,
 and let $\Honest_a \defined \PartySet \setminus Z_a^{\star}$ be the set of {\it honest} parties.
  We show that corresponding to {\it every} $S_m \in \SharingSpec$, almost-surely,
   all {\it honest} parties in $S_m$
  eventually output the share $s^{(m)}$, where $s^{(m)}$ is the share picked by $\D$ corresponding to 
  $S_m$. The $\Z_a$-correctness then follows from the fact that since $\D$ is {\it honest}, it selects the shares
  $s^{(1)}, \ldots, s^{(q)}$, satisfying the condition $s^{(1)} + \ldots + s^{(q)} = s$. So consider an {\it arbitrary}
  $S_m \in \SharingSpec$ and let $\Honest_m = S_m \setminus Z_a^{\star}$ be the set of {\it honest}
  parties in $S_m$. 
  
  We first note that {\it every} honest party participates with {\it some} input in the instance of $\BA$ at local time $2\Delta + 2 \TimeBCAST$. 
  Hence from the $\Z_a$-almost-surely liveness and $\Z_a$-consistency properties of $\BA$ in the {\it asynchronous} network, 
  all honest parties eventually compute a common output, during the instance of $\BA$. Now there are two possible cases
   with respect to the output of $\BA$, according to which the parties proceed to compute their shares.
\begin{myenumerate}
\item {\bf The output of $\BA$ is 1}: 
  Since the output of $\BA$ is $1$, due to the $\Z_a$-validity of $\BA$ in the {\it asynchronous} network, at least one honest
    party, say $P_h$, has received the sets $\C_1, \ldots, \C_q$ from the broadcast of $\D$ 
    through regular-mode, within time $2\Delta + \TimeBCAST$ and accepted these sets.
    It then follows that all honest parties also receive the sets $\C_1, \ldots, \C_q$ from the broadcast of $\D$, either through
    regular-mode or through fallback-mode. 
    This follows from the $\Z_a$-weak validity and $\Z_a$-fallback validity properties of $\BCAST$
    in the {\it asynchronous} network.
    Since $P_h$ has accepted $\C_m$, it has verified that $\C_m$ constitutes a clique
    in the consistency graph $G_h^{(m)}$. This implies that corresponding to
    $S_m$, {\it all honest} parties in $\C_m$ have received the {\it same} share from $\D$, which is $s^{(m)}$, since we are
    considering an {\it honest} $\D$.
    We will show that all {\it honest} parties in $S_m$ 
    eventually output $s^{(m)}$ as the share, corresponding to $S_m$. 
    
    So consider an {\it arbitrary honest} $P_i \in S_m$.
        From the protocol steps, if the output of $\BA$ is $1$,
        then $P_i$ computes its share corresponding to $S_m$, based on one of the following three conditions.
        Assuming that at least one of these conditions eventually hold for $P_i$, we
        first show that the share computed by $P_i$ corresponding to $S_m$, is bound to be
        $s^{(m)}$. This is followed by showing that indeed at least one of these conditions eventually hold for $P_i$.        
      \begin{myitemize}
            \item[--] {\bf Condition A}: At time $2\Delta + 2\TimeBCAST$,
            party $P_i$ received the $\resolve(m, s^{(m)})$ message
              from the broadcast of $\D$, as well as from the broadcast of a subset of parties $\C'_m \subseteq \C_m$ through
             regular-mode,
            where $\C_m \setminus \C'_m \in \Z_s$.
            Clearly, in this case, $P_i$ outputs $s^{(m)}$ as the share corresponding to $S_m$.          
            \item[--] {\bf Condition B}: At time $2\Delta$, there exists a subset of parties $\C''_m \subseteq \C_m$ 
            where $\C_m \setminus \C''_m \in \Z_s$, such that $P_i$ received a common value from {\it all}
            the parties in $\C''_m$. We claim that the subset $\C''_m$ is bound to contain at least
            one {\it honest} party from $\C_m$, which would have sent 
            $s^{(m)}$ to $P_i$, due to which $P_i$ will output $s^{(m)}$ as the share corresponding to $S_m$.
              In more detail, let $S_m \setminus \C_m = Z_{\alpha} \in \Z_s$, and
            $\C_m \setminus \C''_m = Z_{\beta} \in \Z_s$.
            Also, note that $\PartySet \setminus S_m = Z_m \in \Z_s$.
            Now, if $\C''_m$ {\it does not} contain any honest party from $\C_m$, it implies that
            $\C''_m \subseteq Z_a^{\star} \in \Z_a$.
            This further implies that $\PartySet \subseteq Z_m \cup Z_{\alpha} \cup Z_{\beta} \cup Z_a^{\star}$, which is a contradiction
            to the $\Q^{(3, 1)}(\PartySet, \Z_s, \Z_a)$ condition.    
            \item[--] {\bf Condition C}: There exists a subset of parties $\C'''_m \subseteq \C_m$, where
            $\C_m \setminus \C'''_m \in \Z_a$, such that $P_i$ received a common value from {\it all}
            the parties in $\C'''_m$. 
            In this case also, one can show that the subset $\C'''_m$ is bound to contain at least
            one {\it honest} party from $\C_m$, who would have sent 
            $s^{(m)}$ to $P_i$. This is because $\Z_s$ and $\Z_a$
            satisfy the $\Q^{(3, 1)}(\PartySet, \Z_s, \Z_a)$ condition and every subset in 
            $\Z_a$ is a subset of some subset in $\Z_s$.     
            Clearly, $P_i$ outputs $s^{(m)}$ as the share corresponding to $S_m$.                              
       \end{myitemize}
    Thus, we have shown that {\it irrespective} of the way $P_i$ would have computed its output share corresponding to
    $S_m$, it is bound to
    be the {\it same} as $s^{(m)}$.
    To complete the proof, we just need to show that at least one of the conditions from A, B and C above eventually holds for 
    $P_i$. For this, we note that in the {\it worst} case, 
    the condition C is bound to {\it eventually} hold, irrespective of conditions A and B.
    This is because the set of {\it honest} parties in $\C_m$, namely the parties in
    $\C_m \setminus Z_a^{\star}$, {\it always} constitute a candidate $\C'''_m$ set for $P_i$. 
    This follows from the fact that the share $s^{(m)}$ from {\it all} the parties in $\C_m \setminus Z_a^{\star}$ will be eventually delivered
    to $P_i$. 
\item {\bf The output of $\BA$ is 0}: 
Since $\D$ is {\it honest}, every pair of parties $P_j, P_k \in \Honest_m$ eventually broadcast
  $\OK(m, j, k)$ and $\OK(m, k, j)$ messages, as they eventually receive the same share $s^{(m)}$ from $\D$
   and exchange among themselves. From the $\Z_a$-validity of $\BCAST$ in the {\it asynchronous} network,
    these messages are eventually delivered to every honest party. 
    Also from the $\Z_a$-consistency of $\BCAST$ in the {\it asynchronous} network, 
    any $\OK$ message which is received by $\D$ from the broadcast of any {\it corrupt} party,
   will eventually be
   received by every other honest party as well. 
   Since $S_m \setminus \Honest_m \in \Z_a$, it follows that 
   {\it all honest} parties will eventually find a subset of parties $\E_m \subseteq S_m$, where
   $S_m \setminus \E_m \in \Z_a$, which constitutes a clique in the consistency graph $G^{(m)}$ 
   of {\it all} honest parties. This is because the set $\Honest_m$ constitutes such a candidate 
   $\E_m$. Consequently, $\D$ eventually finds and broadcasts $\E_m$.
   From the $\Z_a$-weal validity and $\Z_a$-fallback validity properties of $\BCAST$
   in the {\it asynchronous} network, $\E_m$ will be eventually received and accepted by {\it all}
   honest parties. 
   
   Next, consider an arbitrary $P_i \in \Honest_m$. We wish to show that $P_i$ eventually outputs
   $s^{(m)}$ as the share, corresponding to $S_m$. Now, there are two possible cases. If $P_i \in \E_m$,
   then from the protocol steps, $P_i$ indeed outputs $s^{(m)}$ as its share, corresponding to $S_m$. So, consider the other case
   when $P_i \not \in \E_m$. Note that all the parties in $\Honest_m$ eventually receive the common share 
   $s^{(m)}$ from $\D$, since $\D$ is {\it honest}. 
   Also note that $\E_m \setminus \Honest_m \in \Z_s$; this
    is because every subset in 
            $\Z_a$ is a subset of some subset in $\Z_s$. 
     Hence, it follows that party $P_i$ will eventually find a candidate subset $\E'_m \subseteq \E_m$, where
     $\E_m \setminus \E'_m \in \Z_s$,  such that $P_i$ receives a {\it common} value from {\it all}
     the parties in $\E'_m$ and set that value as its share, corresponding to $S_m$.
      This is because the subset $(\Honest_m \cap \E_m)$ always constitute such a candidate $\E'_m$ set.
     Hence, it is confirmed that $P_i$ is guaranteed to output some share corresponding to $S_m$.
     To complete the proof, we need to show that this share is the {\it same} as $s^{(m)}$.
     
     So, let $P_i$ find a candidate $\E'_m$ set, satisfying the above conditions,
     based on which it computes its output share corresponding to $S_m$.
     We claim that this set $\E'_m$ contains at least one {\it honest} party from 
     $\Honest_m$; i.e.~$\Honest_m \cap \E'_m \neq \emptyset$. 
     On the contrary, let the candidate $\E'_m$ for $P_i$ consists of only {\it corrupt}
     parties. That is, $\E'_m \subseteq Z_a^{\star}$.
     We consider the worst case scenario where 
     $Z_a^{\star} \in \Z_s$ as well, since every subset in $\Z_a$
     is assumed to be a subset of some subset in $\Z_s$. 
     Also note that $S_m = \PartySet \setminus Z_m$, where $Z_m \in \Z_s$.
     Let $S_m \setminus \E_m \subseteq Z_{\beta} \in \Z_a$.
     And let $\E_m \setminus \E'_m \subseteq Z_{\alpha} \in \Z_s$.
     Hence, we get that $\PartySet \subseteq Z_m \cup Z_a^{\star} \cup Z_{\alpha} \cup Z_{\beta}$, which
     is a contradiction, since $\Z_s$ and $\Z_a$
            satisfy the $\Q^{(3, 1)}(\PartySet, \Z_s, \Z_a)$ condition. 
\end{myenumerate}
\end{proof}

We next prove the {\it privacy} property.
\begin{lemma}
\label{lemma:VSSPrivacy}
In protocol $\VSS$, if  $\D$ is {\it honest} and participates with input $s$, then irrespective of the type of network, 
 the view of the adversary remains independent of $s$.
\end{lemma}
\begin{proof}
We prove privacy in a {\it synchronous} network. 
 The privacy in an {\it asynchronous} network automatically follows, since every subset in $\Z_a$ is a subset of some
  subset in $\Z_s$. So, consider a {\it synchronous} network, and let $\D$ be {\it honest}.
  Let $Z_c \in \Z_s$ be the set of {\it corrupt} parties. Then consider the set $S_c \in \SharingSpec$, where
  $S_c \defined \PartySet \setminus Z_c$.
  We claim that, throughout the protocol, the adversary does not learn anything about the share $s^{(c)}$, and its view remains
  independent of $s^{(c)}$. The $\Z_s$-privacy then follows from the fact that
  $\D$ selects the share $s^{(c)}$ randomly, and 
   the probability distribution of $s^{(c)}$ is {\it independent} of $s$.
   
   Since $\D$ is {\it honest}, it sends the share $s^{(c)}$, {\it only} to the parties in $S_c$, which consists of {\it only}
   honest parties. Similarly, as part of the pairwise consistency tests, the share $s^{(c)}$ is exchanged {\it only}
   among the parties in $S_c$. Moreover, since $S_c$ consists of {\it no} corrupt parties, it follows that {\it no}
   party from $S_c$ ever broadcasts an $\NOK(c, \star)$ message, corresponding to $S_c$.
   Consequently, {\it no} party from $S_c$, as well as $\D$, ever broadcasts a $\resolve(c, s^{(c)})$ message.
   Thus, throughout the protocol, the view of the adversary remains independent of the share $s^{(c)}$.
\end{proof}

We next proceed to prove the commitment properties. We start with the {\it synchronous} network.
\begin{lemma}
\label{lemma:VSSSynchronousCommitment}
In protocol $\VSS$, if $\D$ is corrupt, then either no honest party obtains any output, or there exists a value $s^{\star} \in \K$
 held by $\D$, which is secret-shared with respect to the sharing specification $\SharingSpec$, such that 
  the following hold.
       \begin{myitemize}
       \item[--] If any honest party outputs its shares at time $\TimeVSS = 2\Delta + 2\TimeBCAST + \TimeBA$, then all honest parties
       output their shares at time $\TimeVSS$.
       \item[--] If any honest party outputs its shares at time $T > \TimeVSS$,
    then every honest party outputs its shares by time $T + 2\Delta$.  
       \end{myitemize}   
\end{lemma}
\begin{proof}
If {\it no} honest party obtains any output, then the lemma holds trivially. So, consider the case when some
  {\it honest} party, say $P_h$,
 obtains an output. We note that {\it every} honest party participates with {\it some} input in the instance of $\BA$ at
 time $2\Delta + 2\TimeBCAST$. 
  Hence, by the $\Z_s$-consistency and $\Z_s$-guaranteed liveness properties of $\BA$ in the {\it synchronous} network, 
  the instance of $\BA$ generates an output at time $2\Delta + 2\TimeBCAST + \TimeBA$ for
  every honest party. Now there are two possible cases.
   \begin{enumerate}
    \item {\bf The output of $\BA$ is $1$}:
     In this case, from Lemma \ref{lemma:GenVSSSync}, it follows that {\it all honest} parties participated with input
     $1$ at time $2\Delta + 2\TimeBCAST$ during the instance of $\BA$. This implies that {\it all honest}
     parties received $\C_1, \ldots, \C_q$ from the broadcast of $\D$ through regular-mode at time 
     $2\Delta + 2\TimeBCAST$ and accepted these sets.
     We next claim that corresponding to {\it every}
    $S_m \in \SharingSpec$, {\it all honest} parties in $S_m$ output some {\it common} share say $s^{(m)}$, corresponding to $S_m$,
    at time  $2\Delta + 2\TimeBCAST + \TimeBA$. 
    Let $s^{\star} \defined s^{(1)} + \ldots + s^{(m)}$. It will then follow  
     that the value $s^{\star}$ is
     secret-shared at time
    $2\Delta + 2\TimeBCAST + \TimeBA$.  
    
    The proof of the above claim closely follows the $\Z_s$-correctness proof in the {\it synchronous} network
     (see the proof of Lemma \ref{lemma:VSSSynchronousCorrectness}). 
     Consider an {\it arbitrary} $S_m \in \SharingSpec$. Then there are two possible cases.
       \begin{myenumerate}
        \item[--] {\bf An $\NOK(m, \star)$ message was received at time $2\Delta + \TimeBCAST$ through regular-mode from
        the broadcast of some party in $S_m$}: In this case, 
        {\it all honest} parties in $S_m$ will output some {\it common} share, say $s^{(m)}$, corresponding to $S_m$
        at time $2\Delta + 2\TimeBCAST + \TimeBA$.  The proof for this is {\it exactly} the same as
        the proof of Lemma \ref{lemma:VSSSynchronousCorrectness} for the same case. 
        \item[--] {\bf No $\NOK(m, \star)$ message was received within time $2\Delta + \TimeBCAST$ through
        regular-mode from the broadcast of any party in $S_m$}:
        Let $\Honest_m$ be the set of {\it honest} parties in $S_m$. Since 
        $\C_m$ is accepted by all the honest parties, it follows that the parties in
        $\Honest_m \cap \C_m$ constitute a clique in the consistency graph $G^{(m)}$
        of {\it all} honest parties. This further implies that all the parties in  
        $\Honest_m \cap \C_m$ received a common share from $\D$ corresponding to $S_m$, say
        $s^{(m)}$. Moreover, 
         from 
         Lemma \ref{lemma:GenVSSSync}, {\it all} the 
          parties in $\Honest_m \cap \C_m$ would have received
         $s^{(m)}$ from 
         $\D$, within time $\Delta$. 
         Now similar to the proof for the same case in Lemma \ref{lemma:VSSSynchronousCorrectness}, 
         it can be concluded that  all
         honest parties output $s^{(m)}$ as the share corresponding to $S_m$,
          at time 
        $2\Delta + 2\TimeBCAST + \TimeBA$.    
    \end{myenumerate}
    \item {\bf The output of $\BA$ is $0$}: Since $P_h$ has obtained an output, it implies that
    it has received sets $\E_1, \ldots, \E_q$ from the broadcast of $\D$ and accepted them. Let
    $T$ be the time at which $P_h$ accepted $\E_1, \ldots, \E_q$. Note that
    $T > 2\Delta + 2\TimeBCAST + \TimeBA$. This is because from the protocol steps,
    the {\it honest} parties start participating in the instance of $\BCAST$ of $\D$ for broadcasting $\E_1, \ldots, \E_q$, 
    {\it only after} time
    $2\Delta + 2\TimeBCAST + \TimeBA$. By the $\Z_s$-consistency and 
    $\Z_s$-fallback consistency of $\BCAST$ in the {\it synchronous} network, 
    {\it all honest} parties will receive and accept the sets $\E_1, \ldots, \E_q$,
    latest by time $T + 2\Delta$. Since the parries have accepted 
    $\E_1, \ldots, \E_q$, it implies that 
    corresponding to every $S_m \in \SharingSpec$, 
     {\it all} the honest parties in $\E_m$ have received a common value from $\D$, say $s^{(m)}$. 
     We claim that all the honest parties in $S_m$ output $s^{(m)}$ as the share, corresponding to $S_m$, latest by time 
     $T + 2\Delta$. This will automatically imply that the value
     $s^{\star} \defined s^{(1)} + \ldots + s^{(m)}$ is secret-shared, latest by time $T + 2\Delta$.
     
     The proof for the above claim is exactly the {\it same} as the proof of Lemma \ref{lemma:VSSAsynchronousCorrectness},
     for the case when the output of $\BA$ is $0$ and is omitted.          
   \end{enumerate}
\end{proof}

We finally prove the commitment property in an {\it asynchronous} network.
\begin{lemma}
\label{lemma:VSSAsynchronousCommitment}
In protocol $\VSS$, if $\D$ is corrupt, then either no honest party obtains any output or there exists some value $s^{\star} \in \K$ held
 by $\D$, such that almost-surely $s^{\star}$ is secret-shared, with respect to the sharing specification $\SharingSpec$.
\end{lemma}
\begin{proof}
If no honest party obtains an output, then the lemma holds trivially. So, consider the case when some honest party, say
 $P_h$, has obtained an output in $\VSS$. Note that {\it every honest} party participates with some input in the instance of 
  $\BA$ at local time $2\Delta + 2\TimeBCAST$. 
  Hence, from the $\Z_a$-almost-surely liveness and $\Z_a$-consistency properties of $\BA$ in the {\it asynchronous} network,
  it follows that the instance of $\BA$  eventually generates a {\it common} output for every honest party.
  Now, there are two possible cases.
  \begin{myenumerate}
  \item {\bf The output of $\BA$ instance is $1$}:
   This implies that all honest parties eventually receive the sets
   $\C_1, \ldots, \C_q$ from the broadcast of $\D$ (either through regular-mode or fallback-mode),
   and accept these sets. The proof for this is identical to the case ``when the output of $\BA$ instance is $1$", in
   the proof of Lemma \ref{lemma:VSSAsynchronousCorrectness}. This further implies that corresponding to each
   $S_m \in \SharingSpec$, all the {\it honest} parties in $\C_m$ have received a common value
    from $\D$, say $s^{(m)}$, since
   the {\it honest} parties in $\C_m$ constitute a clique. We claim that all the {\it honest} parties in $S_m$ eventually output
   $s^{(m)}$ as the share, corresponding to the set $S_m$. This will automatically imply that the value
   $s^{\star} \defined s^{(1)} + \ldots + s^{(q)}$ is eventually secret-shared among the parties. 
   The proof for the claim again closely follows the case ``when the output of $\BA$ instance is $1$", in
   the proof of Lemma \ref{lemma:VSSAsynchronousCorrectness}. 
   
     So, consider an {\it arbitrary honest} $P_i \in S_m$, and let
     $Z_a^{\star}$ be the set of {\it corrupt} parties. 
    From the protocol steps, $P_i$ computes its share corresponding to $S_m$ based on one of the following three conditions.
        Assuming that at least one of these conditions eventually hold for $P_i$, we
        first show that the share computed by $P_i$ corresponding to $S_m$ is bound to be
        $s^{(m)}$. This is followed by showing that at least one of these conditions eventually hold for $P_i$.        
      \begin{myitemize}
            \item[--] {\bf Condition A}: At time $2\Delta + 2\TimeBCAST$, there exists some value $s'^{(m)}$ such that
            party $P_i$ received the $\resolve(m, s'^{(m)})$ message
              from the broadcast of $\D$, as well as from the broadcast of a subset of parties $\C'_m \subseteq \C_m$ through
             regular-mode,
            where $\C_m \setminus \C'_m \in \Z_s$. And consequently $P_i$ sets $s'^{(m)}$ as its share corresponding to $S_m$.
            
             We argue that $\C'_m$ is bound to contain {\it at least} one honest party from $\C_m$, which broadcasts  
             $\resolve(m, s'^{(m)})$ message, where $s'^{(m)} = s^{(m)}$. 
                 In more detail, let $S_m \setminus \C_m = Z_{\alpha} \in \Z_s$ and
            $\C_m \setminus \C'_m = Z_{\beta} \in \Z_s$.
            Also, note that $\PartySet \setminus S_m = Z_m \in \Z_s$.
            Now if $\C'_m$ {\it does not} contain any honest party from $\C_m$, it implies that
            $\C'_m \subseteq Z_a^{\star} \in \Z_a$.
            This further implies that $\PartySet \subseteq Z_m \cup Z_{\alpha} \cup Z_{\beta} \cup Z_a^{\star}$, which is a contradiction
            to the $\Q^{(3, 1)}(\PartySet, \Z_s, \Z_a)$ condition.          
            \item[--] {\bf Condition B}: At time $2\Delta$, there exists a subset of parties $\C''_m \subseteq \C_m$, where $\C_m \setminus \C''_m \in \Z_s$, such that $P_i$ received a common value from {\it all}
            the parties in $\C''_m$, say $s'^{(m)}$. And consequently, $P_i$ sets $s'^{(m)}$ as its share corresponding to $S_m$.

             We claim that the subset $\C''_m$ is bound to contain at least
            one {\it honest} party from $\C_m$, who would have sent 
            $s'^{(m)} = s^{(m)}$ to $P_i$ within time $2\Delta$.
            The proof for the claim is similar to the case above where we have shown that
            $\C'_m$ is bound to contain at least
            one {\it honest} party from $\C_m$.     
            \item[--] {\bf Condition C}: There exists a subset of parties $\C'''_m \subseteq \C_m$ where
            $\C_m \setminus \C'''_m \in \Z_a$, such that $P_i$ received a common value from {\it all}
            the parties in $\C'''_m$, say $s'^{(m)}$. And consequently, $P_i$ sets $s'^{(m)}$ as its share, corresponding to
            $S_m$. In this case also, one can show that the subset $\C'''_m$ is bound to contain at least
            one {\it honest} party from $\C_m$, who would have sent 
            $s'^{(m)} = s^{(m)}$ to $P_i$. This is because $\Z_s$ and $\Z_a$
            satisfy the $\Q^{(3, 1)}(\PartySet, \Z_s, \Z_a)$ condition and every subset in 
            $\Z_a$ is a subset of some subset in $\Z_s$.                                   
       \end{myitemize}
    Thus, we have shown that {\it irrespective} of the way $P_i$ would have computed its output share corresponding to
    $S_m$, it is bound to
    be the {\it same} as $s^{(m)}$. 
    Now it is easy to see that at least one of the conditions A, B and C above, eventually holds for 
    $P_i$. Specially, the condition C is bound to {\it eventually} hold, irrespective of conditions A and B.
    This is because the set of {\it honest} parties in $\C_m$, namely the parties in
    $\C_m \setminus Z_a^{\star}$, {\it always} constitute a candidate $\C'''_m$ set for $P_i$. 
    This follows from the fact that the share $s^{(m)}$ from {\it all} the parties in $\C_m \setminus Z_a^{\star}$ will be eventually delivered
    to $P_i$. 
  \item {\bf The output of $\BA$ instance is $0$}:
 Since $P_h$ has computed its output, it follows that it has received the sets $\E_1, \ldots, \E_q$ from the broadcast of
  $\D$, and accepted these sets. From the $\Z_a$-weak consistency and $\Z_a$-fallback consistency properties of $\BCAST$
  in the {\it asynchronous} network, it follows that {\it all honest} parties eventually receive these sets and accept them.
  We also note that corresponding to every $S_m \in \SharingSpec$, {\it all honest} parties
  in $\E_m$ received the same share from $\D$, say $s^{(m)}$. Now similar to the proof for the case ``the output of $\BA$ instance is $0$"
  in the proof of Lemma \ref{lemma:VSSAsynchronousCommitment}, it can be shown that {\it all honest} parties in $S_m$
  eventually set $s^{(m)}$ as the share, corresponding to $S_m$.
  This automatically implies that the value $s^{\star} \defined s^{(1)} + \ldots + s^{(q)}$
  is eventually secret-shared.
  \end{myenumerate}
\end{proof}
\begin{lemma}
\label{lemma:VSSCC}
Protocol $\VSS$ incurs a communication of $\Order(|\Z_s| \cdot (n^4 \log{|\K|} + n^5  (\log n + \log{|\Z_s|})))$ bits and invokes one instance of 
 $\BA$.
\end{lemma}
\begin{proof}
During phase I, corresponding to every $S_m \in \SharingSpec$, dealer
  $\D$ needs to send a share, consisting of one element of $\K$, to {\it all} the parties
 in $S_m$. This incurs a total communication of $\Order(|\Z_s| \cdot n)$ elements from $\K$.
  During phase II, every party in $S_m$ sends an element from $\K$ to every other party in $S_m$. 
  This incurs a total communication of $\Order(|\Z_s| \cdot n^2)$ elements from $\K$.
  During phase III, every party in $S_m$ may broadcast an $\OK$ message for every other party in $S_m$.
  This requires broadcasting $\Order(|\Z_s| \cdot n^2  \cdot (\log n + \log{|\Z_s|}))$ bits,
   as each $\OK$ message encodes the identity of
  two parties and the identity of a set from $\Z_s$, requiring $2\log n + \log{|\Z_s|}$ bits. 
  Additionally, during phase III, corresponding to any $S_m \in \SharingSpec$, each party in $S_m$
  may broadcast an $\NOK$ message. This requires broadcasting $\Order(|\Z_s| \cdot n \cdot (\log{|\Z_s|} + \log{n}))$ bits, 
  as each $\NOK$ message encodes the identity of a party and the identity of a set from $\Z_s$.
    During phase IV, corresponding to each $S_m \in \SharingSpec$, up to $\Order(n)$ elements from $\K$ may be broadcasted
  to resolve the conflicts. This incurs a total broadcast of  $\Order(|\Z_s| \cdot n)$ elements from $\K$. Additionally, $\D$
  may broadcast sets $\C_1, \ldots, \C_q$, where each set can be encoded by an $n$-bit vector. 

  The communication complexity now follows by summing up all the above costs and from the communication complexity of
  the protocol $\BCAST$, along with the fact that each element from $\K$
   can be represented by $\log{|\K|}$ bits, and the fact that $q = |\Z_s|$.
\end{proof}
\begin{remark}[{\bf Further improvement in the communication complexity of $\VSS$}]
\label{remark:improvement}
The complexity of the phase III can be significantly reduced by making the following modification: 
 party $P_i$ now broadcasts a {\it single} $\OK(i, j)$ message corresponding to $P_j$, if
  the pairwise consistency test is positive between $P_i$'s and $P_j$'s share across {\it all} the sets in $\SharingSpec$, to which
  both $P_i$ and $P_j$ belongs. That is, if corresponding to {\it every} $S_m \in \SharingSpec$ such that $P_i, P_j \in S_m$,
  the condition $s_i^{(m)} = s_j^{(m)}$ holds. Consequently, during phase III, only $\Order(n^2)$ $\OK$ messages need to be broadcasted,
  where the size of each message will now be {\it only} $\Order(\log n)$ bits. This will reduce the communication complexity
  of $\VSS$ to $\Order(|\Z_s| \cdot n^4 \cdot ( \log{|\K|} + \log{|\Z_s|} + \log n) + n^5 \log n)$ bits, along with one instance of $\BA$.
\end{remark}

The proof of the following theorem now follows from Lemma \ref{lemma:VSSSynchronousCorrectness}-\ref{lemma:VSSCC}
 and Remark \ref{remark:improvement}. \\~\\
 \noindent {\bf Theorem \ref{thm:VSS}.}
{\it Let $\Adv$ be an adversary characterized by an adversary structure $\Z_s$ in a synchronous network and adversary structure $\Z_a$ in an asynchronous network satisfying the following conditions.
  \begin{myitemize}
  \item[--] $\Z_s \neq \Z_a$;
  \item[--] For every subset $Z \in \Z_a$, there exists a subset $Z' \in \Z_s$, such that $Z \subseteq Z'$;
  \item[--] $\Z_s$ and $\Z_a$ satisfy the $\Q^{(3, 1)}(\PartySet, \Z_s, \Z_a)$ condition.
  \end{myitemize}
  Moreover, let $\SharingSpec = \{S_m : S_m = \PartySet \setminus Z_m \; \mbox{ and } \; Z_m \in \Z_s \}$.
    Then protocol $\VSS$ achieves the following properties, where $\D$ has a private input $s \in \K$ for $\VSS$.
   \begin{myitemize}
   \item[--] If $\D$ is honest, then the following hold.
         \begin{myitemize}
         \item[--] {\bf $\Z_s$-correctness}: In a synchronous network, $s$ is secret-shared with respect to $\SharingSpec$, at 
            time $\TimeVSS = 2\Delta + 2\TimeBCAST + \TimeBA$. 
         \item[--] {\bf $\Z_a$-correctness}: In an asynchronous network, almost-surely, $s$ is eventually secret-shared, 
          with respect to $\SharingSpec$.
         \item[--] {\bf Privacy}: The view of $\Adv$ remains independent of $s$, irrespective of the network type.  
         \end{myitemize}   
   \item[--] If $\D$ is corrupt, then either no honest party obtains any output or there exists some
   $s^{\star} \in \K$, such that the following hold. 
          \begin{myitemize}
          \item[--]  {\bf $\Z_a$-commitment}: In an asynchronous network, almost-surely, $s^{\star}$ is eventually secret-shared, 
          with respect to $\SharingSpec$.          
           \item[--]  {\bf $\Z_s$-commitment}: In a synchronous network, $s^{\star}$ is secret-shared, with respect to $\SharingSpec$,
            such that the following hold.
              \begin{myitemize}
		       \item[--] If any honest party outputs its shares at time $\TimeVSS$, then all honest parties
		       output their shares at time $\TimeVSS$.
		       \item[--] If any honest party outputs its shares at time $T > \TimeVSS$,
		    then every honest party outputs its shares by time $T + 2\Delta$.  
	       \end{myitemize}   
          \end{myitemize}
   \item[--] {\bf Communication Complexity}: The protocol incurs a communication of
    $\Order(|\Z_s| \cdot n^4 \cdot ( \log{|\K|} + \log{|\Z_s|} + \log n) + n^5 \log n)$ bits and invokes one instance of $\BA$.
   \end{myitemize}
 }


\section{Properties of the Preprocessing Phase Protocol}
\label{app:Preprocessing}
In this section, we formally present our preprocessing phase protocol
 and prove its properties. We first start with the description of our ACS
 protocol and proof of its properties.
 \subsection{Protocol $\ACS$ and Its Properties}
    For simplicity, we present $\ACS$ when $L = 1$;
   the modifications for a general  $L$ are straightforward. 
\begin{protocolsplitbox}{$\ACS(\Q)$}{Agreement on a common subset of parties.
 The above code is for each $P_i \in \PartySet$.}{fig:ACS}
  \begin{myitemize}
  \item[--] {\bf Phase I --- Secret Sharing the Input}: If $P_i \in \Q$, the do the following.
       \begin{myitemize}
        \item[--] On having the input $x_i$, act as a dealer $\D$ 
           and invoke an instance $\VSS^{(i)}$ of $\VSS$ with input
            $x_i$.
          \item[--] Participate in the instance $\VSS^{(j)}$ invoked by every $P_j \in \Q$,
         and {\color{red} wait for time $\TimeVSS$}.               
        Initialize a set $\CSet_i = \emptyset$ {\color{red} after time $\TimeVSS$} 
        and include $P_j \in \Q$ in 
            $\CSet_i$ if an output is computed in the instance $\VSS^{(j)}$.  
         \end{myitemize}   

     \item[--] {\bf Phase II --- Identifying the Common Set of Selected Parties}: 
           \begin{myitemize}
            \item[--] Corresponding to every $P_j \in \Q$,
             participate in an instance of $\BA^{(j)}$ of $\BA$ with input $1$, if $P_j \in \CSet_i$.
            \item[--] Once $1$ has been obtained as the output from instances of $\BA$ corresponding to a set of parties
            in $\Q \setminus Z$ for some $Z \in \Z_s$, participate 
            with input $0$ in all the $\BA$ instances $\BA^{(j)}$, such that 
            $P_j \in \Q$ and 
            $P_j \not \in \CSet_i$.
            \item[--] Once all the instances of $\BA$ corresponding to the parties in $\Q$ have produced a binary output,
             then output $\CoreSet$, which is the set of parties $P_j \in \Q$ such that
                           $1$ is obtained as the output in the instance $\BA^{(j)}$. 
             \end{myitemize} 
       \end{myitemize}
\end{protocolsplitbox}
 
 We next prove the properties of the protocol $\ACS$, assuming $L = 1$. 
 \begin{lemma}
  \label{lemma:ACS}.
 Let $\Adv$ be an adversary characterized by an adversary structure $\Z_s$ in a synchronous network and adversary structure $\Z_a$ in an asynchronous network satisfying the following conditions.
  \begin{myitemize}
  \item[--] $\Z_s \neq \Z_a$;
  \item[--] For every subset $Z \in \Z_a$, there exists a subset $Z' \in \Z_s$, such that $Z \subseteq Z'$;
  \item[--] $\Z_s$ and $\Z_a$ satisfy the $\Q^{(3, 1)}(\PartySet, \Z_s, \Z_a)$ condition.
  \end{myitemize}
  Moreover, let $\SharingSpec = \{S_m : S_m = \PartySet \setminus Z_m \; \mbox{ and } \; Z_m \in \Z_s \}$.
    Furthermore, let $\Q \subseteq \PartySet$, such that
   $\Z_s$ and $\Z_a$
   {\it either} satisfy the $\Q^{(1, 1)}(\Q, \Z_s, \Z_a)$ condition 
   {\it or} the $\Q^{(3, 1)}(\Q, \Z_s, \Z_a)$ condition.  
   Then, $\ACS$ achieves the following, where every (honest) $P_i \in \Q$ has input $x_i \in \K$ for $\ACS$.
   \begin{myitemize}
   \item[--] {\bf $\Z_s$-correctness}: If the network is synchronous,
    then at time $\TimeACS \defined \TimeVSS + 2\TimeBA$, all honest parties output a common subset of parties
     $\CoreSet \subseteq \Q$,
    where $\Q \setminus \CoreSet \in \Z_s$, such that the following hold.
      \begin{myitemize}
      \item[--] All honest parties from $\Q$ will be present in $\CoreSet$. 
      \item[--] Corresponding to every $P_j \in \CoreSet$, there exists some $x^{\star}_j \in \K$, where
      $x^{\star}_j = x_j$ for an honest $P_j$, such that $x^{\star}_j$ is secret-shared with respect to $\SharingSpec$.
      \end{myitemize}
   \item[--] {\bf $\Z_a$-correctness}: If the network is asynchronous,
   then almost-surely, the honest parties eventually output a subset $\CoreSet \subseteq \Q$
    where $\Q \setminus \CoreSet \in \Z_s$. Moreover, 
    corresponding to every $P_j \in \CoreSet$, there exists an $x^{\star}_j \in \K$, where
      $x^{\star}_j = x_j$ for an honest $P_j$, such that $x^{\star}_j$ is eventually secret-shared with respect to $\SharingSpec$.
    \item[--] {\bf Privacy}: Irrespective of the network type, the view of the adversary remains independent of the inputs $x_i$ corresponding to
    the honest parties $P_i \in \Q$. 
    \item[--] {\bf Communication Complexity}: The protocol incurs a communication of
    $\Order(|\Z_s| \cdot n^5 \cdot ( \log{|\K|} + \log{|\Z_s|} + \log{n} ) + n^6 \log n)$  bits and invokes $\Order(n)$ instances of $\BA$.
   \end{myitemize}     
   \end{lemma}
 \begin{proof}
 The {\it privacy} property simply follows from the {\it privacy} property of $\VSS$, while {\it communication complexity} follows 
 from the {\it communication complexity} of $\VSS$ and 
  the fact that up to $\Order(n)$ instances of
  $\VSS$ may be involved, since $|\Q| = \Order(n)$. We next prove the {\it correctness} property.

We first consider a {\it synchronous} network. Let $Z_s^{\star} \in \Z_s$ be the set of {\it corrupt} parties, and
 let $\Honest \defined \Q \setminus \Z_s^{\star}$
  be the set of {\it honest} parties.
   We note that $\Honest \neq \emptyset$, since it is given that
    $\Z_s$ and $\Z_a$
   {\it either} satisfy the $\Q^{(1, 1)}(\Q, \Z_s, \Z_a)$ condition 
   {\it or} $\Q^{(3, 1)}(\Q, \Z_s, \Z_a)$ condition.  
     Corresponding to each $P_j \in \Honest$, every {\it honest} $P_i$ obtains the output
   $\{[x_j]_m \}_{P_i \in S_m}$ at time $\TimeVSS$ during $\VSS^{(j)}$,
    which follows from the $\Z_s$-correctness of $\VSS$ in
    the {\it synchronous} network.
  Consequently, at time $\TimeVSS$, the set $\CSet_i$ of {\it every honest}
  $P_i$ will satisfy the condition $\Q \setminus \CSet_i \in \Z_s$.
  This is because $\Honest \subseteq \CSet_i$ will hold at time $\TimeVSS$. 
   Now, corresponding to each $P_j \in \Honest$, each $P_i \in \Honest$
    starts participating with input $1$ in the instance $\BA^{(j)}$ at time
  $\TimeVSS$. Hence, from the $\Z_s$-validity and $\Z_s$-guaranteed liveness properties of
   $\BA$ in the {\it synchronous} network, it follows that at time $\TimeVSS + \TimeBA$,
    every $P_i \in \Honest$ obtains the output $1$ during the instance $\BA^{(j)}$
  corresponding to every $P_j \in \Honest$. Consequently, at time $\TimeVSS + \TimeBA$,
   every {\it honest} party will start participating in the remaining $\BA$ instances for which no input has been provided yet (if there are any),
  and from the $\Z_s$-guaranteed liveness and $\Z_s$-consistency properties
   of $\BA$ in the {\it synchronous} network,
   these $\BA$ instances will produce common outputs for every honest party at time $\TimeACS = \TimeVSS + 2 \TimeBA$.
  Hence, at time $\TimeACS$, every honest party outputs a common $\CoreSet$, where $\Q \setminus \CoreSet \in \Z_s$,
   and where each $P_j \in \Honest$ will be present in $\CoreSet$.
  We next wish to show that corresponding to {\it every} $P_j \in \CoreSet$, 
  there exists some value which is secret-shared among the parties with respect to $\SharingSpec$.
    
 Consider an {\it arbitrary} party $P_j \in \CoreSet$. If $P_j \in \Honest$,
  then as argued above, every $P_i \in \Honest$ computes the shares $\{[x_j]_m \}_{P_i \in S_m}$ at time $\TimeVSS$ itself. 
  Next, consider a {\it corrupt} $P_j \in \CoreSet$. Since $P_j \in \CoreSet$, it follows that the instance $\BA^{(j)}$ produces the output $1$
  for all honest parties.
   This further implies that at least one 
  {\it honest} $P_i$ must have computed some output during the instance $\VSS^{(j)}$ by time
   $\TimeVSS + \TimeBA$ (implying that $P_j \in \CSet_i$) and participated
  with input $1$ in the instance $\BA^{(j)}$. This is because if, at time $\TimeVSS + \TimeBA$, 
  party  $P_j$ {\it does not} belong to the $\CSet_i$ set of {\it any} honest $P_i$,
   then it implies that {\it all honest} parties start participating with input
   $0$ in the instance $\BA^{(j)}$ at time $\TimeVSS + \TimeBA$. 
   Then, from the $\Z_s$-validity of $\BA$ in the {\it synchronous} network, every honest party would have
   obtained the output $0$ in the instance
  $\BA^{(j)}$ and hence, $P_j$ will not be present in $\CoreSet$, which is a contradiction. 
  Now, if $P_i$ has computed
   some output during $\VSS^{(j)}$ at time $\TimeVSS + \TimeBA$, then from the $\Z_s$-commitment of $\VSS$ in the
   {\it synchronous} network, it follows that
   there exists some value $x_j^{\star}$, such that $x_j^{\star}$ will be secret-shared with respect to $\SharingSpec$ 
   by time $\TimeVSS + \TimeBA + 2 \Delta$.
  Since $2 \Delta < \TimeBA$, it follows that at time $\TimeACS$, every $P_i \in \Honest$ has $\{[x^{\star}_j]_m \}_{P_i \in S_m}$,
   thus proving the 
  {\it correctness} property in a {\it synchronous} network.
  
  We next consider an {\it asynchronous} network. 
   Let $Z_a^{\star} \in \Z_a$ be the set of {\it corrupt} parties, and
 let $\Honest \defined \Q \setminus \Z_a^{\star}$
  be the set of {\it honest} parties. Notice that $\Q \setminus \Honest \in \Z_s$, since every subset in $\Z_a$ is a subset of some
  subset in $\Z_s$. 
    Now, irrespective of the way messages are scheduled, there
    will eventually be subset of parties $\Q \setminus Z$ for some $Z \in \Z_s$, such that
    all the parties in $\Honest$ participate with input $1$ in the instances of $\BA$ corresponding to the parties in
    $\Q \setminus Z$. This is because
  corresponding to every $P_j \in \Honest$, every $P_i \in \Honest$ {\it eventually} computes
  an output during the instance $\VSS^{(j)}$,
   which follows from the 
  $\Z_a$-correctness of $\VSS$ in the {\it asynchronous} network. 
   So, even if the {\it corrupt} parties $P_j$ do not invoke their respective $\VSS^{(j)}$ instances, 
   there will be a set of $\BA$ instances corresponding to the parties in $\Q \setminus Z$ for some $Z \in \Z_s$ 
    in which all
   the parties in $\Honest$ eventually participate with input $1$.
   Consequently, from the $\Z_a$-almost-surely liveness and $\Z_a$-consistency properties of 
   $\BA$ in the {\it asynchronous} network, these $\BA$ instances
  eventually produce the output $1$ for all the parties in $\Honest$.
   Hence, all the parties in $\Honest$ eventually participate with some input in the remaining $\BA$ instances,
  which almost-surely produce some output for every honest party eventually. 
   From the properties of $\BA$ in the {\it asynchronous} network,
   it then follows that all the honest
  parties output the same $\CoreSet$.
  
  Now, consider an {\it arbitrary} 
  $P_j \in \CoreSet$. It implies that the honest parties computed the output $1$ during
   the instance $\BA^{(j)}$, which further implies that at least one
   {\it honest} $P_i$ participated with input $1$ in $\BA^{(j)}$ after computing
    its output in the instance $\VSS^{(j)}$. If $P_j$ is {\it honest}, then the
   $\Z_a$-correctness of $\VSS$ in the {\it asynchronous} network 
    guarantees that $x_j$ will be {\it eventually} secret-shared with respect to $\SharingSpec$, during $\VSS^{(j)}$.
   On the other hand, if $P_j$ is {\it corrupt}, then the $\Z_a$-commitment 
   of $\VSS$ in the {\it asynchronous} network 
    guarantees that there exists some $x^{\star}_j \in \K$, such that  
    $x^{\star}_j$ will {\it eventually} be secret-shared with respect to $\SharingSpec$, during $\VSS^{(j)}$.
 \end{proof}
 We next discuss the modifications needed in the protocol $\ACS$ when each party in $\Q$ has $L$ inputs.
  \paragraph{\bf $\ACS$ for $L$ Inputs:}
  Protocol $\ACS$ can be easily extended if each party has $L$
   inputs. Now, each $P_j$ shares $L$ values through instances of $\VSS$.
   Moreover, the parties participate with input $1$ in the instance $\BA^{(j)}$ if they have computed some output
   in {\it all} the $L$ instances of $\VSS$ invoked by $P_j$ as a dealer. 
   The rest of the protocol steps remain the same. The protocol will now incur a communication of
    $\Order(L \cdot |\Z_s| \cdot n^5 (\log{|\K|} + \log{|\Z_s|}) + n^6 \log n)$ bits and invokes $\Order(n)$ instances of $\BA$.
 \subsection{Multiplication Protocol and Its Properties}
 In this section, we formally present our multiplication protocol $\Mult$ and
  prove its properties. For simplicity, we first consider the case when 
  $L = 1$, and present the formal details of the protocol in Fig \ref{fig:Mult}.
  \begin{protocolsplitbox}{$\Mult([a], [b])$}{The perfectly-secure multiplication protocol}{fig:Mult}
	\justify
Let $[a]_1, \ldots, [a]_q$ and $[b]_1, \ldots, [b]_q$ be the shares, corresponding to $[a]$ and $[b]$ respectively, where
 every $P_i \in \PartySet$ holds the shares $\{ [a]_m, [b_m] \}_{P_i \in S_m}$. 
\begin{myitemize}
\item[--] For every ordered pair $(l, m) \in \{1, \ldots, q \} \times \{1, \ldots, q \}$, the parties do the following to compute
 $[a_l \cdot b_m]$, where
 $a_l \defined [a]_l$ and $b_m \defined [b]_m$.
  \begin{myitemize}
   \item[--] Let $\Q_{l, m} \defined S_l \cap S_m$. The parties execute an instance $\ACS(\Q_{l, m})$ of $\ACS$, where 
   every $P_i \in \Q_{l, m}$ participates with the input $a_l \cdot b_m$.
   \item[--] Let $\R_{l, m} \subseteq \Q_{l, m}$ be the common subset of parties obtained as the output
   from the instance $\ACS(\Q_{l, m})$, where $\Q_{l, m} \setminus \R_{l, m} \in \Z_s$.
    Let $r \defined |\R_{l, m}|$ and let $\R_{l, m} = \{P_{\alpha_1}, \ldots, P_{\alpha_r}\}$.
    Moreover, corresponding to $P_{\alpha_i} \in \R_{l ,m}$, let 
     $v_i$ be the value which is secret-shared on the behalf of 
     $P_{\alpha_i}$ during $\ACS(\Q_{l, m})$.\footnote{If $P_{\alpha_i}$ is {\it honest}, then $v_i = a_l \cdot b_m$ holds.}.
    \item[--] The parties publicly check whether $v_1, \ldots, v_r$ are all equal. For  this, the parties locally compute
    $r - 1$ differences $[d_1] \defined [v_1] - [v_2], \ldots, [d_{r - 1}] \defined [v_1] - [v_r]$. This is followed by publicly reconstructing the differences
    $d_1, \ldots, d_{r - 1}$ by invoking instances of $\Rec$, and checking if all of them are $0$.
      \begin{myitemize}
      \item[--] If $d_1 = \ldots = d_{r - 1} = 0$, then the parties set $[a_l \cdot b_m] = [v_1]$.
      \item[--] Else, the parties publicly reconstruct $[a]_l$ and $[b]_m$ by invoking instances
      $\Rec([a]_l, \SharingSpec)$ and  $\Rec([b]_m, \SharingSpec)$
       of $\Rec$.
      The parties then set $[a_l \cdot b_m]$ to the default sharing of $a_l \cdot b_m$, where 
       $[a_l \cdot b_m]_1 = a_l \cdot b_m$ and 
        $[a_l \cdot b_m]_2 = \ldots =  [a_l \cdot b_m]_q = 0$.\footnote{The vector of shares 
        $(s, 0, \ldots, 0)$ can be considered as a default sharing of any given $s \in \K$.}
      \end{myitemize}
  \end{myitemize}
\item[--] The parties output $[a \cdot b] = \displaystyle \sum_{(l, m) \in \{1, \ldots, q \} \times \{1, \ldots, q \}} [a_l \cdot b_m]$.
\end{myitemize}
\end{protocolsplitbox}

 We next prove the properties of $\Mult$. We first start with a helping lemma.
 \begin{lemma}
\label{lemma:MultHelpingLemma}
In protocol $\Mult$, the following hold for every ordered pair $(l, m) \in \{1,
\ldots, q \} \times \{1, \ldots, q\}$.
\begin{myitemize}
\item[--] $\Z_s$ and $\Z_a$ satisfy the $\Q^{(1, 1)}(\Q_{l, m}, \Z_s, \Z_a)$ condition, where $\Q_{l, m} = S_l \cap S_m$.
\item[--] In a synchronous network, at time $\TimeACS$, all honest parties will 
 compute a set $\R_{l, m} \subseteq \Q_{l, m}$, 
 where $\Q_{l, m} \setminus \R_{l, m} \in \Z_s$, such 
 that $\R_{l, m}$ contains all honest parties from $\Q_{l, m}$. 
\item[--] In an asynchronous network, almost-surely, all honest parties will eventually 
 compute a set $\R_{l, m} \subseteq \Q_{l, m}$, 
 where $\Q_{l, m} \setminus \R_{l, m} \in \Z_s$, such 
 that $\R_{l, m}$ contains at least one honest party from $\Q_{l, m}$. 
\end{myitemize}
\end{lemma}
\begin{proof}
The first property follows from the fact that $\Z_s$ and $\Z_a$ satisfy
  the $\Q^{(3, 1)}(\PartySet, \Z_s, \Z_a)$ condition, and $\Q_{l, m} \defined S_l \cap S_m = \Partyset \setminus (Z_l \cup Z_m)$.
  Hence, if $\Z_s$ and $\Z_a$ {\it do not} satisfy
  the $\Q^{(1, 1)}(\PartySet, \Z_s, \Z_a)$ condition, then it implies that there exist sets $Z_{\alpha} \in \Z_s$
   and  $Z_{\beta} \in \Z_a$ such that
   $\Q_{l, m} \subseteq Z_{\alpha} \cup Z_{\beta}$. This further implies that $Z_{\alpha} \cup Z_{\beta} \cup Z_l \cup Z_m \subseteq \Partyset$, 
   which contradicts the fact that
  $\Z_s$ and $\Z_a$ satisfy
  the $\Q^{(3, 1)}(\PartySet, \Z_s, \Z_a)$ condition.
  
  For proving the second property, we consider a {\it synchronous} network. Let $Z_s^{\star} \in \Z_s$ be the set of {\it corrupt} parties.
  Since $\Z_s$ and $\Z_a$ satisfy the $\Q^{(1, 1)}(\Q_{l, m}, \Z_s, \allowbreak \Z_a)$ condition,
  from the $\Z_s$-correctness of $\ACS$ in the {\it synchronous} network, it follows that all honest 
   parties will compute $\R_{l, m}$ as the output of the instance $\ACS(\Q_{l, m})$, such that
   $\Q_{l, m} \setminus \R_{l, m} \in \Z_s$. Moreover, 
    {\it all} the parties from $\Q_{l, m} \setminus Z_s^{\star}$
   will be present in $\R_{l, m}$. This proves the second property.
   
   We next consider an {\it asynchronous} network. Let $Z_a^{\star} \in \Z_a$ be the set of {\it corrupt} parties.
   From the $\Z_a$-correctness of $\ACS$ in the {\it asynchronous} network, it follows that, 
   almost-surely, 
   all honest 
   parties will eventually compute $\R_{l, m}$ as the output of the instance $\ACS(\Q_{l, m})$ such that
   $\Q_{l, m} \setminus \R_{l, m} \in \Z_s$. Moreover, $\R_{l, m} \not \subset Z_a^{\star}$, as otherwise,
   $\Z_s$ and $\Z_a$ {\it do not} satisfy the $\Q^{(1, 1)}(\Q_{l, m}, \Z_s, \Z_a)$ condition, which is a contradiction. 
   Consequently, $\R_{l, m}$ will consists of {\it at least} one honest party from $\Q_{l, m}$. 
\end{proof}
 We now proceed to prove the properties of $\Mult$. 
  \begin{lemma}
  \label{lemma:MultSynchronousProperties}
  In a synchronous network, all honest parties output $[c]$ within time
     $\TimeMult = \TimeACS + 2\Delta$, where $c = a \cdot b$.
    Moreover, the view of the adversary remains independent of $a$ and $b$.
  \end{lemma}
  \begin{proof}
  To prove the lemma, we claim that for each ordered pair $(l, m) \in \{1, \ldots, q \} \allowbreak \times \{1, \ldots, q\}$, all honest
   parties securely compute a secret-sharing of the summand $[a]_l \cdot [b]_m$, within time
   $\TimeACS + 2\Delta$, without revealing any additional information to the adversary.
    The proof then follows from the fact that 
    the following holds:
    \[ [c] = [a \cdot b] = \displaystyle \sum_{(l, m) \in \{1, \ldots, q \} \times \{1, \ldots, q \}} [[a]_l \cdot [b]_m].\]
   We now proceed to prove our claim, for which we consider
   an {\it arbitrary} ordered pair $(l, m) \in \{1, \ldots, q \} \times \{1, \ldots, q\}$.
   
   From Lemma \ref{lemma:MultHelpingLemma}, at time $\TimeACS$, all honest parties will compute the set $\R_{l, m}$
   as the output of the instance $\ACS(\Q_{l, m})$, where $\R_{l, m} \subseteq \Q_{l, m}$, and where
   {\it all honest} parties from $\Q_{l, m}$ will be present in $\R_{l, m}$. Let $|\R_{l, m}| = r$ and $\R_{l, m} = \{P_{\alpha_1}, \ldots, P_{\alpha_r} \}$.
    Moreover, without loss of generality, let $P_{\alpha_1}$ be {\it honest}. 
    From the $\Z_s$-correctness of $\ACS$ in the {\it synchronous} network, corresponding to {\it every} $P_{\alpha_j} \in \R_{l, m}$, there exists
    some value, $v_j$, which will be secret-shared among the parties on the behalf of $P_{\alpha_j}$ during the instance
    $\ACS(\Q_{l, m})$.
    Moreover, since $P_{\alpha_1}$ is assumed to be {\it honest}, from the protocol steps, $v_1 = [a]_l \cdot [b]_m$.
    
    Now there are two possible cases:
  \begin{myitemize}
  \item {\it Every party in $\R_{l, m}$ participates with input
   $[a]_l \cdot [b]_m$ during $\ACS(\Q_{l, m})$}: In this case, $v_1 = v_2 = \ldots = v_r = [a]_l \cdot [b]_m$ and hence
   $d_1 = \ldots = d_{r - 1} = 0$. From the properties of $\Rec$, within time $\TimeACS + \Delta$,
    the honest parties reconstruct the $r - 1$
   differences $d_1, \ldots, d_{r - 1}$
    and find all of them to be $0$. Hence, they set $[[a]_l \cdot [b]_m]$ to $[v_1]$, where
    $v_1$ is the same as $[a]_l \cdot [b]_m$. 
    The privacy in this case follows from privacy of $\ACS$ and the fact that the
   adversary only learns the $r - 1$ differences, which are all $0$.   
  \item {\it Some party in $\R_{l, m}$ participates with an input, which is different from $[a]_l \cdot [b]_m$, during
   $\ACS(\Q_{l, m})$}: Let $P_i \in \R_{l, m}$ be a corrupt party, corresponding to which $v_i$
   is shared during $\ACS(\Q_{l, m})$,
    where $v_i \neq [a]_l \cdot [b]_m$. Since $v_i \neq v_1$, it follows that at least one of the 
    $r- 1$ differences $d_1, \ldots, d_{r - 1}$ will be non-zero,
  and the parties detect the same when these differences are publicly reconstructed at time $\TimeACS + \Delta$.
   In this case, the parties reconstruct the shares
  $[a]_l$ and $[b]_m$ at time $\TimeACS + 2\Delta$,
   and take a default secret-sharing of $[a]_l \cdot [b]_m$.
    The privacy in this case follows from the fact that since there
  exists a {\it corrupt} party in $\R_{l, m}$, the adversary is already aware of the shares 
  $[a]_l$ and $[b]_m$ and hence, publicly reconstructing these values
  does not add any new information to the view of the adversary.
  \end{myitemize}
  \end{proof}
 
 We next consider prove the properties in an {\it asynchronous} network.
   \begin{lemma}
  \label{lemma:MultAsynchronousProperties}
  In an asynchronous network, almost-surely, all honest parties eventually output $[c]$, 
   where $c = a \cdot b$.
    Moreover, the view of the adversary remains independent of $a$ and $b$.
  \end{lemma}
   \begin{proof}
  The proof is very similar to the proof of 
  Lemma \ref{lemma:MultSynchronousProperties}. Namely, we show that 
  almost-surely, corresponding to
   each ordered pair $(l, m) \in \{1, \ldots, q \} \allowbreak \times \{1, \ldots, q\}$, all honest
   parties eventually and securely compute a secret-sharing of the summand $[a]_l \cdot [b]_m$.
   For this, we consider
   an {\it arbitrary} ordered pair $(l, m) \in \{1, \ldots, q \} \times \{1, \ldots, q\}$.
   From Lemma \ref{lemma:MultHelpingLemma}, almost-surely, all honest parties will eventually
   compute a set $\R_{l, m}$
   as the output of the instance $\ACS(\Q_{l, m})$, where $\R_{l, m} \subseteq \Q_{l, m}$, and where {\it at least} one
   {\it honest} party from $\Q_{l, m}$ will be present in $\R_{l, m}$. Let $|\R_{l, m}| = r$ and $\R_{l, m} = \{P_{\alpha_1}, \ldots, P_{\alpha_r} \}$.
    Moreover, without loss of generality, let $P_{\alpha_1}$ be {\it honest}. 
    From the $\Z_a$-correctness of $\ACS$ in the {\it asynchronous} network, 
    corresponding to {\it every} $P_{\alpha_j} \in \R_{l, m}$, there exists
    some value, $v_j$, which will be eventually secret-shared among the parties on the behalf of $P_{\alpha_j}$, during the instance
    $\ACS(\Q_{l, m})$.
    Moreover, since $P_{\alpha_1}$ is assumed to be {\it honest}, from the protocol steps, $v_1 = [a]_l \cdot [b]_m$.
    
     Now there are two possible cases:
  \begin{myitemize}
  \item {\it Every party in $\R_{l, m}$ participates with input
   $[a]_l \cdot [b]_m$ during $\ACS(\Q_{l, m})$}: In this case, $v_1 = v_2 = \ldots = v_r = [a]_l \cdot [b]_m$ and hence
   $d_1 = \ldots = d_{r - 1} = 0$. From the properties of $\Rec$, 
    the honest parties eventually reconstruct the $r - 1$
   differences $d_1, \ldots, d_{r - 1}$
    and find all of them to be $0$. Hence they set $[[a]_l \cdot [b]_m]$ to $[v_1]$, where
    $v_1$ is the same as $[a]_l \cdot [b]_m$. 
    The privacy in this case follows from privacy of $\ACS$, and the fact that the
   adversary only learns $r - 1$ differences which are all $0$.   
  \item {\it Some party in $\R_{l, m}$ participates with an input, which is different from $[a]_l \cdot [b]_m$, during
   $\ACS(\Q_{l, m})$}: Let $P_i \in \R_{l, m}$ be a corrupt party, corresponding to which $v_i$
   is shared during $\ACS(\Q_{l, m})$, 
   where $v_i \neq [a]_l \cdot [b]_m$. Since $v_i \neq v_1$, it follows that at least one of the 
    $r- 1$ differences $d_1, \ldots, d_{r - 1}$ will be non-zero,
  and the parties detect the same when these differences are eventually reconstructed.
   In this case, the parties eventually reconstruct the shares
  $[a]_l$ and $[b]_m$, 
   and take a default secret-sharing of $[a]_l \cdot [b]_m$.
    The privacy in this case follows from the fact that since there
  exists a {\it corrupt} party in $\R_{l, m}$, the adversary is already aware of the shares 
  $[a]_l$ and $[b]_m$ and hence, publicly reconstructing these values
  does not add any new information to the view of the adversary.
  \end{myitemize}
  \end{proof}

  The proof of Lemma \ref{lemma:Mult} now easily follows from Lemma \ref{lemma:MultSynchronousProperties}
   and Lemma \ref{lemma:MultAsynchronousProperties}. The communication complexity follows from the fact that
    $q^2 = |\Z_s|^2$ instances of $\ACS$ are executed.  
  \begin{lemma}
   \label{lemma:Mult}
  Let $\Adv$ be an adversary, characterized by an adversary structure $\Z_s$ in a synchronous network and adversary structure $\Z_a$ in an asynchronous network, satisfying the conditions $\Con$ (see Condition \ref{condition:Con} in Section \ref{sec:prelims}).
   Moreover, let $\SharingSpec = \{S_m : S_m = \PartySet \setminus Z_m \; \mbox{ and } \; Z_m \in \Z_s \}$.
    Then protocol $\Mult$ achieves the following properties, where the inputs of the parties are $[a]$ and $[b]$.
   \begin{myitemize}
   \item[--] {\bf $\Z_s$-correctness}: In a synchronous network, all honest parties output $[c]$ within time
     $\TimeMult = \TimeACS + 2\Delta$, where $c = a \cdot b$.
   \item[--] {\bf $\Z_a$-correctness}: In an asynchronous network, almost-surely, the honest parties eventually output
   $[c]$, where $c = a \cdot b$.
   \item[--] {\bf Privacy}: Irrespective of the network type, the view of the adversary remains independent of $a$ and $b$.
   \item[--] {\bf Communication Complexity}: $\Mult$ incurs a communication of 
   $\Order(|\Z_s|^3 \cdot n^5 (\log{|\K|} + \log{|\Z_s| + \log n})
    + |\Z_s|^2 \cdot n^6 \log n)$ bits and invokes $\Order(|\Z_s|^2 \cdot n)$ instances of $\BA$.
   \end{myitemize}
  \end{lemma}
  We next discuss the modifications needed in the protocol $\Mult$ to handle the case when $L > 1$.
  \paragraph{\bf Protocol $\Mult$ for $L$ Pairs of Inputs:} If the input for $\Mult$ is $\{([a^{(\ell)}], \allowbreak [b^{(\ell)}]) \}_{\ell = 1, \ldots L}$,
 then during the instance of $\ACS(\Q_{l, m})$, each party in $\Q_{l, m}$ will have to share $L$ summands. Similarly, corresponding to
  $\R_{l, m}$, the parties reconstruct $(|\R_{l, m}| - 1) \cdot L$ number of difference values. The rest of the protocol steps remain the same.
  With these modification, $\Mult$ will now incur a
   communication of  $\Order(L \cdot |\Z_s|^3 \cdot n^5 (\log{|\K|} + \log{|\Z_s|} + \log n) 
   + |\Z_s|^2 \cdot n^6 \log n)$ bits and invokes $\Order(|\Z_s|^2 \cdot n)$
    instances of $\BA$.
  \subsection{Protocol $\Offline$ and Its Properties}
  In this section, we formally present our preprocessing phase protocol $\Offline$
   and prove its properties.  For the sake of simplicity, we first explain the protocol
   to generate one random secret-shared multiplication-triple. The protocol is presented in Fig \ref{fig:preprocessing}.
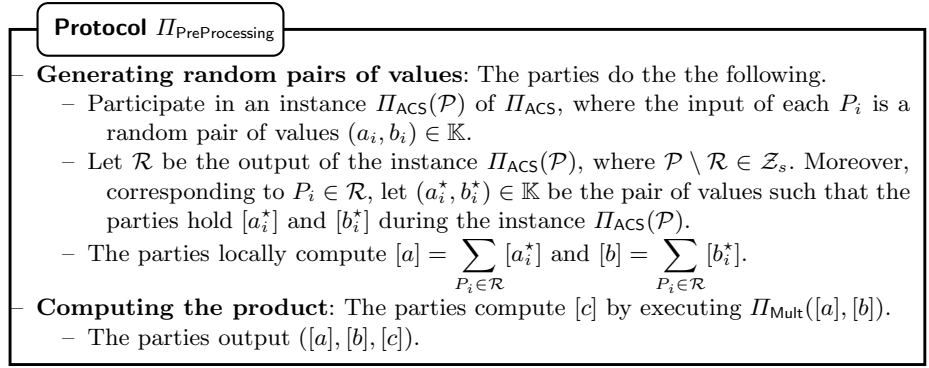
   \begin{protocolsplitbox}{$\Offline$}{The preprocessing phase protocol for generating a secret-sharing of one
  random multiplication-triple}{fig:preprocessing}
	\justify
\begin{mydescription}
\item[--] {\bf Generating random pairs of values}: The parties do the
 the following.
  \begin{myitemize}
  \item[--] Participate in an instance $\ACS(\PartySet)$ of
 $\ACS$, where the input of each $P_i$ is a random pair of values $(a_i, b_i) \in \K$. 
  \item[--] Let $\R$ be the output of the instance $\ACS(\PartySet)$, where $\PartySet \setminus \R \in \Z_s$. Moreover, corresponding
   to $P_i \in \R$, let $(a_i^{\star}, b_i^{\star}) \in \K$ be the pair of values such that the parties hold
   $[a_i^{\star}]$ and $[b_i^{\star}]$ during the instance $\ACS(\PartySet)$.
  \item[--] The parties locally compute $[a] = \displaystyle \sum_{P_i \in \R} [a^{\star}_i]$
   and $[b] = \displaystyle \sum_{P_i \in \R} [b^{\star}_i]$.
  \end{myitemize}
 \item[--] {\bf Computing the product}:  The parties compute
  $[c]$ by executing $\Mult([a], [b])$.
   \begin{myitemize}
    \item[--] The parties output $([a], [b], [c])$.
   \end{myitemize} 
 \end{mydescription}
\end{protocolsplitbox}

We next prove the properties of the protocol $\Offline$. 
   \begin{lemma}
    \label{lemma:Offline}
 Let $\Adv$ be an adversary characterized by an adversary structure $\Z_s$ in a synchronous network and adversary structure $\Z_a$ in an asynchronous network, satisfying the conditions $\Con$ (see Condition \ref{condition:Con} in Section \ref{sec:prelims}).
   Moreover, let $\SharingSpec = \{S_m : S_m = \PartySet \setminus Z_m \; \mbox{ and } \; Z_m \in \Z_s \}$.
    Then $\Offline$ achieves the following.
      \begin{myitemize}
   \item[--] {\bf $\Z_s$-correctness}: In a synchronous network, all honest parties output $([a], [b], \allowbreak [c])$ within time
     $\TimeOffline = \TimeACS + \TimeMult$, where $c = a \cdot b$.
   \item[--] {\bf $\Z_a$-correctness}: In an asynchronous network, almost-surely, the honest parties eventually output
   $([a], [b], [c])$, where $c = a \cdot b$.
   \item[--] {\bf Privacy}: Irrespective of the network type, the view of the adversary remains independent of $a, b$ and $c$.
   \item[--] {\bf Communication Complexity}: 
   $\Offline$ incurs a communication of 
   $\Order(L \cdot |\Z_s|^3 \cdot n^5 (\log{|\K|} + \log{|\Z_s|} + \log n) + |\Z_s|^2 \cdot n^6 \log n)$ bits and invokes $\Order(|\Z_s|^2 \cdot n)$ instances of $\BA$.
   \end{myitemize}
\end{lemma}
\begin{proof}
We first consider a {\it synchronous} network. From the $\Z_s$-correctness of $\ACS$ in the {\it synchronous} network,
 it follows that at time $\TimeACS$, all honest parties will output a common subset of parties $\R \subseteq \PartySet$, where
  $\PartySet \setminus \R \in \Z_s$. Moreover, corresponding to every $P_i \in \R$, there will be a pair of values
  $(a_i^{\star}, b_i^{\star}) \in \K$, which will be secret-shared with respect to $\SharingSpec$ on the behalf of $P_i$, during
  the instance of $\ACS$. Furthermore, $(a_i^{\star}, b_i^{\star}) = (a_i, b_i)$ for every {\it honest}
  $P_i \in \R$. Also, if $P_i \in \R$ is {\it honest}, then the privacy property of $\ACS$ guarantees that the view of the adversary
  remains independent of $(a_i, b_i)$. The $\Z_s$-correctness property of $\ACS$ in the {\it synchronous} network also guarantees that
  {\it all} honest parties from $\PartySet$ will be present in $\R$. Now, since the {\it honest} parties in $\PartySet$ secret-share
  random pairs of values during $\ACS$, it follows that $(a, b)$ will be random from the point of view of the adversary.
  Finally, the $\Z_s$-correctness property of $\Mult$ in the {\it synchronous}
  network guarantees that the parties output $([a], [b], [c])$ at time $\TimeACS + \TimeMult$, where $c = a \cdot b$.
   Moreover, the privacy property of $\Mult$ in the {\it synchronous} network guarantees that adversary does not learn any additional
   information about $a, b$ and $c$ (except that $c = a \cdot b$), during the instance of $\Mult$  as well.
   
   The proof of the properties in an {\it asynchronous} network is almost the same as above, except that we now use the
   $\Z_a$-correctness property of $\ACS$ in the {\it asynchronous} network (which guarantees that the
   parties eventually compute a common $\R$, containing {\it at least} one honest party) and the
   $\Z_a$-correctness property of $\Mult$ in the {\it asynchronous} network. To avoid repetition, we do not give the formal details.
   
   The communication complexity follows from the communication complexity of $\ACS$ and $\Mult$.
\end{proof}
We next discuss the modifications needed in the protocol $\Offline$ to generate $c_M$ number of secret-shared multiplication-triples.
\paragraph{\bf $\Offline$ for Generating $c_M$ Random Multiplication-triples:}
 To generate secret-sharing of $c_M$ random multiplication-triples, the instance of $\ACS$ during the first stage is executed
  by setting $L = c_M$. Consequently, there will be $c_M$ pairs of secret-shared values generated on the behalf of {\it each} party in $\R$.
  Moreover, the pairs of values shared by the {\it honest} parties in $\R$ will be random from the point of view of the adversary. 
  Consequently, summing the secret-sharing of the pairs of values shared by the parties in $\R$, leads to $c_M$
  pairs of random values being secret-shared during the {\it first} stage. Then, during the {\it second} stage, the parties execute an instance of
  $\Mult$, with the input being the $c_M$ pairs of random secret-shared pairs of values from the first stage. This securely leads 
  to  a secret-sharing of the product of each pair.  
  The resultant protocol will now incur a
   communication of  $\Order(c_M \cdot |\Z_s|^3 \cdot n^5 (\log{|\K|} + \log{|\Z_s|} + \log n) + |\Z_s|^2 \cdot n^6 \log n)$ bits, and invokes $\Order(|\Z_s|^2 \cdot n)$
    instances of $\BA$.

\section{The Circuit-Evaluation Protocol and Its Properties}
\label{app:MPC}
Protocol $\PiMPC$ for securely evaluating the circuit $\ckt$ is presented in Fig \ref{fig:MPC}. 
\begin{protocolsplitbox}{$\PiMPC(\ckt, \Z_s, \Z_a)$}{A best-of-both-worlds perfectly-secure protocol for securely evaluating the arithmetic 
 circuit $\ckt$}{fig:MPC}
  \begin{myitemize}
  \item[--] {\bf Preprocessing and Input-Sharing} --- The parties do the following:
    \begin{myitemize}
     \item[--] Each $P_i \in \Partyset$, on having the input $x_i$ for $f$, 
     participates in the instance
    $\ACS(\PartySet)$ of $\ACS$ with input  $x_i$. 
      Let $\CoreSet$ be the common subset of parties obtained as the output during $\ACS(\PartySet)$, 
      where $\PartySet \setminus \CoreSet \in \Z_s$. Corresponding to every
     $P_j \not \in \CoreSet$, set $x_j = 0$, and set $[x_j]$ to the default
     secret-sharing, where $[x_j]_1 = [x_j]_2 = \ldots = [x_j]_q = 0$.     
     \item[--] Participate in an instance of $\Offline$
     to generate $c_M$ number of secret-shared, random multiplication-triples
      $\{[\mathbf{a}^{(j)}], [\mathbf{b}^{(j)}], [\mathbf{c}^{(j)}] \}_{j = 1, \ldots, c_M}$.
     \end{myitemize}  
   \item[--] {\bf Circuit Evaluation} --- Let $G_1, \ldots, G_{m}$ be a publicly-known topological ordering
       of the gates of $\ckt$. For $k = 1, \ldots, m$, the parties do the following for gate $G_k$:
         \begin{myitemize}
	   \item[--] {\it If $G_k$ is an addition gate}: the parties locally compute $[w] = [u] + [v]$, 
	   where $u$ and $v$ are gate-inputs, and $w$ is the gate-output.
	   \item[--] {\it If $G_k$ is a multiplication-with-a-constant gate with constant $c$}: the parties locally compute
	    $[v] = c \cdot [u]$, where $u$ is the gate-input, and $v$ is the gate-output.
	   \item[--] {\it If $G_k$ is an addition-with-a-constant gate with constant $c$}: the parties locally compute 
	   $[v] = c +  [u]$, where $u$ is the gate-input, and $v$ is the gate-output.
	   \item[--] {\it If $G_k$ is a multiplication gate}: Let $G_k$ be the $\ell^{th}$ multiplication gate in $\ckt$, where $\ell \in \{1, \ldots, c_M \}$,
	    and let $([a^{(\ell)}], [b^{(\ell)}], [c^{(\ell)}])$ be the $\ell^{th}$ shared multiplication-triple generated during $\Offline$.
	   Moreover,  let $[u]$ and $[v]$ be the shared gate-inputs of $G_k$. 
            Then, the parties participate in an instance $\BatchBeaver(([u], [v]), ([a^{(\ell)}], [b^{(\ell)}], [c^{(\ell)}]))$ of $\BatchBeaver$, 
            and obtain $[w]$, where $w = u \cdot v$.	   
     \end{myitemize}
    \item[--] {\bf Output Computation} --- Let $[y]$ be the secret-shared circuit-output. 
    The parties participate in an instance $\Rec(y, \SharingSpec)$ of $\Rec$ and reconstruct $y$.    
      \item[--] {\bf Termination}: Each $P_i$ does the following.
          \begin{myitemize}
          \item[--] If $y$ has been obtained during output computation, then send the message $(\ready, y)$ to all the parties.
          \item[--] If the message $(\ready, y)$ is received from a set of parties ${\cal C}$, where ${\cal C} \not \in \Z_s$, then 
          send $(\ready, y)$ message to all the parties if no $\ready$ message is sent earlier.
           \item[--] If the message $(\ready, y)$ is received from a set of parties $\PartySet \setminus Z$, for some $Z \in \Z_s$, 
           then terminate all sub-protocols, output $y$,
           and terminate.           
          \end{myitemize}   
   \end{myitemize}
\end{protocolsplitbox}

\noindent {\bf Theorem \ref{thm:MPC}.}
{\it Let $\Adv$ be an adversary characterized by an adversary structure $\Z_s$ in a synchronous network and adversary structure $\Z_a$ in an asynchronous network satisfying the conditions $\Con$ (see Condition \ref{condition:Con} in Section \ref{sec:prelims}). 
 Moreover, let $f: \K^n \rightarrow \K$ be a function represented by an arithmetic circuit $\ckt$ over $\K$, consisting of
 $c_M$ number of multiplication gates, with a multiplicative depth of $D_M$ and where each party $P_i$ has an input $x_i \in \K$
  for $f$.
 Furthermore, let $\SharingSpec = \{S_m : S_m = \PartySet \setminus Z_m \; \mbox{ and } \; Z_m \in \Z_s \}$.
 Then, $\PiMPC$ incurs a communication of 
 $\Order(c_M \cdot |\Z_s|^3 \cdot n^5 (\log{|\K|} + \log{|\Z_s|} + \log n) + |\Z_s|^2 \cdot n^6 \log n)$ bits, invokes $\Order(|\Z_s|^2 \cdot n)$ instances of $\BA$, 
   and achieves the following. 
      \begin{myitemize}
      \item[--] In a synchronous network, all honest parties output $y = f(x_1, \ldots, x_n)$ at time $(30n + D_M + 6k + 38) \cdot \Delta$,
       where $x_j = 0$ for every $P_j \not \in \CoreSet$, such that
       $\PartySet \setminus \CoreSet \in \Z_s$
       and every honest party from $\PartySet$ will be present in $\CoreSet$; here $k$ is the constant from Lemma \ref{lemma:ABAGuarantees}, as determined by the underlying
       (existing) perfectly-secure ABA protocol $\ABA$.
       \item[--] In an asynchronous network, almost-surely, the honest parties eventually output $y = f(x_1, \ldots, x_n)$, 
       where $x_j = 0$ for every $P_j \not \in \CoreSet$,
       and where $\PartySet \setminus \CoreSet \in \Z_s$.
       \item[--] The view of the adversary remains independent of the inputs of the honest parties in $\CoreSet$.                
   \end{myitemize}
}
\begin{proof}
Consider a {\it synchronous} network. Let $Z_s^{\star} \in \Z_s$ be the set of 
 {\it corrupt} parties, and let $\Honest_s \defined \PartySet \setminus Z_s^{\star}$ be the set of {\it honest} parties.
   From the $\Z_s$-correctness property of $\Offline$ in the {\it synchronous} network,
  at time $\TimeOffline$, the (honest)
 parties will have $c_M$ number of
  secret-shared multiplication-triples, shared with respect to $\SharingSpec$, from the instance of $\Offline$. 
  From the $\Z_s$-correctness property of $\ACS$ in the {\it synchronous} network, at time $\TimeACS$, the (honest)
  parties will have a common subset $\CoreSet$ from the instance of $\ACS$, where {\it all} {\it honest} parties will be present in $\CoreSet$,
   and  where $\PartySet \setminus \CoreSet \in \Z_s$.
  Moreover, corresponding to {\it every} $P_j \in \CoreSet$, there will be some $x_j \in \K$ held by $P_j$  
  (which will be the same as $P_j$'s input for $f$ for an {\it honest} $P_j$), 
  such that $x_j$ will be
  secret-shared with respect to $\SharingSpec$.
   As $\CoreSet$ will be known {\it publicly}, the parties take a default secret-sharing of $0$ on the behalf of the parties 
   $P_j$ outside $\CoreSet$, by considering $x^{(j)} = 0$.
  Since $\TimeACS < \TimeOffline$, it follows that at time $\TimeOffline$, the parties will hold a secret-sharing of $c_M$ multiplication-triples
  and secret-sharing of $x_1, \ldots, x_n$. 
  
  The circuit-evaluation will take $D_M \cdot \Delta$ time. This follows from the fact that linear gates are
   evaluated locally (non-interactively), while all the {\it independent} multiplication gates can be evaluated in parallel
  by running the corresponding instances of $\BatchBeaver$ in {\it parallel}, where each such instance requires $\Delta$ time.
   From the $\Z_s$-correctness property of $\BatchBeaver$
   in the {\it synchronous} network, the
  multiplication-gates will be evaluated correctly and hence, during the output-computation phase, the parties will hold a 
  secret-sharing of $y$ (with respect to $\SharingSpec$), where 
  $y = f(x_1, \ldots, x_n)$.
   From the properties of $\Rec$, it will take $\Delta$ time for every party to reconstruct $y$. Hence, during the termination phase,
  {\it all} honest parties will send a $\ready$ message for $y$. Since $\PartySet \setminus \Honest_s \in \Z_s$, 
   every honest party will then terminate with output $y$ at time
  $\TimeOffline + (D_M + 2) \cdot \Delta$. By substituting the values of 
   $\TimeOffline, \TimeVSS, \TimeACS, \TimeBCAST, \TimeBA, \TimeSBA$   and
  $\TimeABA$ and by noting that all instances of $\BCAST$ in $\PiMPC$ are invoked with $\Z = \Z_s$, we get that
  the parties terminate the protocol at time $(D_M + 30n + 6k + 38) \cdot \Delta$, where $k$ is the 
  constant from Lemma \ref{lemma:ABAGuarantees}, as determined by the underlying
       (existing) perfectly-secure ABA protocol $\ABA$.

     If we consider an {\it asynchronous} network, then the proof is similar as above, except that we 
     now use the security properties of the building blocks
  $\Offline, \allowbreak \ACS, \BatchBeaver$ and $\Rec$ in the {\it asynchronous} network. Let $Z_a^{\star} \in \Z_a$ be the set of {\it corrupt}
  parties, and let $\Honest_a \defined \PartySet \setminus Z_a^{\star}$ be the set of {\it honest} parties. 
  During the termination phase, 
  the parties in $Z_a^{\star}$ may send $\ready$ messages for $y'$, where $y' \neq y$. Since every subset in $\Z_a$ is a subset of some
   subset in $\Z_s$, it follows that no honest party will terminate with output $y'$, where $y' \neq y$.
   On the other hand, all the parties in $\Honest_a$ will eventually compute the output $y$, and will send a $\ready$ message
   for $y$, which is eventually delivered to every honest party. 
    Now, consider an {\it honest} party $P_h$, who  terminates with output $y$. We wish to show that every honest party eventually
  terminates the protocol with output $y$. 
  This is because $P_h$ must have received $\ready$ messages for $y$ from a subset of parties $\PartySet \setminus Z_{\alpha}$, for
  some $Z_{\alpha}$. Since $\Z_s$ satisfies the $\Q^{(3)}(\PartySet, \Z_s)$ condition, it follows that that 
  $\Honest_a \cap (\PartySet \setminus Z_{\alpha}) \not \in \Z_s$. Now the $\ready$ messages for $y$ from the set of
  parties  $\Honest_a \cap (\PartySet \setminus Z_{\alpha}) $ 
  are eventually delivered to {\it every} 
  honest party. Consequently, irrespective of which stage of the protocol an honest party is in, every party in $\Honest_a$
   (including $P_h$) eventually
   sends a $\ready$ message for $y$, which is eventually delivered. Since $\PartySet \setminus \Honest_a \in \Z_s$
   (as every subset in $\Z_a$ is a subset of some subset from $\Z_s$), 
    this implies that every honest party eventually terminates with output $y$.
  
  From the privacy property of $\ACS$, corresponding to every {\it honest} $P_j \in \CoreSet$,
   the input $x_j$ will be random from the point of view
  of the adversary. Moreover, from the 
  privacy property of $\Offline$, the multiplication-triples generated through $\Offline$ will be random from the point of view of the
  adversary. During the evaluation of linear gates, no interaction happens among the parties and hence, 
  no additional information about the inputs of the honest parties is revealed.
  The same is true during the evaluation of multiplication-gates as well, which follows from the privacy property of $\BatchBeaver$. 

  The communication complexity of the protocol follows from the communication complexity of $\Offline, \ACS$ and $\BatchBeaver$.

\end{proof}

\end{document}